\newif\ifarxiv
\arxivtrue                % arXiv version

\ifarxiv
  \documentclass[acmsmall,screen,nonacm]{acmart}
\else
  \documentclass[manuscript,screen,review]{acmart}
\fi

\setcopyright{none}

\ifarxiv
\fi

\PassOptionsToPackage{svgnames}{xcolor}
\acmJournal{TOCL}
\acmDOI{}
\acmArticle{}

\renewcommand{\footnotetextcopyrightpermission}[1]{}

\usepackage{enumerate}
\usepackage{enumitem}
\usepackage[all,cmtip,2cell]{xy}

\theoremstyle{acmplain}
\newtheorem{thm}{Theorem}
\newtheorem{prop}[thm]{Proposition}
\newtheorem{defi}[thm]{Definition}
\newtheorem{lemma}[thm]{Lemma}
\newtheorem{cor}[thm]{Corollary}
\theoremstyle{acmdefinition}
\newtheorem{rem}[thm]{Remark}
\newtheorem{ex}[thm]{Example}

\newcommand{\timeout}{\textsc{t/o}}
\DeclareMathOperator{\Aut}{Aut}

\DeclareMathOperator{\lit}{lit}

\DeclareMathOperator{\syn}{syn}
\DeclareMathOperator{\id}{id}
\DeclareMathOperator{\Mod}{Mod}
\DeclareMathOperator{\Stab}{Stab}

\DeclareMathOperator{\sem}{sem}

\DeclareMathOperator{\dep}{dep}

\DeclareMathOperator{\loc}{loc}

\DeclareMathOperator{\PHP}{\mathsf{PHP}}

\DeclareMathOperator{\DPHP}{\mathsf{DPHP}}
\DeclareMathOperator{\Sym}{\mathsf{Sym}}
\DeclareMathOperator{\XG}{\mathsf{XG}}
\DeclareMathOperator{\DXG}{\mathsf{DXG}}

\DeclareMathOperator{\NEXP}{\mathsf{NEXP}}

\DeclareMathOperator{\lex}{lex}

\DeclareMathOperator{\Orb}{Orb}
\DeclareMathOperator{\can}{can}

\DeclareMathOperator{\sort}{sort}

\DeclareMathOperator{\sync}{sync}
\newcommand{\BF}{\mathrm{BF}}        
\newcommand{\B}{\mathcal{B}} 
\newcommand{\A}{\mathbb{A}} 
\newcommand{\supp}{\operatorname{supp}}
\newcommand{\Step}{\operatorname{Step}}
\newcommand{\eval}[2]{[#1]_{#2}} 
\newcommand{\Var}{\operatorname{Var}}
\newcommand{\Lit}{\operatorname{Lit}}
\newcommand{\Eq}{\operatorname{Eq}}
\newcommand{\OAnd}{\operatorname{And}}
\newcommand{\Dep}{D}
\newcommand{\Sk}{\mathbb{S}}     
\newcommand{\Gsyn}{G_{\mathrm{syn}}}
\newcommand{\Glit}{G_{\mathrm{lit}}}
\usepackage{siunitx}
\usepackage{forest}
\usepackage{tikz}
\usetikzlibrary{backgrounds,fit}
\usetikzlibrary{arrows.meta, positioning, calc}
\usetikzlibrary{decorations.pathreplacing}
\usepackage{xspace}
\usetikzlibrary{positioning, decorations.pathmorphing,arrows,chains}
\usepackage{subcaption}
\usepackage{booktabs}
\usepackage{tabularx}
\usepackage{array}

\definecolor{amber}{rgb}{1.0, 0.75, 0.0}

\title{Symmetries of Dependency Quantified Boolean Formulas}

\author{Paul Seip}
\affiliation{
  \institution{Johannes Kepler University Linz}
  \city{Linz}
  \country{Austria}
}
\orcid{0009-0000-3825-0329}
\email{paulus.seip@jku.at}
\author{Clemens Hofstadler}
\affiliation{
  \institution{Johannes Kepler University Linz}
  \city{Linz}
  \country{Austria}
}
\orcid{0000-0002-3025-0604}
\email{clemens.hofstadler@jku.at}

\author{Manuel Kauers}
\affiliation{
  \institution{Johannes Kepler University Linz}
  \city{Linz}
  \country{Austria}
}
\orcid{0000-0001-8641-6661}
\email{manuel.kauers@jku.at}

\author{Martina Seidl}
\affiliation{
  \institution{Johannes Kepler University Linz}
  \city{Linz}
  \country{Austria}
}
\orcid{0000-0002-3267-4494}
\email{martina.seidl@jku.at}
\renewcommand{\shortauthors}{Seip et al.}

\begin{document}

\begin{abstract}
 Symmetry breaking is a standard technique for reducing redundancy in automated reasoning, with well-developed theories for SAT and QBF. In this paper we study symmetries of dependency quantified Boolean formulas (DQBF), a generalisation of QBF which allows for explicit quantifier dependencies. The possibility of incomparable dependency sets makes extending existing QBF techniques non-trivial. We develop a theoretical framework for syntactic and semantic symmetries of DQBF that accounts for these dependency constraints. Building on this framework we establish a general lex-leader symmetry-breaking construction that strictly generalises the corresponding construction for QBF, using a lex-permissibility condition to ensure that local transformations respect the dependencies. We complement this construction with a canonisation method that exploits problem-specific symmetries without requiring lex-permissibility. Proof-of-concept experiments demonstrate the benefits of lex-leader symmetry breaking and the substantial additional gains provided by canonisation predicates for the succinct pigeonhole principle, although the impact of symmetry breaking is solver-dependent.
\end{abstract}



\maketitle

%\tableofcontents

\section{Introduction} \label{sect:intro}
\emph{Symmetries} play an important role in automated reasoning, both in solvers as well as in proof systems. On the one hand, symmetries can enhance proof systems by introducing new deduction rules that exploit symmetries, akin to the ``without loss of generality'' method frequently used by mathematicians to deal with similar cases. Exploiting such symmetries results in exponentially more powerful proof systems, both in SAT~\cite{Kri85, Urq99} and in QBF~\cite{KS18a, SSS20}. A canonical symmetry-rich benchmark family in propositional logic is the \emph{pigeonhole principle} \cite{CR79}. The (false) propositional formula $\PHP_n$ expresses that $n+1$ pigeons can be placed into $n$ holes while no two pigeons share a hole. It has a central role in the propositional proof complexity landscape, as one of the first exponential lower bounds for resolution refutation proofs~\cite{Hak85}. However, already Krishnamurthy~\cite{Kri85} observed that augmenting resolution with a symmetry rule yields polynomial-length refutations by exploiting the interchangeability of pigeons and holes; this phenomenon was subsequently studied systematically by Urquhart~\cite{Urq99}. Since then, symmetry-based reasoning has been generalised in several directions, including local symmetry rules~\cite{AraiU00}, homomorphism and dynamic rules~\cite{Szeider05}, and extensions to QBF proof systems~\cite{KS18a,SSS20}. Symmetry has also been incorporated into modern clausal and pseudo-Boolean proof systems~\cite{BGMN23, HHW15}.

On the other hand, symmetries can be exploited to enhance solver performance. If one is able to detect symmetries in the search space one is dealing with, then collapsing these symmetries can significantly speed up the solving time  since the solver only has to focus on non-symmetric potential solutions. These \emph{symmetry breaking} techniques have already been extensively used in SAT~\cite{ASM06, DBDDM12, DBBD16, Sak21}, CSP~\cite{CJJPS06, symmetry-SAT, GPP06}, and QBF~\cite{AJS07b,AJS07a,AMS04} to improve solver performance. Recent progress in symmetry breaking for SAT solving has resulted in the development of the tool \textsc{satsuma}~\cite{ABR26,ACH26}, which has demonstrated success on a wide variety of benchmarks in SAT~\cite{FHI26}. 
Symmetries have also been used in QBF model enumeration~\cite{PHHS25}.%: since symmetries map models to models, it suffices to enumerate one model per symmetry class and to obtain the remaining ones by applying the symmetries.

Apart from serving as a general preprocessing technique for solvers, symmetry breaking has also been important in solving difficult combinatorial problems and conjectures in mathematics. Prominent examples include: the proof that the packing chromatic number of the infinite square grid is 15 \cite{SH23}, Keller's conjecture~\cite{BHMN22}, and Schur Number Five~\cite{Heu18}.

\emph{Dependency Quantified Boolean formulas (DQBF)}~\cite{Bub09, PRA01} represent a rich and expressive class of logical formulas that extends QBF by allowing existentially quantified variables to depend on specific subsets of universally quantified variables. Every QBF induces an equivalent DQBF with a linearly ordered sequence of dependency sets. In contrast to QBFs, the nuanced quantification of DQBFs allows also for non\nobreakdash-linear dependencies. This makes DQBF a powerful framework for encoding a variety of problems in verification, synthesis, and soft-/hardware engineering, see~\cite{CHLT22, SW18} and references therein. The extended expressive power, however, comes at the cost of increased computational complexity: the decision problem for DQBF is \textsf{NEXPTIME}-complete~\cite{PRA01}. Chen et al. \cite{CHLT22} provide a method for reducing general \textsf{NEXP}-complete problems to DQBF, thereby demonstrating that DQBF can be seen as a canonical language to encode \textsf{NEXP}-complete problems; similar to the role that SAT plays for \textsf{NP}-complete problems. Examples include \emph{succinctly represented \textsf{NP}-complete problems}, which represent an exponentially large instance implicitly, typically by a polynomial-size Boolean circuit. For many standard \textsf{NP}-complete problems, this raises the complexity of the succinct version to \textsf{NEXP}-completeness. As a result, succinct combinatorial problems such as 3-colourability, Hamiltonian cycle, set packing, and subset sum have a natural encoding in DQBF. These provide a rich source of examples with plenty of symmetry that could potentially be exploited.

\paragraph{Motivation and outline.} In~\cite{KS18} two of the authors give a formal treatment of symmetries in QBF. The main technical difficulty of our work is lifting these techniques to the more general class of DQBF, which is non-trivial due to the possible non-linear dependency structure of DQBFs. Indeed, in Example~\ref{ex:nec_condition} we present an example that shows that a naive lifting of the lex-leader encoding constructed in~\cite{KS18} to DQBF may destroy the truth-value of the original formula.

In this paper, we develop a rigorous group-theoretical framework for studying and exploiting symmetries in DQBF that accounts for these dependency constraints. We address the aforementioned obstruction by characterising when interpretations can be glued consistently and deriving \emph{lex-permissibility} conditions that ensure sound local lex-leader symmetry breaking. We moreover complement the general lex-leader construction with a canoniser-based method for constructing symmetry-breaking predicates, which can be used to exploit problem-specific symmetries that are not captured by the lex-leader construction. We demonstrate our methods on different DQBF benchmark families, showing that they can substantially improve the performance of state-of-the-art DQBF solvers.
%Example~\ref{ex:nec_condition} shows that this condition cannot simply be dropped, while Section~\ref{sec:tightness} shows that it is not necessary either, and leaves a syntactic characterisation open. We also provide proof-of-concept experiments to demonstrate its effect on solving. 
This work is structured as follows. 

\begin{enumerate}[label=(\roman*)]
    \item In Section \ref{sect:prelim} we treat the preliminaries. 
    \item In Section \ref{sec:symm} we develop a group-theoretical framework that provides a theoretical foundation for studying symmetries in DQBF. This includes a definition of syntactic symmetry, semantic symmetry, and the interaction between the two. Following~\cite{KS18}, we consider \emph{generalised syntactic symmetries}, which allow variables to be mapped to Boolean formulas, subject to admissibility conditions. This generality captures natural DQBF symmetries beyond literal permutations, including arbitrary permutations of pigeons and holes in the succinct pigeonhole principle.
    \item In Section~\ref{sect:lex_leader} we build on this framework to prove the correctness of a general lex-leader symmetry-breaking construction for a given DQBF and its syntactic symmetry group (Theorem \ref{thm:symm_breaker}), and discuss its tightness. Instantiating our theorem for QBFs yields a strict generalisation of \cite[Theorem~22]{KS18}, a corresponding lex-leader symmetry-breaking construction for QBF. 
    \item In Section~\ref{sec:canonisers}, we develop a canonisation framework for function encodings and demonstrate this by deriving a compact canonicality predicate for the succinct pigeonhole principle. These canonisation predicates enable problem-specific symmetry breaking beyond the scope of the general lex-leader construction.
    \item In Section~\ref{sec:experiments}, we report proof-of-concept experiments using manually derived symmetry-breaking predicates for different DQBF benchmark families, namely two false variants and a true variant of a hidden-vector XOR game, and the succinct pigeonhole principle. The latter illustrates the additional benefits of specialised canonicality predicates exploiting generalised syntactic symmetries.
\end{enumerate}

The contributions of this paper are primarily theoretical. Our experiments provide a proof of concept illustrating the practical potential of the proposed methods. A fully automated implementation and a comprehensive experimental evaluation are left to future work.

\section{Preliminaries} \label{sect:prelim}
In this section, we treat the preliminaries on which our framework is built. Throughout the paper, we denote $[k]$ for the set $\{1,\ldots,k\}$. 

\subsection{Groups and Graphs} \label{prelim:groups}
\paragraph{Groups and actions.} A \emph{group} $\mathbb{G}$ is a triple $\mathbb{G} = (G,\cdot, e)$ containing a \emph{carrier set} $G$, an associative binary operation $\cdot \colon G \times G \to G$, and a designated \emph{neutral element} $e \in G$ with the property that $g \cdot e = e \cdot g = g$ for all $g \in G$. Moreover, for every $g \in G$ there exists an \emph{inverse} $g^{-1} \in G$ such that $g \cdot g^{-1} = g^{-1} \cdot g = e$. We often denote the group $\mathbb{G}$ by its carrier set $G$. 
A subgroup $\mathbb{H}$ of $\mathbb{G}$ is a non-empty subset $H \subseteq G$ which is closed under the group operation and under inverses. We denote $\mathbb{H} \leq \mathbb{G}$. A permutation of a set $V$ is a bijective map from $V$ to itself. Two permutations can be composed to yield another permutation. Moreover, bijections are invertible, and the identity $\id$ is a bijection of $V$. Thus the set of all permutations of $V$ forms a group $\Sym(V)$ with composition as operation, it is called the \emph{symmetric group} on $V$. In particular, if $V = \{1,\ldots,n\}$ we denote this by $S_n$. A \emph{homomorphism} between two groups $\mathbb{G}$ and $\mathbb{G}'$ is a map $\alpha : G \to G'$ that preserves the identity and the group operation, i.e. $\alpha(e) = e'$, and $\alpha(g \cdot h) = \alpha(g) \cdot \alpha(h)$. A (left) \emph{group action} of a group $\mathbb{G}$ on a set $X$ is a map 
\[
\ast \colon G \times X \to X, 
\]
which preserves the identity element and the group operation: $e \ast x = x$ and $(g \cdot h) \ast x = g \ast (h \ast x)$ for all $x \in X$. We denote $G \curvearrowright X$. An example of a group action is the symmetric group $\Sym(V)$ acting on the set $V$ via $\Sym(V) \times V \to V : (g,v) \mapsto g(v)$. A right group action is the dual concept: $\ast \colon X \times G \to X$ satisfying analogous conditions. Given a group action, we can define a relation on $X$ as $x \sim y$ if there exists an element $g \in G$ with $g \ast x = y$. Since this is an equivalence relation, it partitions the set $X$ into equivalence classes called \emph{orbits}. We denote the orbit of $x$ as $\Orb_{\mathbb{G}}(x) = \{y \in X \ | \ x \sim y\}$.  Any subgroup $\mathbb{G}' \leq \mathbb{G}$ also acts on $X$, and induces a finer orbit partition. The \emph{stabiliser subgroup} of an element $x \in X$ is given by $\Stab_{G}(x) = \{g \in G \ | \ g \ast x = x\}$. A subset $K \subseteq X$ is \emph{$G$-invariant} if $g \ast x \in K$ for every $g \in G$ and $x \in K$.

\paragraph{Graphs.} A \emph{graph} is a pair $\mathcal{G} = (V,E)$ which consists of a set of \emph{vertices} $V$ and a set of \emph{edges} $E \subseteq V \times V$. A graph automorphism of $\mathcal{G}$ is a bijective map $f : V \to V$ which preserves adjacency: $(u,v) \in E$ if and only if $(f(u),f(v)) \in E$ for all $u,v \in V$. The collection of all automorphisms of $\mathcal{G}$ forms a group $\Aut(\mathcal{G})$ which is called its \emph{automorphism group}. The concept of a graph automorphism directly generalises to graphs with additional structure, such as directed graphs or coloured graphs. 

\subsection{Boolean Algebras} \label{sec:BA}
A \emph{Boolean algebra} is an algebra $\mathbb{B} = (B,\wedge,\vee,\neg,0,1)$ such that $(B,\wedge,\vee,0,1)$ is a bounded distributive lattice and $b \wedge \neg b = 0$ and $b \vee \neg b = 1$ for every $b \in B$. A \emph{homomorphism} between two Boolean algebras $\mathbb{B}$ and $\mathbb{B}'$ is a map $B \to B'$ preserving $\wedge,\vee,\neg,0,1$; an \emph{endomorphism} is a homomorphism from a Boolean algebra to itself, and an \emph{automorphism} is a bijective endomorphism. We write $\Aut(\mathbb{B})$ for the group of automorphisms of~$\mathbb{B}$.

For a finite set of variables $V$, let $\BF(V)$ denote the set of Boolean formulas generated by $V$, the constants $\bot,\top$, and the connectives $\neg,\wedge,\vee,\to,\leftrightarrow,\oplus$. An \emph{assignment} is a map $\sigma \colon V \to \{0,1\}$; the set of all assignments is denoted
by $\A(V)$. Each assignment extends uniquely to an evaluation map $\eval{\,\cdot\,}{\sigma}\colon \BF(V) \to \{0,1\}$ using the usual Boolean operations. We write $\varphi[v := \psi]$ for the formula obtained by replacing every occurrence of $v$ in $\varphi$ by $\psi$. Formulas $\varphi,\psi \in \BF(V)$ are \emph{propositionally equivalent}, written $\varphi \equiv \psi$, if $\eval{\varphi}{\sigma} = \eval{\psi}{\sigma}$ for every $\sigma \in \A(V)$. Propositional equivalence is a congruence with respect to the Boolean connectives, so the quotient
\[
\B(V) := \BF(V)/{\equiv}
\]
is a Boolean algebra, with operations induced by the connectives and constants $0 = [\bot]$, $1 = [\top]$. It is the free Boolean algebra generated by $\{[v] \mid v \in V\}$. That is, every map $\alpha\colon V \to C$ into a Boolean algebra with carrier~$C$ extends uniquely to a homomorphism $\B(V) \to C$. The elements of $\B(V)$ correspond exactly to the Boolean functions $\mathbb A(V)\to\{0,1\}$. Since there are $2^{|V|}$ assignments, $\B(V)$ is finite, with $2^{2^{|V|}}$ elements, and hence $\Aut(\B(V))$ is a finite group.

\begin{rem}[Convention]\label{rem:convention}
From now on we identify propositionally equivalent formulas: we work in $\B(V)$ and write $\varphi$ for the equivalence class of a formula $\varphi \in \BF(V)$. Accordingly, \emph{equality} of formulas means propositional equivalence, $\varphi = \psi$ iff $\varphi \equiv \psi$ in $\BF(V)$. Evaluation $\eval{\varphi}{\sigma}$ is well-defined on $\B(V)$, and we use the same notation for it. When a concrete formula is needed, for instance as input to a solver, any representative of the class may be chosen.
\end{rem}

By freeness, an endomorphism $g$ of $\B(V)$ is determined by its values $g(v)$, $v \in V$, and acts on an arbitrary formula by substitution: $g(\varphi)$ is obtained from $\varphi$ by replacing every variable $v$ with $g(v)$. Conversely, every choice of formulas $g(v) \in \B(V)$ for $v \in V$ determines an endomorphism. Thus endomorphisms of $\B(V)$ are precisely substitutions, considered up to propositional equivalence. Every endomorphism $g$ induces a map on assignments,
\[
\hat g\colon \A(V) \to \A(V), \qquad
\hat g(\sigma)(v) := \eval{g(v)}{\sigma},
\]
which first applies $g$ to each variable and then evaluates the result. The following identity is the only property of endomorphisms that we use in the sequel.

\begin{lemma}[Evaluation identity]\label{lem:evaluation}
Let $g$ be an endomorphism of $\B(V)$. Then for every $\varphi \in \B(V)$ and every $\sigma \in \A(V)$,
\[
\eval{g(\varphi)}{\sigma} = \eval{\varphi}{\hat g(\sigma)}.
\]
\end{lemma}

\begin{proof}
By induction on a representative of $\varphi$. If $\varphi$ is a variable $v$, then $\eval{g(v)}{\sigma} = \hat g(\sigma)(v) = \eval{v}{\hat g(\sigma)}$ by definition of $\hat g$. If $\varphi$ is a constant, both sides agree since $g$ preserves constants. For a connective, say $\varphi = \varphi_1 \wedge \varphi_2$, we have $g(\varphi) = g(\varphi_1) \wedge g(\varphi_2)$ since $g$ is a homomorphism, and therefore
\[
  \eval{g(\varphi)}{\sigma}
  = \eval{g(\varphi_1)}{\sigma} \wedge \eval{g(\varphi_2)}{\sigma}
  = \eval{\varphi_1}{\hat g(\sigma)} \wedge \eval{\varphi_2}{\hat g(\sigma)}
  = \eval{\varphi}{\hat g(\sigma)}
\]
by the induction hypothesis. The remaining connectives are analogous.
\end{proof}

\begin{cor}\label{cor:hat}
Let $f,g$ be endomorphisms of $\B(V)$.
\begin{enumerate}[label=(\roman*)]
\item $\widehat{g \circ f} = \hat f \circ \hat g$.
\item $g$ is an automorphism if and only if $\hat g$ is a bijection. In
  this case $\widehat{g^{-1}} = \hat g^{-1}$.
\end{enumerate}
\end{cor}
\begin{proof}
(i) For $v \in V$ and $\sigma \in \A(V)$, Lemma~\ref{lem:evaluation} gives
\[
  \widehat{g \circ f}(\sigma)(v)
  = \eval{g(f(v))}{\sigma}
  = \eval{f(v)}{\hat g(\sigma)}
  = \hat f(\hat g(\sigma))(v).
\]
(ii) Suppose $g$ is an automorphism with inverse $h$. By (i), $\hat h \circ \hat g = \widehat{g \circ h} = \widehat{\mathrm{id}} = \mathrm{id}$ and likewise $\hat g \circ \hat h = \mathrm{id}$. Hence $\hat g$ is a bijection with inverse $\hat h = \widehat{g^{-1}}$.

Conversely, suppose $\hat g$ is a bijection. If $g(\varphi) = g(\psi)$, then Lemma~\ref{lem:evaluation} yields $\eval{\varphi}{\hat g(\sigma)} = \eval{\psi}{\hat g(\sigma)}$ for every $\sigma \in \A(V)$. Since $\hat g$ is surjective, $\varphi$ and $\psi$ agree under every assignment, so $\varphi = \psi$. Thus $g$ is injective, and since $\B(V)$ is finite, $g$ is bijective. The inverse of a bijective homomorphism is again a homomorphism, so $g$ is an automorphism.
\end{proof}

Part (i) says that the assignment $g \mapsto \hat g$ is contravariant: composition reverses. Part (ii) is a convenient criterion for constructing automorphisms: it suffices to prescribe $g(v)$ for $v \in V$ and to check that the induced map on assignments is a bijection.

We will also need to speak about the variables a formula actually depends on. 

\begin{defi}[Support]\label{def:support}
Let $\varphi \in \B(V)$. A variable $v \in V$ is \emph{essential} for $\varphi$ if $\varphi[v := \bot] \neq \varphi[v := \top]$. The \emph{support} $\supp(\varphi) \subseteq V$ of $\varphi$ is the set of its essential variables.
\end{defi}
For $W \subseteq V$ we write $\B(W)$ for the subalgebra of $\B(V)$ generated by the variables in~$W$. Its elements are the classes of formulas in $\BF(W)$, and it is itself free on $W$, whence the notation agrees with $\B(W) = \BF(W)/{\equiv}$. The support characterises membership in these subalgebras, as the following lemma shows. 

\begin{lemma}\label{lem:support}
For $\varphi \in \B(V)$ and $W \subseteq V$, we have $\varphi \in \B(W)$ if and only if $\supp(\varphi) \subseteq W$.
\end{lemma}
\begin{proof}
If $\varphi \in \B(W)$, choose a representative in $\BF(W)$. Substituting $\bot$ or $\top$ for a variable $v \notin W$ does not change it, so no such $v$ is essential. Conversely, suppose $\supp(\varphi) \subseteq W$. For every $v \in V$ we have the Shannon expansion
\[
\varphi = (v \wedge \varphi[v := \top]) \vee (\neg v \wedge \varphi[v := \bot]),
\]
and if $v$ is not essential the right-hand side equals $\varphi[v := \bot]$. Substituting $\bot$ for every $v \in V \setminus W$ in turn therefore yields a formula equal to $\varphi$ whose representative lies in $\BF(W)$, so $\varphi \in \B(W)$.
\end{proof}

In particular, any substitution $g$ satisfying $\supp(g(v)) \subseteq W$ for every $v\in W$ maps $\B(W)$ into itself. If such a substitution $g$ is also an automorphism of $\B(V)$, its restriction to the finite subalgebra $\B(W)$ is injective, and hence $g(\B(W))=\B(W)$. For a complementation-preserving literal permutation, this support condition holds precisely when its underlying variable permutation preserves $W$ setwise.

\begin{ex}\label{ex:support}
Let $V = \{x,y,z\}$ and consider
\[
  \varphi = (x \to y) \wedge (z \vee \neg z).
\]
The variable $z$ occurs in $\varphi$ but is not essential: substituting either constant for $z$ gives  $\varphi[z := \bot] = \varphi[z := \top] = x \to y$. Hence $\supp(\varphi) = \{x,y\}$, and by Lemma~\ref{lem:support} we have $\varphi \in \B(\{x,y\})$; indeed $\varphi = x \to y$ in $\B(V)$. By contrast, $x$ is essential, since $\varphi[x := \bot] = \top \neq y = \varphi[x := \top]$.

Now let $g$ be the endomorphism of $\B(V)$ determined by
\[
g(x) = x, \qquad g(y) = x \oplus y, \qquad g(z) = z.
\]
The induced map on assignments is $\hat g(a,b,c) = (a,\, a \oplus b,\, c)$, which is a bijection (it is its own inverse), so $g$ is an automorphism by Corollary~\ref{cor:hat}(ii). Applying $g$ to $\varphi$ by substitution gives
\[
g(\varphi) = (x \to (x \oplus y)) \wedge (z \vee \neg z)
            = \neg x \vee \neg y,
\]
with $\supp(g(\varphi)) = \{x,y\}$.
\end{ex}
  
\subsection{Dependency Quantified Boolean Formulas} \label{prelim:DQBF}

\paragraph{Syntax.} A \emph{quantified Boolean formula} (QBF) (in prenex normal form) on a set of variables $V = \{v_1,\dots,v_m\}$ is a formula of the form
\[
  Q_1 v_1\, Q_2 v_2 \cdots Q_m v_m .\, \varphi
\]
with quantifiers $Q_i \in \{\exists,\forall\}$ and \emph{matrix} $\varphi \in \BF(V)$. An existentially quantified variable $v_i$ depends on all universally quantified variables preceding it, which creates a linear dependency structure of the variables.

A \emph{dependency quantified Boolean formula} (DQBF) is a generalisation of a QBF which allows for a non-linear dependency structure of the variables. This is achieved by explicitly annotating each existential variable with a set of universal variables that it may depend on. Since the dependency structure is explicitly encoded in this format, we only need one block of universal quantifiers followed by one block of existential quantifiers. Formally, a DQBF is a formula of the form
\[
\forall x_1,\dots,x_n\; \exists y_1(D_1),\dots,y_k(D_k) .\, \varphi,
\]
where $X = \{x_1,\dots,x_n\}$ and $Y = \{y_1,\dots,y_k\}$ are two disjoint finite sets of variables, $D_1,\dots,D_k \subseteq X$ are the \emph{dependency sets} of $y_1,\dots,y_k$ respectively, and $\varphi \in \BF(X \cup Y)$. The part 
\[
P = \forall x_1,\dots,x_n\,\exists y_1(D_1),\dots,y_k(D_k)
\]
is called the \emph{prefix}, and we write $\Phi = P.\varphi$. We denote (D)QBFs by uppercase Greek letters $\Phi,\Psi,\dots$. We say that $\Phi$ is in prenex conjunctive normal form (PCNF) if its matrix $\varphi$ is in CNF. 

The truth value of a (D)QBF, defined below, depends only on the class of its matrix in $\B(X \cup Y)$, and so do all notions introduced in Sections~\ref{sec:symm}--\ref{sec:canonisers}. In accordance with Remark~\ref{rem:convention} we therefore identify $P.\varphi$ with $P.\psi$ whenever $\varphi \equiv \psi$, and treat the matrix as an element of $\B(X \cup Y)$ throughout those sections. In particular, the variables the matrix depends on are those in $\supp(\varphi)$. For a concrete syntactic representation $\varphi \in \BF(V)$, we write $\Var(\varphi)$ and $\Lit(\varphi)$ for the sets of variables and literals occurring in $\varphi$, respectively. Unlike $\supp(\varphi)$, these sets need not be invariant under propositional equivalence. We write $\Lit(V) \subseteq \BF(V)$ for the set of all literals over $V$. For a DQBF $\Phi=P.\varphi$ in PCNF, following \cite{Blinkhorn20} we define its \emph{size} as the size of its matrix, 
\[
|\Phi| := |\varphi| := \sum_{C\in\varphi}|C|,
\]
counting literal occurrences. All size bounds refer to the supplied syntactic representations. 

Let $\Phi = Q_1 v_1 \cdots Q_m v_m .\, \varphi$ be a QBF with $X = \{x_1,\dots,x_n\}$ and $Y = \{y_1,\dots,y_k\}$ its sets of universal and existential variables respectively, ordered according to position of occurrence. We define its \emph{associated DQBF} by
\[
\Phi^{\mathrm{dep}} = \forall X\, \exists y_1(D_{y_1}) \cdots \exists y_k(D_{y_k}) .\, \varphi,
\qquad
D_y = \{x \in X \mid x <_\Phi y\},
\]
for every existential variable $y \in Y$. That is, $D_y$ is the set of all universal variables that occur before $y$ in the QBF prefix. Thus every QBF canonically induces an equivalent DQBF with linear dependency structure $D_{y_1} \subseteq D_{y_2} \subseteq \dots \subseteq D_{y_k}$. For dependency sets $D_i, D_j \subseteq X$, we write $D_i \parallel D_j$ if $D_i$ and $D_j$ are \emph{incomparable} with respect to inclusion, that is, $D_i \not\subseteq D_j$ and $D_j \not\subseteq D_i$.

It will be useful to extend the notion of dependency sets from existential variables to formulas over the existential variables. For a prefix $P = \forall X\,\exists y_1(D_1),\dots,y_k(D_k)$ and $\psi \in \B(Y)$, define
\[
\Dep(\psi) := \bigcup_{y_\ell \in \supp(\psi)} D_\ell \;\subseteq X.
\]
Thus $\Dep(\psi)$ collects the universal variables that $\psi$ may depend on through the existential variables essential for it. In particular $\Dep(y_j) = D_j$, $\Dep(\psi)$ is monotone in $\supp(\psi)$, and $\Dep(\psi) = \emptyset$ whenever $\psi$ is a constant. This notion is well-defined on $\B(Y)$ since $\supp$ is, and it will be used to formulate the admissibility condition for symmetries in Section~\ref{sec:syntactic}.

\paragraph{Semantics.} For a prefix $P = \forall x_1,\ldots, x_n \exists y_1(D_1), \ldots, y_k(D_k)$, an \emph{interpretation} for $P$ is a tuple $s = (s_1,\ldots,s_k)$ of functions $s_j \colon \{0,1\}^{|D_j|} \to \{0,1\}$ for $j = 1,\ldots,k$. The function $s_j$ specifies the value of the existential variable $y_j$ given the values of the universal variables $D_j \subseteq X$ that it depends on, and is called a \emph{Skolem function}. Every Skolem function $s_j$ can be represented by a Boolean formula in $\B(D_j)$. We denote by $\mathbb{S}(P)$ the set of all interpretations of $P$. An interpretation of $P$ together with an assignment of the variables $X$ naturally induces an assignment of the variables $Y$ as well, by applying the Skolem functions to the values of $X$. Formally, given an interpretation $s = (s_1,\ldots,s_k) \in \mathbb{S}(P)$ and an assignment $\sigma \in \mathbb{A}(X)$, the \emph{induced assignment} $\sigma_s \in \mathbb{A}(X \cup Y)$ of $s$ and $\sigma$ is defined as 
\[
\begin{cases}
    \sigma_s(x_i) = \sigma(x_i), & 1 \leq i \leq n, \\
    \sigma_s(y_j) = s_j(\sigma(x_{i_1}), \ldots, \sigma({x_{i_d}})), & 1 \leq j \leq k, 
\end{cases}
\]
where $D_j = \{x_{i_1}, \ldots, x_{i_d}\}$ is the dependency set of $y_j$ and $i_1 < \ldots < i_d$. 
\begin{rem} \label{skolem_extensions}
Let us call a map $\mathbb{A}(X) \to \mathbb{A}(X \cup Y)$ which is right-inverse to the canonical restriction map an \emph{extension}.\footnote{Concretely, it extends an assignment of the variables $X$ to an assignment of the variables $X \cup Y$, leaving the assignment on $X$ unchanged.} On the one hand, an interpretation $s \in \mathbb{S}(P)$ induces an extension $(-)_s \colon \mathbb{A}(X) \to \mathbb{A}(X \cup Y) \colon \sigma \mapsto \sigma_s$ which additionally satisfies the \textit{dependency condition}: $\sigma_s(y_j) = \tau_s(y_j)$ for all $\sigma, \tau \in \mathbb{A}(X)$ with $\sigma |_{D_j} = \tau |_{D_j}$. On the other hand, an extension $\alpha : \mathbb{A}(X) \to \mathbb{A}(X \cup Y)$ that satisfies the dependency condition also induces an interpretation $s^\alpha \in \mathbb{S}(P)$. These two operations are mutually inverse. This gives us two perspectives on interpretations: a tuple $s = (s_1,\ldots,s_k) \in \mathbb{S}(P)$, or an extension $(-)_s : \mathbb{A}(X) \to \mathbb{A}(X \cup Y)$ satisfying the dependency condition. This alternative perspective will turn out useful later in the paper, when studying how syntactic symmetries induce semantic symmetries. 
\end{rem}
The \emph{truth value} of a DQBF $\Phi$ under an interpretation $s$ is given by:  
\[
[\Phi]_s = \bigwedge_{\sigma \in \mathbb{A}(X)} [\varphi]_{\sigma_s}. 
\]
We say that an interpretation $s$ is a \emph{model} of $\Phi$ if $[\Phi]_s = 1$. $\Phi$ is called \emph{true} if it has at least one model, and is \emph{false} otherwise. A model of a DQBF can be visualised by a tree with dependency edges, as Example \ref{ex1} illustrates. For each existential variable $y_j$ in a DQBF $\Phi$, the dependency edges induce an equivalence relation on $\mathbb{A}(X)$. Formally, we have that $\sigma \sim_j \tau$ if $\sigma |_{D_j} = \tau |_{D_j}$. We call the induced equivalence classes \emph{dependency clusters}. Specifically, for $\sigma \in \mathbb{A}(X)$ and $j \in \{1,\ldots,k\}$, we use the notation
\[
U_{\sigma,j}
:=
\{\tau \in \mathbb A(X) \mid
\tau|_{D_j}=\sigma|_{D_j}\},
\]
for the dependency cluster of $\sigma$ with respect to $D_j$. Note that the dependency condition of an extension in Remark \ref{skolem_extensions} precisely states that an interpretation $s \in \mathbb{S}(P)$ takes the same value on all points in a dependency cluster.

\begin{ex} \label{ex1}
Consider the DQBF 
\[
\Phi = P.\varphi = \forall x_1,x_2,x_3 \exists y_1 (\{x_2\}), y_2(\{x_3\}).((x_1 \leftrightarrow y_1) \vee (x_2 \leftrightarrow y_2) \vee (x_1 \leftrightarrow x_3)). 
\]
The dependency sets are given by $D_1 = \{x_2\}$ and  $D_2 = \{x_3\}$, which are not linearly ordered (i.e. $D_1 \nsubseteq D_2$ and $D_2 \nsubseteq D_1$). A model for a DQBF can be visualised as a tree which selects the value of the existential variables for every configuration (path) of universal variables, in such a way that on every branch the matrix $\varphi$ is true. The dependency structure can be visualised by dashed lines in such a tree: two nodes in the layer $y_i$ are tied together if their configurations of the universal variables agree on the dependency set $D_i$. The dependency clusters consist of all points that are tied together by these dashed lines. A model $s$ of $\Phi$ is visualised in Figure \ref{fig:model_exmp_mix1}. The Skolem functions of this model are given by $y_1 = f_{1}(x_2)=x_2$ and $y_2 = f_{2}(x_3)=x_3$. In tuple notation the model is given by $s = (x_2,x_3)$.

\begin{figure}[h] 
\begin{center}
\begin{forest}
for tree={draw, circle, inner sep=2pt, s sep=35pt, l=23pt}
[, phantom, s sep=2cm
[{}, name=x1
    % x2 = 0
    [{}, edge label={node[font=\small, midway,above]{$0$}}, name=x20
        % x3 = 0
        [{}, edge label={node[font=\small, midway,left]{$0$}}, name=x300
            [{}, edge label={node[font=\small,midway,left]{$0$}}, name=y1000
                [{}, edge label={node[font=\small,midway,left]{$0$}}, name=y2000
                    [{}, edge label={node[font=\small,midway,left]{$0$}}]
                ]
            ]
            [{}, edge label={node[font=\small,midway,right]{$1$}}, name=y1001
                [{}, edge label={node[font=\small,midway,right]{$0$}}, name=y2001
                    [{}, edge label={node[font=\small,midway,right]{$1$}}]
                ]
            ]
        ]
        % x3 = 1
        [{}, edge label={node[font=\small, midway,right]{$1$}}, name=x301
            [{}, edge label={node[font=\small,midway,left]{$0$}}, name=y1010
                [{}, edge label={node[font=\small,midway,left]{$1$}}, name=y2010
                    [{}, edge label={node[font=\small,midway,left]{$0$}}]
                ]
            ]
            [{}, edge label={node[font=\small,midway,right]{$1$}}, name=y1011
                [{}, edge label={node[font=\small,midway,right]{$1$}}, name=y2011
                    [{}, edge label={node[font=\small,midway,right]{$1$}}]
                ]
            ]
        ]
    ]
    % x2 = 1
    [{}, edge label={node[font=\small,midway,above]{$1$}}, name=x21
        % x3 = 0
        [{}, edge label={node[font=\small, midway,left]{$0$}}, name=x310
            [{}, edge label={node[font=\small,midway,left]{$0$}}, name=y1100
                [{}, edge label={node[font=\small,midway,left]{$0$}}, name=y2100
                    [{}, edge label={node[font=\small,midway,left]{$0$}}]
                ]
            ]
            [{}, edge label={node[font=\small,midway,right]{$1$}}, name=y1101
                [{}, edge label={node[font=\small,midway,right]{$0$}}, name=y2101
                    [{}, edge label={node[font=\small,midway,right]{$1$}}]
                ]
            ]
        ]
        % x3 = 1
        [{}, edge label={node[font=\small, midway,right]{$1$}}, name=x311
            [{}, edge label={node[font=\small,midway,left]{$0$}}, name=y1110
                [{}, edge label={node[font=\small,midway,left]{$1$}}, name=y2110
                    [{}, edge label={node[font=\small,midway,left]{$0$}}]
                ]
            ]
            [{}, edge label={node[font=\small,midway,right]{$1$}}, name=y1111
                [{}, edge label={node[font=\small,midway,right]{$1$}}, name=y2111
                    [{}, edge label={node[font=\small,midway,right]{$1$}}]
                ]
            ]
        ]
    ]
]
]
\node [above=3pt of x1] {$s$};
\node (p) [left=15pt of y2000] {$y_{2}$:};
\node at (p |- y1000) {$y_{1}$:};
\node at (p |- x300) {$x_{3}$:};
\node at (p |- x20) {$x_{2}$:};
\node at (p |- x1) {$x_{1}$:};
\path (y1000) edge[dashed, bend right=20]  (y1001);
\path (y1010) edge[dashed, bend right=20]  (y1011);
\path (y1100) edge[dashed, bend right=20]  (y1101);
\path (y1110) edge[dashed, bend right=20]  (y1111);
\path (y1001) edge[dashed, bend right=20]  (y1100);
\path (y1011) edge[dashed, bend right=20]  (y1110);
\path (y2000) edge[dashed, bend right=20]  (y2010);
\path (y2001) edge[dashed, bend right=20]  (y2011);
\path (y2100) edge[dashed, bend right=20]  (y2110);
\path (y2101) edge[dashed, bend right=20]  (y2111);
\path (y2010) edge[dashed, bend right=20]  (y2100);
\path (y2011) edge[dashed, bend right=20]  (y2101);
% --- BIG RED BOX AROUND ALL y1 NODES ---
\begin{scope}[on background layer]
\node[
  draw=teal,
  fill=teal,
  fill opacity=0.15,
  draw opacity=1,
  very thick,
  rounded corners,
  inner sep=4pt,
  fit=(y1000)(y1001)
] {};
\node[
  draw=amber,
  fill=amber,
  fill opacity=0.15,
  draw opacity=1,
  very thick,
  rounded corners,
  inner sep=4pt,
  fit=(y1010)(y1011)
] {};
\node[
  draw=teal,
  fill=teal,
  fill opacity=0.15,
  draw opacity=1,
  very thick,
  rounded corners,
  inner sep=4pt,
  fit=(y1100)(y1101)
] {};
\node[
  draw=amber,
  fill=amber,
  fill opacity=0.15,
  draw opacity=1,
  very thick,
  rounded corners,
  inner sep=4pt,
  fit=(y1110)(y1111)
] {};
\node[
  draw=red,
  fill=red,
  fill opacity=0.15,
  draw opacity=1,
  very thick,
  rounded corners,
  inner sep=4pt,
  fit=(y2000)
] {};
\node[
  draw=blue,
  fill=blue,
  fill opacity=0.15,
  draw opacity=1,
  very thick,
  rounded corners,
  inner sep=4pt,
  fit=(y2001)
] {};
\node[
  draw=red,
  fill=red,
  fill opacity=0.15,
  draw opacity=1,
  very thick,
  rounded corners,
  inner sep=4pt,
  fit=(y2010)
] {};
\node[
  draw=blue,
  fill=blue,
  fill opacity=0.15,
  draw opacity=1,
  very thick,
  rounded corners,
  inner sep=4pt,
  fit=(y2011)
] {};
\node[
  draw=red,
  fill=red,
  fill opacity=0.15,
  draw opacity=1,
  very thick,
  rounded corners,
  inner sep=4pt,
  fit=(y2100)
] {};
\node[
  draw=blue,
  fill=blue,
  fill opacity=0.15,
  draw opacity=1,
  very thick,
  rounded corners,
  inner sep=4pt,
  fit=(y2101)
] {};
\node[
  draw=red,
  fill=red,
  fill opacity=0.15,
  draw opacity=1,
  very thick,
  rounded corners,
  inner sep=4pt,
  fit=(y2110)
] {};
\node[
  draw=blue,
  fill=blue,
  fill opacity=0.15,
  draw opacity=1,
  very thick,
  rounded corners,
  inner sep=4pt,
  fit=(y2111)
] {};
\end{scope}
\end{forest}
\end{center}

    \caption{Model $s$ of $\Phi$, with coloured dependency clusters.}
\label{fig:model_exmp_mix1}
\Description{A binary assignment tree for three universal variables. The existential values are y1 equals x2 and y2 equals x3. Dashed links and coloured regions identify dependency clusters that must share an existential value.}
\end{figure}
\end{ex}

We end this section with a short lemma that will be useful in the remainder of the paper. The proof follows immediately from the relevant definitions. 
\begin{lemma} \label{lem:eval_conjunction} 
Let $P$ be a DQBF prefix with universal variables $X$ and existential variables $Y$. For all $\varphi,\psi \in \B(X\cup Y)$ and $s \in \mathbb{S}(P)$, we have
\[
[P.(\varphi \land \psi)]_s
=
[P.\varphi]_s\land[P.\psi]_s.
\]
\end{lemma}

\subsection{Succinct NP-complete Problems} \label{sect:succinct}
DQBF was introduced in 1979 by Peterson and Reif \cite{PR79} in the context of multiplayer games with incomplete information. They established that its truth problem is NEXPTIME-complete, see also the subsequent journal treatment by \citet{PRA01}.
\begin{thm}[{\citep[Theorem~5.2.1]{PRA01}}]
SAT(DQBF) is \textsf{NEXPTIME}-complete.
\end{thm}
Chen  et al. \cite{CHLT22} provide a general method for reducing $\NEXP$-complete problems to DQBF and they provide explicit reductions from a variety of well-known $\NEXP$-complete problems to DQBF, including succinctly represented \textsf{NP}-complete problems such as graph colourability, subset sum, Hamiltonian cycle, and set packing. We illustrate such a reduction for graph colourability. 

\paragraph{Succinct graph colourability.} A Boolean circuit $C(\overline{u},\overline{v})$, with $|\overline{u}| = |\overline{v}| = n$, represents a graph $G(C) = (V,E)$ where $V = \{0,1\}^n$ and $(\overline{a},\overline{b}) \in E$ if and only if $C(\overline{a}, \overline{b}) = 1$. Given an input circuit $C$ the task is to decide whether $G(C)$ is $3$-colourable, which is $\NEXP$-complete. The existence of a proper $3$-colouring function (i.e. no two adjacent vertices have the same colour) can be encoded as a DQBF as follows. Let $c \colon V \to \{01,10,11\}$ denote a proper $3$-colouring function. We use the notation $X = (x_1,\ldots,x_n), Y = (y_1,\ldots,y_n)$, $F(X) = (f_1(X), f_2(X))$ and $F'(Y) = (f'_1(Y), f'_2(Y))$. Define
\begin{align} 
\begin{split}\label{graph_col}
    \Psi_C := \forall X, Y \exists F(X), F'(Y).[& (X = Y \to F(X) = F'(Y)) \\
    \wedge \ & F(X) \neq (0,0) \wedge F'(Y) \neq (0,0) \\
    \wedge \ & (C(X,Y) = 1 \to F(X) \neq F'(Y))]. 
\end{split}
\end{align}

It is not hard to see that models of $\Psi_C$ correspond exactly with proper 3-colourings of the graph $G(C)$, with the colouring $c : V \to \{01,10,11\}$ given by $c(X) = F(X) = F'(X)$. For example for $n=2$, the square $G = C_4$ can be represented by the circuit $C(\overline{u}, \overline{v}) = (u_1 \oplus v_1) \oplus (u_2 \oplus v_2)$. One possible proper colouring is illustrated in Figure \ref{fig:c4-colouring}. 

\begin{figure}[h]
\centering
% --- Left subfigure ---
\begin{subfigure}[t]{0.48\linewidth}
\centering
\resizebox{0.7\linewidth}{!}{%
\begin{tikzpicture}[
    baseline=(current bounding box.south),
    vertex/.style={
        circle, draw, minimum size=2mm,
        inner sep=0pt, scale=0.8
    },
    >=Latex,
    font=\small
]
% Vertices
\node[vertex, fill=red!70]         (v00) at (0,1.6) {};
\node[vertex, fill=green!60!black] (v01) at (1.6,1.6) {};
\node[vertex, fill=green!60!black] (v10) at (0,0) {};
\node[vertex, fill=red!70]         (v11) at (1.6,0) {};

% Edges
\draw[thick] (v00) -- (v01) -- (v11) -- (v10) -- (v00);

% Labels
\node[above=2pt] at (v00.north) {$00$};
\node[above=2pt] at (v01.north) {$01$};
\node[below=2pt] at (v10.south) {$10$};
\node[below=2pt] at (v11.south) {$11$};

% Title
\node[font=\normalsize\bfseries] at (0.8,2.5) {$C_4$};

% Legend
\node[anchor=west, align=left, font=\scriptsize] at (2.5,0.8) {
    \textbf{Colour code}\\[2pt]
    {\color{red!70}\large$\bullet$} $01$ = Red\\
    {\color{blue!70}\large$\bullet$} $10$ = Blue\\
    {\color{green!60!black}\large$\bullet$} $11$ = Green
};
\end{tikzpicture}%
}
\caption{A proper colouring of $C_4$.}
\label{fig:c4-colouring}
\end{subfigure}
\hfill
% --- Right subfigure ---
\begin{subfigure}[t]{0.48\linewidth}
\centering
\resizebox{0.9\linewidth}{!}{%
\begin{tikzpicture}[
    baseline=(current bounding box.south),
    vertex/.style={
        circle, draw, minimum size=3mm,
        inner sep=0pt, scale=0.4
    },
    >=Latex,
    font=\small
]
% Title and column headings
\node[font=\normalsize\bfseries] at (1.5,3.0)
    {$\mathrm{DPHP}_2$};
\node at (0,2.4) {Inputs};
\node at (3,2.4) {Output values};

% Inputs
\node[vertex, fill=blue!70] (p00) at (0,1.8) {};
\node[vertex, fill=blue!70] (p01) at (0,0.6) {};
\node[vertex, fill=blue!70] (p10) at (0,-0.6) {};
\node[vertex, fill=blue!70] (p11) at (0,-1.8) {};

% Output values
\node[vertex, fill=blue!70] (h00) at (3,1.8) {};
\node[vertex, fill=blue!70] (h01) at (3,0.6) {};
\node[vertex, fill=red!70]  (h10) at (3,-0.6) {};
\node[vertex, draw=gray, fill=gray!20]
    (h11) at (3,-1.8) {};

% Candidate mapping
\draw[->, thick, blue!70] (p00) -- (h00);
\draw[->, thick, blue!70] (p01) -- (h01);
\draw[->, thick, red!70]  (p10) -- (h10);
\draw[->, thick, red!70]  (p11) -- (h10);

% Input labels
\node[left=3pt] at (p00.west) {$00$};
\node[left=3pt] at (p01.west) {$01$};
\node[left=3pt] at (p10.west) {$10$};
\node[left=3pt] at (p11.west) {$11$};

% Output labels
\node[right=3pt] at (h00.east) {$00$};
\node[right=3pt] at (h01.east) {$01$};
\node[right=3pt] at (h10.east) {$10$};
\node[right=3pt, text=gray] at (h11.east)
    {$11=H$ \quad (excluded)};

% Excluded hole
\draw[gray, thick]
    (2.88,-1.92) -- (3.12,-1.68)
    (2.88,-1.68) -- (3.12,-1.92);

% Collision annotation
\node[anchor=west, text=red!70, font=\scriptsize]
    at (4,-0.6) {$c(10)=c(11)$};
\end{tikzpicture}%
}
\caption{A candidate for $\mathrm{DPHP}_2$ avoiding $H=11$.}
\label{fig:dphp-two}
\end{subfigure}

\caption{Illustrations of succinct graph $3$-colourability and the succinct pigeonhole principle.}
\label{fig:succinct-examples}
\Description{Left: a four-cycle with vertices labelled by two-bit strings and a proper colouring using two colours. Right: four inputs map to three allowed output values; inputs 10 and 11 collide at output 10, while output 11 is excluded.}
\end{figure}

\paragraph{Succinct pigeonhole principle.} The graph 3-colourability example can be generalised in a straightforward way to $k$-colourability. A special case, namely $(2^n-1)$-colourability of the complete graph $K_{2^n}$ (encoded by the circuit $C(\overline{u}, \overline{v}) = (\overline{u} \neq \overline{v})$) expresses the statement that there exists an injective function $c : \{0,1\}^n \to \{0,1\}^n \backslash \{0^n\}$. This resembles the functional pigeonhole principle in propositional logic, which expresses that there exists an injective function $[n] \to [n-1]$. Since the domain $\{0,1\}^n$ has size $2^n$, this is a succinct encoding of an exponential pigeonhole principle. Note that we can moreover make the `missing hole' from the codomain variable, expressing that there exists an injection $\{0,1\}^n \to \{0,1\}^n \backslash \{h\}$ for some $h \in \{0,1\}^n$. This leads to the following DQBF, which we denote by $\DPHP$ for `dependency pigeonhole principle':
\[
\DPHP_n := \forall X\,\forall Y\, \exists H(\emptyset)\, \exists F(X)\,\exists F'(Y). \bigl( (X=Y) \leftrightarrow (F=F') \bigr) \land (F \neq H) \land (F' \neq H).
\]
All tuples have length $n$. A model would determine a constant $h \in \{0,1\}^n$ and an injective function $c \colon \{0,1\}^n \to  \{0,1\}^n \setminus \{h\}$. An example for $n = 2$ is illustrated in Figure \ref{fig:dphp-two}. Appendix~\ref{app:dphp-cnf} gives a PCNF encoding of $\DPHP_n$ of size $\Theta(n)$. Thus, $\DPHP_n$ expresses the pigeonhole principle for $2^n$ pigeons and $2^n-1$ available holes in size $\Theta(n)$. By comparison, the standard unary propositional encoding uses one variable per pigeon-hole pair and, for each hole, a clause excluding each pair of pigeons from occupying it simultaneously. It therefore has $\Theta(2^{2n})$ variables and size $\Theta(2^{3n})$.

\section{Symmetries of DQBF} \label{sec:symm}
A symmetry is a transformation of an object that leaves a certain property \emph{invariant}. It was already observed in SAT \cite{Sak21}, CSP \cite{CJJPS06}, and QBF \cite{KS18} that there are two types of symmetries in this context: \emph{syntactic symmetries} and \emph{semantic symmetries}; for DQBF we use the same distinction. A syntactic symmetry, on the one hand, acts on the formula itself -- a syntactic object. It is a transformation of the formula that preserves its syntactic structure in an appropriate sense. A semantic symmetry, on the other hand, acts on models of the formula -- a semantic object. It transforms models of a formula $\Phi$ into models of $\Phi$. In this section we define and study syntactic and semantic symmetries of DQBF, and their interaction. 

\subsection{Syntactic Symmetries}\label{sec:syntactic}
The traditional definition of a symmetry of a propositional formula $\varphi$ is a permutation of its literals which respects negation and preserves the formula up to commutativity of the logical operations and double negation~\cite{DBDDM12, Sak21}. This definition was extended to QBF~\cite{AMS04}, with an additional condition that ensures that the quantifier blocks are preserved. However, in~\cite{KS18} the authors observed that the limitation to permutations of literals is unnecessarily restrictive. A natural generalisation is to consider automorphisms of $\B(V)$, that is, invertible substitutions of Boolean formulas for variables. By Corollary~\ref{cor:hat}(ii), such an automorphism may be specified by prescribing a formula $g(v)$ for every $v \in V$ and checking that the induced map $\hat g$ on assignments is a bijection. We illustrate this with two examples.

\begin{ex}\label{ex:automorphisms}
Let $V = \{x,y\}$. First, consider the literal permutation $g(x) = \neg y$, $g(y) = \neg x$. The induced map on assignments is $\hat g(a,b) = (\neg b, \neg a)$, which is a bijection satisfying $\hat g^{-1} = \hat g$. Hence $g$ extends to an automorphism of $\B(V)$. As a genuinely generalised example, let $g(x) = x$ and $g(y) = x \oplus y$. Then $\hat g(a,b) = (a, a \oplus b)$ is again a bijection with $\hat g^{-1} = \hat g$, so $g$ extends to an automorphism of $\B(V)$. Unlike the first example, this transformation is not induced by a permutation of literals.
\end{ex}

For a DQBF, an automorphism of $\B(X \cup Y)$ must additionally respect the prefix. We now formulate the appropriate condition.

\paragraph{Prefix-admissible automorphisms.}
Let $P = \forall X\,\exists y_1(D_1)\cdots\exists y_k(D_k)$ be a prefix. Recall from Section~\ref{prelim:DQBF} that for $\psi \in \B(Y)$ we write $\Dep(\psi) = \bigcup_{y_\ell \in \supp(\psi)} D_\ell$ for the set of universal variables that $\psi$ may depend on.

\begin{defi}[Prefix-admissible]\label{def:prefix-admissible}
An automorphism $g$ of $\B(X \cup Y)$ is \emph{prefix-admissible for $P$} if
\begin{enumerate}[label=(\roman*)]
\item $g(x) \in \B(X)$ for every $x \in X$ and $g(y) \in \B(Y)$ for every $y \in Y$;
\item $g^{-1}(x) \in \B(D_j)$  for every $j \in [k]$ and every $x \in \Dep(g(y_j))$.
\end{enumerate}
\end{defi}
Condition (i) says that $g$ does not mix universal and existential variables. Since $g$ is injective and $\B(X)$, $\B(Y)$ are finite, this implies $g(\B(X)) = \B(X)$ and $g(\B(Y)) = \B(Y)$. Condition (ii) ensures that the dependency structure is preserved. It requires that if $g(y_j)$ depends on $x$, the inverse image $g^{-1}(x)$ can be expressed using only variables in $D_j$. This is precisely what is needed for $g$ to transform interpretations into interpretations, as Proposition~\ref{prop:action} below shows. 

The following lemma shows that a prefix-admissible automorphism $g$ induces a bijection on $\A(X \cup Y)$ which restricts to a bijection on $\A(X)$.

\begin{lemma}\label{lem:restriction}
Let $P=\forall X\,\exists y_1(D_1)\cdots\exists y_k(D_k)$ be a DQBF prefix, and let $g\in\Aut(\B(X\cup Y))$ satisfy condition~(i) of Definition~\ref{def:prefix-admissible}. Then the induced assignment map
$\hat g:\A(X\cup Y)\to\A(X\cup Y)$ and its component $\hat g_X:\A(X)\to\A(X)$ are bijections.\footnote{We may drop the subscript and write $\hat g$ also for its restriction $\hat g_X$ to $\mathbb A(X)$ if its meaning is clear from the context.}
\end{lemma}
\begin{proof}
Since $g$ is an automorphism, $\hat g$ is bijective by Corollary~\ref{cor:hat}(ii). By condition (i), $\hat g(\sigma)(x) = \eval{g(x)}{\sigma}$ depends only on $\sigma|_X$ for $x \in X$, and likewise $\hat g(\sigma)(y)$ depends only on $\sigma|_Y$ for $y \in Y$. Hence $\hat g$ decomposes as
\[
\hat g(\sigma_X, \sigma_Y) = \bigl(\hat g_X(\sigma_X),\, \hat g_Y(\sigma_Y)\bigr)
\]
for maps $\hat g_X\colon \A(X) \to \A(X)$ and $\hat g_Y\colon \A(Y) \to \A(Y)$. Since $\hat g$ is bijective, both component maps are bijective. In particular, $\hat g_X$ is a bijection.
\end{proof}

Every prefix-admissible automorphism acts on interpretations $\Sk(P)$ via conjugation. Recall from Remark~\ref{skolem_extensions} that every interpretation $s \in \Sk(P)$ can be viewed as an extension $(-)_s\colon \A(X) \to \A(X \cup Y)$ satisfying the dependency condition, and vice versa.

\begin{prop}[Induced action on interpretations] \label{prop:action}
Let $P = \forall X\,\exists y_1(D_1)\cdots\exists y_k(D_k)$ be a DQBF prefix, and let $g$ be prefix-admissible for $P$. Then $g$ induces a bijection $\Sk(P) \to \Sk(P)$, $s \mapsto g(s)$, where $g(s)$\footnote{By abuse of notation, we write $g$ both for the prefix-admissible map as well as the induced map on interpretations.} is defined by the extension
\[
\xymatrix{
\mathbb{A}(X) \ar[r]^{\hat g^{-1}_X} & \mathbb{A}(X) \ar[r]^{(-)_s} & \mathbb{A}(X \cup Y) \ar[r]^{\hat g} & \mathbb{A}(X \cup Y). 
}
\]
That is, $g(s) \in \Sk(P)$ is the unique interpretation with the property that
\[
  (-)_{g(s)} = \hat g \circ (-)_s \circ \hat g_X^{-1}.
\]
\end{prop}
\begin{proof}
By Lemma~\ref{lem:restriction}, $\hat g$ and $\hat g_X$ are bijections. First, the map extends the universal assignment: for every $\sigma \in \A(X)$, putting $\rho = \hat g_X^{-1}(\sigma)$, we have
\[
  \bigl(\hat g \circ (-)_s \circ \hat g_X^{-1}\bigr)(\sigma)\big|_X
  = \hat g(\rho_s)\big|_X = \hat g_X(\rho) = \sigma.
\]
It remains to verify the dependency condition. Let $\sigma,\tau \in \A(X)$ agree on $D_j$, and put $\rho = \hat g_X^{-1}(\sigma)$ and $\eta = \hat g_X^{-1}(\tau)$. Let $y_\ell \in \supp(g(y_j))$ and $x_i \in D_\ell$. Then $x_i \in \Dep(g(y_j))$, so condition (ii) gives $g^{-1}(x_i) \in \B(D_j)$ by Lemma~\ref{lem:support}. Hence $\eval{g^{-1}(x_i)}{\sigma} = \eval{g^{-1}(x_i)}{\tau}$. By Corollary~\ref{cor:hat}(ii), $\widehat{g^{-1}} = \hat g^{-1}$, and therefore
\[
\rho(x_i) = \eval{g^{-1}(x_i)}{\sigma} = \eval{g^{-1}(x_i)}{\tau} = \eta(x_i).
\]
Thus $\rho$ and $\eta$ agree on $D_\ell$. Since $s \in \Sk(P)$, the dependency condition yields $\rho_s(y_\ell) = \eta_s(y_\ell)$. This holds for every $y_\ell \in \supp(g(y_j))$, so $\rho_s$ and $\eta_s$ agree on $\supp(g(y_j))$, and consequently $\eval{g(y_j)}{\rho_s} = \eval{g(y_j)}{\eta_s}$ by Lemma~\ref{lem:support}. By the definition of $(-)_{g(s)}$ and of $\hat g$, it follows that
\[
\sigma_{g(s)}(y_j) = \eval{g(y_j)}{\rho_s} = \eval{g(y_j)}{\eta_s}
= \tau_{g(s)}(y_j).
\]
Hence $(-)_{g(s)}$ satisfies the dependency condition for $y_j$, and therefore $g(s) \in \Sk(P)$.

Finally, suppose $g(s) = g(t)$. Then
\[
\hat g \circ (-)_s \circ \hat g_X^{-1} = \hat g \circ (-)_t \circ \hat g_X^{-1}.
\]
Since $\hat g$ and $\hat g_X$ are bijective, it follows that $(-)_s = (-)_t$, hence $s = t$. Thus $s \mapsto g(s)$ is injective, and since $\Sk(P)$ is finite, it is bijective.
\end{proof}
Figure~\ref{fig:action} illustrates the action of a prefix-admissible automorphism on an interpretation $s \in \Sk(P)$.

\begin{figure}[h]
\begin{center}
\scalebox{0.7}{
\begin{tikzpicture}

% --- LEFT COLUMN (s) ---
\node[draw, minimum width=1.6cm, minimum height=2.4cm] (ul) {$\hat g^{-1}(\sigma)$};
\node[draw, minimum width=1.6cm, minimum height=2.4cm, below=7mm of ul] (ll) {$\hat g^{-1}(\sigma)_s$};
\node (slabel) at ($(ul.north)+(0,0.4)$) {$s$};

% --- RIGHT COLUMN (g(s)) ---
\node[draw, minimum width=1.6cm, minimum height=2.4cm, right=4cm of ul] (ur) {$\sigma$};
\node[draw, minimum width=1.6cm, minimum height=2.4cm, below=7mm of ur] (lr) {$\sigma_{g(s)}$};
\node at ($(ur.north)+(0,0.4)$) {$g(s)$};

% --- ARROWS ---
\draw[->]
($(ur.west)+(0,0.4)$)
to[bend right=30]
node[above] {$\hat g^{-1}$}
($(ul.east)+(0,0.4)$);

\draw[->] (ll) to[bend right=30] node[below] {$_{\hat g}$} (lr);
\draw[->] (ul) -- node[right] {$(-)_s$} (ll);

% --- BRACES ---
\draw[decorate, decoration={brace, mirror, amplitude=6pt}]
($(ul.north west)+(-0.6,0)$) -- ($(ul.south west)+(-0.6,0)$)
node[midway, left=6pt] {$X$};

\draw[decorate, decoration={brace, mirror, amplitude=6pt}]
($(ll.north west)+(-0.6,0)$) -- ($(ll.south west)+(-0.6,0)$)
node[midway, left=6pt] {$Y$};

% --- GREY BOX ---
\node[
  draw,
  rounded corners,
  fill=gray!20,
  fill opacity=0.3,
  inner sep=6pt,
  fit=(ul)(ll)
] {};

\end{tikzpicture}
}
\end{center}
\caption{Visualisation of the action of a prefix-admissible map $g$ on an interpretation $s \in \mathbb{S}(P)$.}
\label{fig:action}
\Description{A commuting diagram for transforming an interpretation: apply the inverse assignment transformation to a universal assignment, extend it using s, and then apply the assignment transformation to obtain its extension under g(s).}
\end{figure}

The construction of Proposition~\ref{prop:action} makes sense for every automorphism satisfying condition~(i) of Definition~\ref{def:prefix-admissible}: by Lemma~\ref{lem:restriction}, $\hat g$ and $\hat g_X$ are then bijections, and the first part of the proof shows that
\[
  \Phi_g\bigl((-)_s\bigr) := \hat g \circ (-)_s \circ \hat g_X^{-1}
\]
is again an extension in the sense of Remark~\ref{skolem_extensions}, for every extension $(-)_s$. What condition~(ii) adds is precisely that $\Phi_g$ preserves the dependency condition.
 
\begin{prop}[Characterisation of prefix-admissibility]\label{prop:adm-char}
Let $P = \forall X\,\exists y_1(D_1)\cdots\exists y_k(D_k)$ be a DQBF prefix and let $g$ be an automorphism of $\B(X \cup Y)$ satisfying condition~(i) of Definition~\ref{def:prefix-admissible}. Then $g$ is prefix-admissible for $P$ if and only if $\Phi_g$ maps $\mathbb S(P)$ into $\mathbb S(P)$.
\end{prop}
 
\begin{proof}
If $g$ is prefix-admissible, then $\Phi_g$ maps $\mathbb S(P)$ into $\mathbb S(P)$ by Proposition~\ref{prop:action}. Conversely, suppose that condition~(ii) fails. Then there are $j \in [k]$ and $x \in \Dep(g(y_j))$ such that $g^{-1}(x) \notin \B(D_j)$, i.e. $\supp(g^{-1}(x)) \not \subseteq D_j$ by Lemma \ref{lem:support}. Fix $x' \in \supp(g^{-1}(x)) \setminus D_j$ and $y_\ell \in \supp(g(y_j))$ with $x \in D_\ell$. Since $y_\ell$ is essential for $g(y_j) \in \B(Y)$, there is an assignment $c \in \A(Y \setminus \{y_\ell\})$ such that the formula $g(y_j)$, with the variables other than $y_\ell$ fixed according to $c$, is a non-constant function $f$ of $y_\ell$. Formally, $f : \{0,1\} \to \{0,1\}$ defined by $f(b) := [g(y_j)]_{c \cup \{y_\ell \mapsto b\}}$ is non-constant, i.e. $f(0) \neq f(1)$. Define $s \in \mathbb S(P)$ by $s_\ell := x$, which is a legitimate Skolem function since $x \in D_\ell$, and $s_i := c(y_i)$ constant for $i \neq \ell$.
 
Let $\sigma \in \A(X)$ and put $\rho := \hat g_X^{-1}(\sigma)$. By Corollary~\ref{cor:hat}(ii), $\rho(x) = \eval{g^{-1}(x)}{\sigma}$, so $\rho_s(y_\ell) = \eval{g^{-1}(x)}{\sigma}$ and $\rho_s(y_i) = c(y_i)$ for $i \neq \ell$. Therefore we have $\rho_s|_Y = c \cup \{y_\ell \mapsto [g^{-1}(x)]_{\sigma}\}$, whence
\[
  \Phi_g\bigl((-)_s\bigr)(\sigma)(y_j)
  = \eval{g(y_j)}{\rho_s}
  = f\bigl(\eval{g^{-1}(x)}{\sigma}\bigr).
\]
Since $x'$ is essential for $g^{-1}(x)$, there is $\sigma \in \A(X)$ such that $\sigma' := \sigma[x' \mapsto \neg \sigma(x')]$ satisfies $\eval{g^{-1}(x)}{\sigma} \neq \eval{g^{-1}(x)}{\sigma'}$. As $f$ is non-constant, $\Phi_g((-)_s)$ assigns different values to $y_j$ at $\sigma$ and $\sigma'$, although $\sigma|_{D_j} = \sigma'|_{D_j}$ because $x' \notin D_j$. Thus $\Phi_g((-)_s)$ violates the dependency condition for $y_j$ and does not lie in $\mathbb S(P)$.
\end{proof}

Under condition (i), condition (ii) is therefore necessary and sufficient for $\Phi_g$ to preserve the dependency condition for extensions. This semantic characterisation of prefix-admissibility yields a short proof that the prefix-admissible automorphisms for $P$ form a group.
  
\begin{cor} \label{cor:adm-group}
Let $P$ be a DQBF prefix. The prefix-admissible automorphisms for $P$
form a subgroup
\[
  \mathrm{Adm}(P) \leq \Aut(\B(X \cup Y)).
\]
\end{cor}
 
\begin{proof}
The identity is prefix-admissible. Condition~(i) is preserved under composition, and under inverses because $g(\B(X)) = \B(X)$ and $g(\B(Y)) = \B(Y)$ by finiteness. For condition~(ii) we use Proposition~\ref{prop:adm-char}. By Corollary~\ref{cor:hat}(i), $\widehat{g \circ h} = \hat h \circ \hat g$ and likewise for the restrictions to $\A(X)$, so
\[
  \Phi_{g \circ h} = \Phi_h \circ \Phi_g .
\]
Hence if $\Phi_g$ and $\Phi_h$ map $\mathbb S(P)$ into $\mathbb S(P)$, so does $\Phi_{g \circ h}$, and $g \circ h$ is prefix-admissible. For inverses, $\Phi_{g^{-1}} = \Phi_g^{-1}$, again by Corollary~\ref{cor:hat}. By Proposition~\ref{prop:action}, $\Phi_g$ restricts to a bijection of $\mathbb S(P)$ onto itself, so its inverse maps $\mathbb S(P)$ into $\mathbb S(P)$. By Proposition~\ref{prop:adm-char}, $g^{-1}$ is prefix-admissible.
\end{proof}
 
Prefix-admissible automorphisms thus act on the syntactic object $\B(X \cup Y)$ as well as on the semantic object $\Sk(P)$. These two actions are compatible, as the following theorem shows.

\begin{thm}[Equivariance identity]\label{thm:equivariance}
Let $P$ be a DQBF prefix, let $g$ be prefix-admissible for $P$, and let $\varphi \in \B(X \cup Y)$. Then, for every $s \in \Sk(P)$,
\[
[P.g(\varphi)]_s = [P.\varphi]_{g(s)}.
\]
\end{thm}
\begin{proof}
Let $s \in \Sk(P)$. Since $\hat g_X$ is bijective, every $\sigma \in \A(X)$ can be written uniquely as  $\sigma = \hat g_X(\rho)$ for some $\rho \in \A(X)$. Hence
\[
  [P.\varphi]_{g(s)}
  = \bigwedge_{\sigma \in \A(X)} \eval{\varphi}{\sigma_{g(s)}}
  = \bigwedge_{\rho \in \A(X)} \eval{\varphi}{\hat g_X(\rho)_{g(s)}}.
\]
By definition of $g(s)$, $\hat g_X(\rho)_{g(s)} = \hat g(\rho_s)$. Moreover, by the evaluation identity (Lemma~\ref{lem:evaluation}), $\eval{\varphi}{\hat g(\rho_s)} = \eval{g(\varphi)}{\rho_s}$. Therefore,
\[
  [P.\varphi]_{g(s)}
  = \bigwedge_{\rho \in \A(X)} \eval{\varphi}{\hat g(\rho_s)}
  = \bigwedge_{\rho \in \A(X)} \eval{g(\varphi)}{\rho_s}
  = [P.g(\varphi)]_s. \qedhere
\]
\end{proof}
The equivariance identity enables us to transform models of $P.g(\varphi)$ into models of $P.\varphi$. Combined with bijectivity of $g$ on $\Sk(P)$ we obtain that prefix-admissible automorphisms preserve truth of DQBFs.

\begin{cor}[Truth preservation]\label{cor:truth}
Let $P$ be a DQBF prefix and let $g$ be prefix-admissible for $P$. Then, for every $\varphi \in \B(X \cup Y)$,
\[
P.\varphi \text{ is true} \iff P.g(\varphi) \text{ is true}.
\]
\end{cor}
\begin{proof}
Follows immediately from Theorem \ref{thm:equivariance} and the fact that $s \mapsto g(s)$ is a bijection on $\Sk(P)$ by Proposition~\ref{prop:action}. 
\end{proof}

We can now define a syntactic symmetry of a DQBF $\Phi = P.\varphi$ as a prefix-admissible automorphism that additionally fixes the matrix.

\begin{defi}[Syntactic symmetry group]\label{def:syntactic}
Let $\Phi = P.\varphi$ be a DQBF and let $G \leq \mathrm{Adm}(P)$ be a group of prefix-admissible automorphisms for $P$. We call $G$ a \emph{syntactic symmetry group} of $\Phi$ if
\[
g(\varphi) = \varphi
\]
for every $g \in G$, i.e. if $g(\varphi)$ and $\varphi$ are propositionally equivalent (Remark \ref{rem:convention}). We call an element $g \in G$ a \emph{syntactic symmetry} of $\Phi$.
\end{defi}

\begin{rem}\label{rem:models-preserved}
By Corollary~\ref{cor:truth}, every prefix-admissible automorphism preserves truth. A syntactic symmetry has a stronger property: it preserves not only truth of $\Phi$ but the whole set of its models. Indeed, since $g(\varphi) = \varphi$, Theorem~\ref{thm:equivariance} gives $[P.\varphi]_{g(s)} = [P.g(\varphi)]_s = [P.\varphi]_s$ for all $s \in \Sk(P)$.
\end{rem}
\begin{rem}\label{rem:generators}
By Corollary~\ref{cor:adm-group}, the prefix-admissible automorphisms fixing $\varphi$ form a subgroup of $\Aut(\B(X \cup Y))$, namely the stabiliser of $\varphi$ in $\mathrm{Adm}(P)$. Consequently, to exhibit a
syntactic symmetry group it suffices to exhibit a set of generators and to check prefix-admissibility and $g(\varphi) = \varphi$ for each generator. The group they generate is then a syntactic symmetry group. This is how the symmetry groups in the examples below are obtained.
\end{rem}

Every complementation-preserving permutation of $\Lit(X \cup Y)$ induces a bijection on assignments, and hence extends to an automorphism of $\B(X \cup Y)$ by Corollary~\ref{cor:hat}(ii). We write $\Glit$ for a group of symmetries of $\Phi$ induced by such literal permutations, and refer to its elements as \emph{literal symmetries}.

\begin{ex} \label{ex:colouring2}
Consider the $3$-colourability encoding $\Psi_C$ in \eqref{graph_col} for $C_4$. The graph automorphisms $\Aut(C_4)\cong D_4$ (acting simultaneously on the $X$-variables and $Y$-variables in $\Psi_C$), the exchange of the two copies $(X,F)$ and $(Y,F')$, and the colour symmetry
\[
(f_1,f_2) \mapsto (f_2,f_1),\qquad
(f'_1,f'_2) \mapsto (f'_2,f'_1),
\]
which exchanges the colours $01$ and $10$, generate the literal symmetry group
\[
G_{\lit}\cong D_4\times S_2\times S_2.
\]
Generalised syntactic symmetries capture all permutations of the three colours $\{01,10,11\}$. For instance,
\[
(f_1,f_2) \mapsto (f_1,f_1\oplus f_2),\qquad
(f'_1,f'_2) \mapsto (f'_1,f'_1 \oplus f'_2),
\]
exchanges $10$ and $11$. Hence, they generate a generalised syntactic symmetry group
\[
G_{\syn}\cong D_4\times S_2\times S_3.
\]
We have $|G_{\syn}|=96$, whereas $|G_{\lit}|=32$.
\end{ex}

\subsection{Semantic Symmetries}
Whereas syntactic symmetries transform Boolean formulas, semantic symmetries transform interpretations. For a DQBF $\Phi = P.\varphi$, a semantic symmetry is a bijection of $\mathbb S(P)$ that preserves the set of its models $\Mod(\Phi)$.

\begin{defi}[Semantic symmetry group] \label{def:sem_symm}
Let $\Phi=P.\varphi$ be a DQBF. A \emph{semantic symmetry group} of $\Phi$ is a subgroup $G \leq \Sym(\mathbb S(P))$ such that
\[
  [P.\varphi]_{f(s)} = [P.\varphi]_s
  \qquad
  \text{for every } f \in G
  \text{ and }s\in\mathbb S(P).
\]
\end{defi}

Given a syntactic symmetry group, we now define its associated semantic symmetry group. This allows different syntactic symmetries to be applied at different universal assignments, providing the flexibility needed for local lex-leader symmetry breaking in Section \ref{sect:lex_leader}. 

\begin{defi}[Associated semantic symmetry group] \label{def:ass_sem_group}
Let $\Phi=P.\varphi$ be a DQBF and let $G_{\syn}$ be a syntactic symmetry group of $\Phi$. We define its \emph{associated semantic symmetry group} $G_{\sem}$ as the set of all bijections $f \colon \mathbb S(P) \to \mathbb S(P)$ such that, for every $s \in \mathbb S(P)$ and every $\sigma \in \mathbb A(X)$, there exists some $g \in G_{\syn}$ with
\[
\hat g(\sigma)_{f(s)}=\hat g(\sigma_s).
\]
\end{defi}

Note that the definition uses only $P$ and $G_{\syn}$, and makes sense for every $G\leq\mathrm{Adm}(P)$: since $g(\top)=\top$ for every automorphism $g$, every such subgroup is a syntactic symmetry group of $P.\top$.

\begin{prop}\label{prop:associated-semantic-group}
Let $\Phi = P.\varphi$ be a DQBF and let $G_{\syn}$ be a syntactic symmetry group of $\Phi$. Then the associated semantic symmetry group $G_{\sem}$ is a semantic symmetry group of $\Phi$.
\end{prop}

\begin{proof}
We first show that $G_{\sem}$ is a group. The identity belongs to $G_{\sem}$, since for every $s\in\mathbb S(P)$ and $\sigma\in\mathbb A(X)$ we may choose $\id\in G_{\syn}$ and obtain $\sigma_s = \sigma_s$. Let $f,h\in G_{\sem}$ and fix $s\in\mathbb S(P)$ and $\sigma\in\mathbb A(X)$. By the defining property of $f$, there exists $g \in G_{\syn}$ such that $\hat g(\sigma)_{f(s)}=\hat g(\sigma_s)$. Applying the defining property of $h$ to $f(s)$ and $\hat g(\sigma)$, there exists $k\in G_{\syn}$ such that
\[
\hat k(\hat g(\sigma))_{h(f(s))}
=\hat k\bigl(\hat g(\sigma)_{f(s)}\bigr).
\]
Substituting the first equality and using $\widehat{g\circ k}=\hat k \circ \hat g$ (by Corollary~\ref{cor:hat}(i)) yields
\[
\widehat{(g\circ k)}(\sigma)_{(h \circ f)(s)}
=\widehat{(g\circ k)}(\sigma_s).
\]
Since $g\circ k\in G_{\syn}$, we conclude that $h \circ f \in G_{\sem}$.

Since $\mathbb S(P)$ is finite, every $f\in G_{\sem}$ has finite order as a bijection of $\mathbb S(P)$. Thus $f^m=\id$ for some $m\geq1$, and closure under composition gives $f^{-1}=f^{m-1}\in G_{\sem}$. Hence $G_{\sem}$ is a group.

It remains to show that this group consists of semantic symmetries. Let $f\in G_{\sem}$ and suppose that $[P.\varphi]_s=0$. Then there exists $\sigma\in\mathbb A(X)$ such that $[\varphi]_{\sigma_s}=0$. By definition of $G_{\sem}$, there exists $g\in G_{\syn}$ such that $\hat g(\sigma)_{f(s)}=\hat g(\sigma_s)$. By Lemma \ref{lem:evaluation} and $g(\varphi) = \varphi$, we obtain
\[
[\varphi]_{\hat g(\sigma)_{f(s)}}
=[\varphi]_{\hat g(\sigma_s)}
=[g(\varphi)]_{\sigma_s}
=[\varphi]_{\sigma_s}
=0.
\]
Hence $[P.\varphi]_{f(s)}=0$. Applying the same implication to $f^{-1}\in G_{\sem}$ and $f(s)$ gives the converse, and therefore $[P.\varphi]_s=[P.\varphi]_{f(s)}$ for every $s\in\mathbb S(P)$.
\end{proof}

The following proposition shows that every syntactic symmetry $g \in G_{\syn}$ induces a semantic symmetry in $G_{\sem}$. 

\begin{prop}\label{prop:induced-semantic-symmetry}
Let $\Phi=P.\varphi$ be a DQBF, let $G_{\syn}$ be a syntactic symmetry group of $\Phi$, and let $g \in G_{\syn}$. Then
\[
\hat g(\sigma)_{g(s)}=\hat g(\sigma_s)
\]
for every $s \in \mathbb S(P)$ and $\sigma \in \mathbb A(X)$. Consequently, the induced bijection $g \colon \mathbb S(P) \to \mathbb S(P), s \mapsto g(s)$ of Proposition \ref{prop:action} belongs to the associated semantic symmetry group $G_{\sem}$.
\end{prop}

\begin{proof}
By definition we have $(-)_{g(s)} = \hat g \circ(-)_s \circ \hat g^{-1}$. Hence $(-)_{g(s)}\circ\hat g = \hat g\circ(-)_s$, and therefore for every $s\in\mathbb S(P)$ and $\sigma\in\mathbb A(X)$, $\hat g(\sigma)_{g(s)} = \hat g(\sigma_s)$. Thus the defining condition of $G_{\sem}$ is satisfied by choosing the same $g\in G_{\syn}$ for every $s$ and $\sigma$. Hence $s\mapsto g(s)$ belongs to $G_{\sem}$.
\end{proof}

In Section \ref{sec:BA}, Corollary \ref{cor:hat}(i) we saw that the assignment $g \mapsto \hat{g}$ from automorphisms of $\B(V)$ to bijections of $\A(V)$ is contravariant, in the sense that $\widehat{g \circ h} = \hat{h} \circ \hat{g}$. This reversal of composition when going from syntax to semantics persists when dealing with symmetries, as the next proposition shows. 

\begin{prop}\label{prop:induced-semantic-action}
The map $\iota \colon G_{\syn} \to G_{\sem}, g \mapsto \bigl( s \mapsto g(s) \bigr)$ is an antihomomorphism, in the sense that $\iota(g \circ h) = \iota(h) \circ \iota(g)$. 
\end{prop}
\begin{proof}
Let $g,h\in G_{\syn}$ and $s\in\mathbb S(P)$. As shown in the proof of Corollary~\ref{cor:adm-group}, $\Phi_{g\circ h} = \Phi_h \circ \Phi_g$, so $(-)_{(g\circ h)(s)} = \Phi_h(\Phi_g((-)_s)) = (-)_{h(g(s))}$ and
hence $(g\circ h)(s) = h(g(s))$. The result follows. 
\end{proof}

By construction, for a DQBF $P.\varphi$, a syntactic symmetry group $G_{\syn}$, and its associated semantic symmetry group $G_{\sem}$, we have two group actions: 
\[
G_{\syn} \times \B(X \cup Y) \to \B(X \cup Y), \quad G_{\sem} \times \mathbb{S}(P) \to \mathbb{S}(P). 
\]
which are both given by function application. $G_{\syn}$ acts on the syntactic object $\B(X \cup Y)$ whereas $G_{\sem}$ acts on the semantic object $\mathbb{S}(P)$, which motivates the terminology. By Proposition \ref{prop:induced-semantic-action}, $G_{\syn}$ moreover induces a (right) action on $\mathbb{S}(P)$, by considering the induced semantic symmetry of a syntactic symmetry. 

\begin{ex}[Continuation of Example~\ref{ex:colouring2}] \label{ex:colouring3}
Recall that models of $\Psi_C$ correspond exactly to proper $3$-colourings $c:\{0,1\}^2\to\{01,10,11\}$. If $g\in G_{\syn}$ consists of a graph automorphism $\pi\in D_4$ and a colour permutation $\rho\in S_3$, then the induced semantic symmetry acts by
\[
g \cdot c = \rho \circ c \circ \pi^{-1}.
\]
Thus syntactic symmetries act on colourings by simultaneously renaming vertices and colours, as illustrated in Figure~\ref{fig:c4-semantic-symmetry} for $\pi : v \mapsto v \oplus 10$ and $\rho : R \leftrightarrow B$. 

\begin{figure}[h]
\centering
\begin{tikzpicture}[
    vertex/.style={
        circle,
        draw,
        minimum size=2mm,
        inner sep=0pt,
        scale=0.9
    },
    >=Latex,
    font=\small
]

% ============================================================
% Original colouring
% ============================================================

\node[font=\bfseries] at (0.8,3) {Original colouring $c$};

\node[vertex,fill=red!70]   (a00) at (0,1.8) {};
\node[vertex,fill=green!60!black] (a01) at (1.6,1.8) {};
\node[vertex,fill=green!60!black] (a10) at (0,0.2) {};
\node[vertex,fill=red!70]   (a11) at (1.6,0.2) {};

\draw (a00)--(a01)--(a11)--(a10)--(a00);

\node[above=2pt] at (a00.north) {$00$};
\node[above=2pt] at (a01.north) {$01$};
\node[below=2pt] at (a10.south) {$10$};
\node[below=2pt] at (a11.south) {$11$};

% ============================================================
% Graph automorphism
% ============================================================

\draw[->,thick] (2.2,1.0) -- (3.6,1.0)
    node[midway,above,align=center]
    {$\pi(v)=v\oplus10$};

% ============================================================
% After moving vertices
% ============================================================

\node[font=\bfseries] at (4.8,3) {$c\circ\pi^{-1}$};

\node[vertex,fill=green!60!black] (b00) at (4,1.8) {};
\node[vertex,fill=red!70]   (b01) at (5.6,1.8) {};
\node[vertex,fill=red!70]   (b10) at (4,0.2) {};
\node[vertex,fill=green!60!black] (b11) at (5.6,0.2) {};

\draw (b00)--(b01)--(b11)--(b10)--(b00);

\node[above=2pt] at (b00.north) {$00$};
\node[above=2pt] at (b01.north) {$01$};
\node[below=2pt] at (b10.south) {$10$};
\node[below=2pt] at (b11.south) {$11$};

% ============================================================
% Colour permutation
% ============================================================

\draw[->,thick] (6.2,1.0) -- (7.6,1.0)
    node[midway,above,align=center]
    {$\rho:\ R\leftrightarrow B$};

% ============================================================
% Final colouring
% ============================================================

\node[font=\bfseries] at (8.8,3) {New colouring $c'$};

\node[vertex,fill=green!60!black] (c00) at (8,1.8) {};
\node[vertex,fill=blue!60]  (c01) at (9.6,1.8) {};
\node[vertex,fill=blue!60]  (c10) at (8,0.2) {};
\node[vertex,fill=green!60!black] (c11) at (9.6,0.2) {};

\draw (c00)--(c01)--(c11)--(c10)--(c00);

\node[above=2pt] at (c00.north) {$00$};
\node[above=2pt] at (c01.north) {$01$};
\node[below=2pt] at (c10.south) {$10$};
\node[below=2pt] at (c11.south) {$11$};

\end{tikzpicture}

\caption{The semantic symmetry induced by a syntactic symmetry $g = (\pi,\rho)$ of $\Psi_{C_4}$.}
\label{fig:c4-semantic-symmetry}
\Description{Three colourings of a four-cycle. First a vertex permutation exchanges labels by XOR with 10; then a colour permutation exchanges red and blue, producing the transformed colouring.}
\end{figure}
There are $18$ proper $3$-colourings of $C_4$. Since every model $s$ satisfies the synchronisation conjunct $(X = Y) \to (F = F')$, the copy-swap symmetry acts trivially on models. Thus, the induced action on $\Mod(\Psi_C)$ factors through $S_3\times D_4$. Under this action the models form two orbits: the $6$ colourings using two colours and the $12$ colourings using all three colours,
\[
\Mod(\Psi_C)
=
\mathcal O_{2\text{-col}}\mathbin{\dot\cup}\mathcal O_{3\text{-col}}.
\]
\end{ex}

\subsection{Local Semantic Symmetries}
The semantic symmetry induced by a syntactic symmetry acts uniformly on interpretations. For lex-leader symmetry breaking, however, we need to apply such a symmetry only on a selected set of universal assignments, while leaving the remainder of the interpretation unchanged. This piecewise construction must respect the dependency constraints: assignments agreeing on $D_j$ must yield the same value for $y_j$. We therefore first characterise when interpretations can be \emph{glued} to obtain a valid interpretation.

\begin{defi}[Gluing condition]
Let $P=\forall X\,\exists y_1(D_1)\cdots\exists y_k(D_k)$ be a DQBF prefix, let
$s,t\in\mathbb S(P)$, and let $U\subseteq\mathbb A(X)$. We say that
$s$ and $t$ satisfy the \emph{gluing condition} with respect to $U$ if,
for every $j\in\{1,\ldots,k\}$ and all $\sigma,\tau\in\mathbb A(X)$ with
$\sigma\in U$, $\tau\notin U$, and $\sigma|_{D_j}=\tau|_{D_j}$, we have $\sigma_s(y_j)=\tau_t(y_j)$. 
\end{defi}
The gluing condition says precisely that no dependency cluster is cut inconsistently by the boundary of $U$.

\begin{lemma}\label{lem:gluing}
Let $P$ be a prefix, $s,t\in\mathbb S(P)$, and $U \subseteq \mathbb A(X)$. If $s$ and $t$ satisfy the gluing condition with respect to $U$, then there exists a unique interpretation $s\sqcup_U t\in\mathbb S(P)$ such that
\[
\sigma_{s\sqcup_U t}
=
\begin{cases}
\sigma_s, & \sigma\in U,\\
\sigma_t, & \sigma\notin U.
\end{cases}
\]
\end{lemma}

\begin{proof}
Define an extension $(-)_{s\sqcup_U t}:\mathbb A(X)\to\mathbb A(X\cup Y)$ by $\sigma \mapsto \sigma_{s \sqcup_U t}$. Since both $(-)_s$ and $(-)_t$ extend the universal assignment, so does $(-)_{s\sqcup_U t}$. It remains to verify the dependency condition. Let $j\in\{1,\ldots,k\}$ and let $\sigma,\tau\in\mathbb A(X)$ satisfy $\sigma|_{D_j}=\tau|_{D_j}$. If either $\sigma,\tau\in U$ or
$\sigma,\tau\notin U$, then
\[
\sigma_{s\sqcup_U t}(y_j)=\tau_{s\sqcup_U t}(y_j)
\]
follows from the fact that $s$ or $t$, respectively, respects the dependency set $D_j$. If $\sigma\in U$ and $\tau\notin U$, then the equality follows directly from the gluing condition. The remaining case is symmetric. Thus $(-)_{s\sqcup_U t}$ satisfies the dependency condition and therefore determines a unique interpretation $s\sqcup_U t\in\mathbb S(P)$ by Remark \ref{skolem_extensions}.
\end{proof}

This pointwise gluing of interpretations can, under an additional invariance condition, be realised by a semantic symmetry, as the next proposition shows.

\begin{prop}[Local semantic symmetry] \label{prop:local_sem_sym}
Let $\Phi=P.\varphi$ be a DQBF with syntactic symmetry group $G_{\syn}$ and associated semantic symmetry group $G_{\sem}$. Let $g\in G_{\syn}$, $s\in\mathbb S(P)$, and let $U\subseteq\mathbb A(X)$ be $g$-invariant, i.e. $\hat g(U)=U$. Suppose that $g(s)$ and $s$ satisfy the gluing condition with respect to $U$, and put
\[
s^{g,U}:=g(s)\sqcup_U s.
\]
Then the permutation $(s\;\;s^{g,U}) \in \Sym(\mathbb S(P))$, which exchanges $s$ and $s^{g,U}$ and fixes every other interpretation, belongs to $G_{\sem}$. 
\end{prop}
\begin{proof}
By Lemma~\ref{lem:gluing}, the interpretation $s^{g,U}=g(s)\sqcup_U s$ is well-defined. 
%If $s^{g,U}=s$, the claim is immediate, so suppose $s^{g,U}\neq s$. 
Also, $(s\;\;s^{g,U})$ is a bijection of $\mathbb S(P)$. It remains to verify the defining condition of $G_{\sem}$. Write $f=(s\;\;s^{g,U})$. First consider $s$ and let $\sigma\in\mathbb A(X)$. If $\sigma \in U$, then $\hat g(\sigma)\in U$ by $g$-invariance, and hence
\[
\hat g(\sigma)_{f(s)}
=
\hat g(\sigma)_{s^{g,U}}
=
\hat g(\sigma)_{g(s)}
=
\hat g(\sigma_s),
\]
where the last equality follows from Proposition \ref{prop:induced-semantic-symmetry}. If $\sigma\notin U$, then $\sigma_{f(s)} = \sigma_{s^{g,U}} = \sigma_s$, so the identity element witnesses the required condition.

Now consider $s^{g,U}$. If $\sigma\in U$, then $\sigma_{s^{g,U}}=\sigma_{g(s)}$. Moreover, applying Proposition \ref{prop:induced-semantic-symmetry} to $g^{-1} \in G_{\syn}$ and $g(s)$, and using that $f(s^{g,U})=s$, we derive
\[
\widehat{g^{-1}}(\sigma)_{f(s^{g,U})}
=
\widehat{g^{-1}}(\sigma)_s
=
\widehat{g^{-1}}(\sigma_{g(s)})
=
\widehat{g^{-1}}(\sigma_{s^{g,U}}).
\]
Thus $g^{-1}$ witnesses the required condition. If $\sigma\notin U$, then $\sigma_{s^{g,U}}=\sigma_s$, and again the identity witnesses the required condition.

Finally, if $t\notin\{s,s^{g,U}\}$, then $f(t)=t$, so the identity witnesses the condition for every $\sigma\in\mathbb A(X)$. Thus $f$ satisfies the defining condition of the associated semantic symmetry group in Definition \ref{def:ass_sem_group}, and therefore $(s\;\;s^{g,U})\in G_{\sem}$. 
\end{proof}

\begin{ex} \label{ex:mixing}
Consider
\[
\Phi=
\forall x_1,x_2\,
\exists y_1(\{x_1\})\,
\exists y_2(\{x_1\})\,
\exists y_3(\{x_2\}).\,
\bigl((y_1\leftrightarrow y_2)\wedge(y_3\leftrightarrow x_2)\bigr).
\]
A model is given by $s=(x_1,x_1,x_2)$. Let $g\in G_{\syn}$ be defined by
\[
g(y_1)=\neg y_1,\qquad g(y_2)=\neg y_2,
\]
and let $g$ act as the identity on all other variables. Then $g(s)=(\neg x_1,\neg x_1,x_2)$. Now let $U=\{(0,0),(0,1)\}$. Since $g$ fixes the universal variables, $\hat g(U)=U$. Moreover, the gluing condition is satisfied. Hence $s^{g,U}:=g(s)\sqcup_U s=(1,1,x_2)$ is well-defined and by Proposition~\ref{prop:local_sem_sym}, $\bigl(s\;\;s^{g,U}\bigr)\in G_{\sem}$. Figure~\ref{fig:local_sem_example} illustrates the construction.

\begin{figure}[h] 
\begin{center}
\begin{forest}
for tree={draw, circle, inner sep=2pt, s sep=35pt, l=22pt}
[, phantom, s sep=2cm
[{}, name=x1
    [{}, edge label={node[font=\small, midway,left]{$0$}}, name=x2L
        [{}, edge label={node[font=\small,midway,left]{$0$}}, name=y100
            [{}, edge label={node[font=\small,midway,left]{$1$}}, name=z100
                [{}, edge label={node[font=\small,midway,left]{$1$}}, name=y200
                    [{}, edge label={node[font=\small,midway,left]{$0$}}, name=e00]
                ]
            ]
        ]
        [{}, edge label={node[font=\small,midway,right]{$1$}}, name=y101
            [{}, edge label={node[font=\small,midway,right]{$1$}}, name=z101
                [{}, edge label={node[font=\small,midway,right]{$1$}}, name=y201
                    [{}, edge label={node[font=\small,midway,right]{$1$}}, name=e01]
                ]
            ]
        ]
    ]
    [{}, edge label={node[font=\small, midway,right]{$1$}}, name=x2R
        [{}, edge label={node[font=\small,midway,left]{$0$}}, name=y110
            [{}, edge label={node[font=\small,midway,left]{$1$}}, name=z110
                [{}, edge label={node[font=\small,midway,left]{$1$}}, name=y210
                    [{}, edge label={node[font=\small,midway,left]{$0$}}, name=e10]
                ]
            ]
        ]
        [{}, edge label={node[font=\small,midway,right]{$1$}}, name=y111
            [{}, edge label={node[font=\small,midway,right]{$1$}}, name=z111
                [{}, edge label={node[font=\small,midway,right]{$1$}}, name=y211
                    [{}, edge label={node[font=\small,midway,right]{$1$}}, name=e11]
                ]
            ]
        ]
    ]
]
]
\node [above=3pt of x1] {$s^{g,U}$};
\node (p) [left=15pt of y200] {$y_3$:};
\node at (p |- z100) {$y_2$:};
\node at (p |- y100) {$y_1$:};
\node at (p |- x2L) {$x_2$:};
\node at (p |- x1) {$x_1$:};
\path (y100) edge[dashed, bend right=20] (y101);
\path (y110) edge[dashed, bend right=20] (y111);
\path (z100) edge[dashed, bend right=20] (z101);
\path (z110) edge[dashed, bend right=20] (z111);
\path (y200) edge[dashed, bend right=20] (y210);
\path (y201) edge[dashed, bend right=20] (y211);
\begin{scope}[on background layer]
\node[
  draw=blue,
  fill=blue,
  fill opacity=0.15,
  draw opacity=1,
  very thick,
  rounded corners,
  inner sep=6pt,
  fit=(y100)(e01)
] {};
\node[
  draw=red,
  fill=red,
  fill opacity=0.15,
  draw opacity=1,
  very thick,
  rounded corners,
  inner sep=6pt,
  fit=(y110)(e11)
] {};
\node[
  draw=blue,
  fill=blue,
  fill opacity=0.15,
  rounded corners,
  inner sep=5pt,
  anchor=west
] (legg) at (3,-2) {};
\node[right=4pt of legg] {$g(s)$};
\node[
  draw=red,
  fill=red,
  fill opacity=0.15,
  rounded corners,
  inner sep=5pt,
  anchor=west
] (legi) at (3,-2.5) {};
\node[right=4pt of legi] {$s$};
\end{scope}
\end{forest}
\end{center}
\caption{The interpretation $s^{g,U}=g(s)\sqcup_U s$.}
\label{fig:local_sem_example}
\Description{An assignment tree for the glued interpretation. The branches with x1 equal to zero use g(s), and those with x1 equal to one use s. The resulting Skolem functions are y1 equals one, y2 equals one, and y3 equals x2.}
\end{figure}
The resulting semantic symmetry $\mathbb{S}(P) \to \mathbb{S}(P)$ applies $g$ to the left part of the tree of $s$, and is the identity on the right part of $s$ as well as on all other interpretations $t \in \mathbb{S}(P)$. As such, it applies the symmetry $g$ `locally'. 
\end{ex}

\section{Lex-leader Symmetry Breaking} \label{sect:lex_leader}

In this section we use the algebraic framework developed above to establish a general lex-leader symmetry-breaking construction for DQBFs, which strictly generalises the corresponding QBF lex-leader construction in \cite{KS18}.

\subsection{Symmetry Breakers}\label{sec:breakers}
Let $\Phi = P.\varphi$ be a DQBF and let $G_{\sem}$ be a semantic symmetry group of $\Phi$. As recalled in Section~\ref{prelim:groups}, the action of $G_{\sem}$ partitions $\mathbb S(P)$ into orbits, and by Definition~\ref{def:sem_symm} each orbit consists either entirely of models of $\Phi$ or entirely of non-models:
\begin{center}
\begin{tikzpicture}
	\path[fill=gray,opacity=0.7] (0.8,1) -- (0,1.4) -- (0,0) -- (0.7,0) -- cycle;
	\path[fill=gray,opacity=0.7] (2.2,0) -- (3,2) -- (4,0) -- cycle;
	\path[fill=gray,opacity=0.7] (1.5,2) -- (1.84,1.45) -- (2.1,2) -- cycle;
	\draw[thick,draw] (0,0) rectangle ++(4.5,2);
	\draw[thick,draw] (0.7,0) -- (0.9,2);
	\draw[thick,draw] (0.8,1) -- (0,1.4);
	\draw[thick,draw] (1.2,0) -- (2.1,2);
	\draw[thick,draw] (2.2,0) -- (3,2);
	\draw[thick,draw] (2.4,0.5) -- (1.5,2);
	\draw[thick,draw] (4,0) -- (3,2);
	\draw[thick,draw] (3.5,1) -- (4.5,0.8);
	\node[align=right,anchor=east] at (-0.7,1.6) {an orbit containing\\no models at all};
	\draw[-*] (-0.7,1.6) -- (0.5,1.6);
	\node[align=right,anchor=east] at (-0.7,0.5) {an orbit containing\\only models};
	\draw[-*] (-0.7,0.5) -- (0.5,0.5);
\end{tikzpicture}
\end{center}
For deciding whether $\Phi$ has a model it therefore suffices to inspect a single interpretation from each orbit. Symmetry breaking exploits this observation by conjoining to the matrix a formula $\psi$ that excludes part of every orbit while retaining at least one of its elements.

\begin{defi}[Conjunctive symmetry breaker]\label{def:breaker}
Let $P$ be a prefix, let $G\leq\Sym(\Sk(P))$, and let $\mathcal K\subseteq\Sk(P)$ be $G$-invariant. A formula $\psi\in\B(X\cup Y)$ is a \emph{$\mathcal K$-conjunctive symmetry breaker} with respect to $G$ if every $G$-orbit contained in $\mathcal K$ contains an interpretation satisfying $P.\psi$, that is, if for every $s\in\mathcal K$ there exists $f\in G$ with $[P.\psi]_{f(s)}=1$. For $\mathcal K=\Sk(P)$ we simply speak of a \emph{conjunctive symmetry breaker}.
\end{defi}

Thus for a DQBF $\Phi = P.\varphi$ with semantic symmetry group $G_{\sem}$, $\psi$ is a $\mathcal{K}$-conjunctive symmetry breaker precisely when every $G_{\sem}$-orbit contained in $\mathcal{K}$ contains an interpretation satisfying $P.\psi$. In this section, we are only concerned with the case $\mathcal{K} = \mathbb{S}(P)$. However, in Section \ref{sec:canonisers} we will use the flexibility in the definition to construct $\mathcal{K}$-conjunctive symmetry breakers for other $\mathcal{K}$. The following proposition is the main property of symmetry breakers: instead of solving $P.\varphi$, we may solve $P.(\varphi \wedge \psi)$.

\begin{prop}[Truth preservation]\label{prop:breaker-truth}
Let $\Phi=P.\varphi$ be a DQBF, let $G_{\sem}$ be a semantic symmetry group of $\Phi$, and let $\mathcal K \supseteq \Mod(\Phi)$ be $G_{\sem}$-invariant. If $\psi$ is a $\mathcal K$-conjunctive symmetry breaker with respect to $G_{\sem}$, then
\[
P.\varphi\text{ is true}\iff P.(\varphi\wedge\psi)\text{ is true}.
\]
\end{prop}
\begin{proof}
The right-to-left implication is immediate from Lemma~\ref{lem:eval_conjunction}. Conversely, let $s$ be a model of $P.\varphi$. Then $s \in \mathcal{K}$, so there exists $f \in G_{\sem}$ with $[P.\psi]_{f(s)} = 1$. Since $f$ is a semantic symmetry, $[P.\varphi]_{f(s)} = [P.\varphi]_s = 1$. By Lemma~\ref{lem:eval_conjunction},
\[
  [P.(\varphi \wedge \psi)]_{f(s)}
  = [P.\varphi]_{f(s)} \wedge [P.\psi]_{f(s)} = 1,
\]
so $f(s)$ is a model of $P.(\varphi \wedge \psi)$.
\end{proof}

Ideally, $\psi$ selects precisely one interpretation from each $G_{\sem}$-orbit. Then $P.(\varphi \wedge \psi)$ contains exactly one model from each orbit containing models of $P.\varphi$, thereby restricting the search to a single representative of each such orbit. One standard technique for choosing a canonical representative from each orbit is to impose a lexicographic order on the set of interpretations $\mathbb{S}(P)$, and to construct a symmetry breaker $\psi$ that selects the least interpretation in this order from every orbit. This is called the \emph{lex-leader} of the orbit, and therefore a lex-leader symmetry breaker $\psi$ tries to enforce
\begin{equation} \label{eq:min_lex}
s \leq_{\lex} g(s), \quad \textrm{for all relevant } g \in G_{\sem}.
\end{equation}

In the remainder of this section we use the framework of Section \ref{sec:symm} to develop such a construction for DQBF (Theorem~\ref{thm:symm_breaker}). A conjunctive symmetry breaker is \emph{complete} if it retains exactly one interpretation from each semantic orbit, and \emph{partial} otherwise. The construction of Theorem~\ref{thm:symm_breaker} is not necessarily complete: every semantic orbit minimum satisfies the constructed predicate, but other interpretations in the same orbit may satisfy it as well. This suffices to preserve truth while reducing symmetry.

\subsection{Lexicographic Order on Interpretations} \label{sec:lex_order_int}
Let $P = \forall X\;\exists y_1(D_1)\cdots\exists y_k(D_k)$ be a topologically ordered prefix, that is,
\[
D_i\subsetneq D_j\quad\Longrightarrow\quad i<j.
\]
Note that this entails no loss of generality since every DQBF is equivalent to one with a topologically ordered prefix, which is obtained simply by reordering the existential quantifiers together with their dependency annotations.
Fix a linear order $x_1 < \cdots < x_n$ on the universal variables, and let $<_{\lex}$ denote the induced lexicographic order on $\mathbb{A}(X)$. We define a lexicographic order on $\mathbb{S}(P)$ as follows. For distinct interpretations $s, t \in \mathbb{S}(P)$, let $\sigma \in \mathbb{A}(X)$ be the $<_{\mathrm{lex}}$-least assignment for which
\[
\sigma_s(y_i) \ne \sigma_t(y_i) \qquad \text{for some } i \in \{1,\ldots,k\}.
\]
Let $j$ be the least index for which $\sigma_s(y_j) \ne \sigma_t(y_j)$. We define
\[
s<t
\quad \Longleftrightarrow \quad
\sigma_s(y_j)<\sigma_t(y_j).
\]
Thus we first select the least universal assignment at which the interpretations differ, and then compare the first differing existential value at that assignment. 

The main challenge of this section is to derive succinct syntactic symmetry-breaking predicates that express the semantic minimality principle \eqref{eq:min_lex} for the above order. For selected combinatorial encodings, in Section \ref{sec:canonisers} we then strengthen the resulting predicates with specialised succinct predicates that break symmetries based on alternative canonisation functions. 

For a syntactic symmetry $g\in G_{\syn}$ and $j\in\{1,\ldots,k\}$, define
\[
\psi_{j,g}
:=
\left(
\bigwedge_{x\in D_j}(x\leftrightarrow g(x))
\wedge
\bigwedge_{i<j}(y_i\leftrightarrow g(y_i))
\right)
\to
(y_j\to g(y_j)).
\]
Intuitively, $\psi_{j,g}$ expresses the lex-leader condition for $g$ locally at $y_j$: whenever $g$ agrees with the current assignment on $D_j$ and on all preceding existential variables, the value of $y_j$ must not be larger than the value obtained after applying $g$. However, not all syntactic symmetries may be used at every existential variable $y_j$ for our lex-leader symmetry-breaking construction; we need the following restriction. 

\begin{defi}[Lex-permissible pair] \label{def:lex_perm} Let
\[
P=\forall X\;\exists y_1(D_1)\cdots\exists y_k(D_k)
\]
be a topologically ordered prefix, and let \(G_{\mathrm{syn}}\) be a syntactic symmetry group of \(\Phi=P.\varphi\). A pair
\[
(j,g) \in [k] \times G_{\mathrm{syn}}
\]
is \emph{lex-permissible} if $\psi_{j,g} = \top$, or if both of the following conditions hold:
\begin{enumerate}[label=(\roman*)]
\item $g(x) \in \B(D_j)$ for every $x \in D_j$;
\item For every \(i\in\{1,\ldots,k\}\) with \(D_i\parallel D_j\),
\[
g(y_i) = y_i
\qquad \text{and} \qquad 
g(x) = x \quad \text{for every } x\in D_i \setminus D_j.
\]
\end{enumerate}
We call a symmetry $g \in G_{\syn}$ lex-permissible if $(j,g)$ is lex-permissible for every $j \in [k]$.
\end{defi}
Intuitively, lex-permissibility ensures that the local lex-leader transformation respects the DQBF dependency structure: condition~(i) keeps the action of $g$ inside each dependency set, while condition~(ii) prevents dependency clusters induced by incomparable dependency sets from being cut inconsistently: if the breaker at $y_j$ is non-trivial, then $g$ must act trivially on those parts of the dependency set $D_i$ that lie outside $D_j$. This condition is illustrated in Figure \ref{fig:nonlinear_dep_condition}. Pairs with $\psi_{j,g} = \top$ are automatically permitted since they impose no constraint.

\begin{figure}[H] 
\begin{center}
\scalebox{0.75}{
\begin{tikzpicture}[
    font=\small,
    >=latex,
    every node/.style={align=center}
]

% --- filled regions ---
\fill[green!12, rounded corners=8pt] (0,0) rectangle (4.6,2);
\fill[blue!12, rounded corners=8pt]  (2.4,0) rectangle (7,2);
\fill[teal!12, rounded corners=0pt]  (2.4,0) rectangle (4.6,2);

% --- outlines ---
% left region
\draw[green!50!black, thick, rounded corners=8pt]
    (0,0) rectangle (4.6,2);

% right region
\draw[blue!60!black, thick, rounded corners=8pt]
    (2.4,0) rectangle (7,2);

% dotted border over the overlap
\draw[dotted, thick] (2.4,2) -- (4.6,2);
\draw[dotted, thick] (2.4,0) -- (4.6,0);

% --- set labels ---
\node[green!50!black] at (0.5,2.25) {$D_j$};
\node[blue!60!black]  at (6.5,2.25) {$D_i$};

% --- region labels ---
\node[green!50!black] at (1.2,1.55) {$D_j\setminus D_i$};
\node at (3.5,1.55) {$D_i \cap D_j$};
\node[blue!60!black] at (5.8,1.55) {$D_i\setminus D_j$};

% --- sample variables ---
\node[circle, draw, inner sep=1pt] (a) at (1.2,0.7) {$a$};
\node[circle, draw, inner sep=1pt] (c) at (3.5,0.7) {$b$};
\node[circle, draw, inner sep=1pt] (d) at (5.8,0.7) {$c$};

% --- y_i moved ---
\node[circle, draw=red, text=red, thick, inner sep=2pt] (yi) at (1.4,3.2) {$y_j$};
\node[text=red] (gyi) at (3.6,3.2) {$g(y_j)\neq y_j$};
\draw[->, red, thick] (yi) -- (gyi);

% --- y_j fixed ---
\node[circle, draw, inner sep=2pt] (yj) at (5.7,3.2) {$y_i$};
\node[right=4pt of yj] {$g(y_i)=y_i$};

% --- dependency hints ---
\draw[dashed, gray] (yi) -- (1.4,2);
\draw[dashed, gray] (yj) -- (5.7,2);

% --- fixed-part annotation ---
\node[right=1cm of d] {$g(x)=x$ for all\\$x\in D_i\setminus D_j$};

\end{tikzpicture}
}
    \caption{An illustration of condition (ii).}
    \label{fig:nonlinear_dep_condition}
\end{center}
\Description{Two overlapping dependency sets, Dj and Di. The symmetry changes yj and fixes yi; it also fixes every universal variable in Di outside Dj.}
\end{figure}
This ensures that the gluing construction from Lemma~\ref{lem:gluing} is applicable. Recall from Section \ref{prelim:DQBF} that for $\sigma \in \mathbb{A}(X)$ and $j \in \{1,\ldots,k\}$ we define the dependency cluster of $\sigma$ with respect to $D_j$ by $U_{\sigma,j} := \{\tau \in \mathbb A(X) \mid \tau|_{D_j}=\sigma|_{D_j}\}$. 

\begin{lemma}\label{lem:lex_gluing}
Let $(j,g) \in [k] \times G_{\syn}$ be a lex-permissible pair and suppose that $[\psi_{j,g}]_{\sigma_s} = 0$ for some $s \in \mathbb{S}(P)$ and $\sigma \in \mathbb{A}(X)$. Then $U_{\sigma,j}$ is $g$-invariant, and $g(s)$ and $s$ satisfy the gluing condition with respect to $U_{\sigma,j}$.
\end{lemma}

\begin{proof}
Since $[\psi_{j,g}]_{\sigma_s}=0$, we have $[g(x)]_\sigma=\sigma(x)$ for every $x\in D_j$. Let $\tau\in U_{\sigma,j}$. By lex-permissibility, $g(x) \in \B(D_j)$ for every $x\in D_j$, and hence
\[
\hat g(\tau)(x)=[g(x)]_\tau=[g(x)]_\sigma=\sigma(x)=\tau(x).
\]
Thus $\hat g(\tau)|_{D_j}=\sigma|_{D_j}$, so $\hat g(\tau)\in U_{\sigma,j}$. Since $U_{\sigma,j}$ is finite and $\hat g$ is bijective, $\hat g(U_{\sigma,j})=U_{\sigma,j}$.

It remains to verify the gluing condition. Let $\tau\in U_{\sigma,j}$ and $\tau'\notin U_{\sigma,j}$ satisfy $\tau|_{D_i}=\tau'|_{D_i}$. We show that
\[
\tau_{g(s)}(y_i)=\tau'_s(y_i).
\]
If $D_j\subseteq D_i$, then $\tau|_{D_i}=\tau'|_{D_i}$ implies $\tau|_{D_j}=\tau'|_{D_j}$, contradicting $\tau'\notin U_{\sigma,j}$. Hence this case cannot occur.

Suppose $D_i\subsetneq D_j$. Since the prefix is topologically ordered, $i<j$. Moreover, $\tau,\sigma$ agree on $D_i$, and $\hat g(\sigma)$ agrees with $\sigma$ on $D_i$. Using the dependency condition and Proposition~\ref{prop:induced-semantic-symmetry}, we obtain
\[
\tau_{g(s)}(y_i)=\sigma_{g(s)}(y_i)=\hat g(\sigma)_{g(s)}(y_i)=[g(y_i)]_{\sigma_s}.
\]
Since $[\psi_{j,g}]_{\sigma_s}=0$, its antecedent is true, and therefore $[g(y_i)]_{\sigma_s}=\sigma_s(y_i)$. As $\tau'$ and $\sigma$ agree on $D_i$, it follows that
\[
\tau_{g(s)}(y_i)=\sigma_s(y_i)=\tau'_s(y_i).
\]

Finally, suppose $D_i\parallel D_j$. Since $\psi_{j,g} \neq \top$, lex-permissibility gives
\[
g(y_i)=y_i,
\qquad
g(x)=x \quad\text{for all }x\in D_i\setminus D_j.
\]
For $x\in D_i\cap D_j$, the same argument as above gives $\hat g(\tau)(x)=\tau(x)$, while this is immediate for $x\in D_i\setminus D_j$. Hence $\hat g(\tau)|_{D_i}=\tau|_{D_i}$. Therefore, using the dependency condition for $g(s)$ and Proposition~\ref{prop:induced-semantic-symmetry},
\[
\tau_{g(s)}(y_i)=\hat g(\tau)_{g(s)}(y_i)=[g(y_i)]_{\tau_s}
=\tau_s(y_i)=\tau'_s(y_i).
\]
Thus $g(s)$ and $s$ satisfy the gluing condition with respect to $U_{\sigma,j}$.
\end{proof}

\begin{lemma}[Lexicographic decrease]\label{lem:lex_decrease}
Let $(j, g) \in [k] \times G_{\syn}$ be lex-permissible and suppose that
$[\psi_{j,g}]_{\sigma_s} = 0$ for some $s \in \mathbb S(P)$ and $\sigma \in \mathbb A(X)$. Put 
\[
s':=g(s)\sqcup_{U_{\sigma,j}}s,
\]
then $s' < s$.
\end{lemma}
\begin{proof}
By Lemma~\ref{lem:lex_gluing} and Lemma \ref{lem:gluing}, $s'$ is well-defined. Since $[\psi_{j,g}]_{\sigma_s}=0$, we have
\[
[g(x)]_\sigma=\sigma(x)\quad(x\in D_j),\qquad
[g(y_i)]_{\sigma_s}=\sigma_s(y_i)\quad(i<j),
\]
and
\[
\sigma_s(y_j)=1,\qquad [g(y_j)]_{\sigma_s}=0.
\]

We first show that $s'_i=s_i$ for every $i<j$. Let $\tau\in U_{\sigma,j}$ (outside $U_{\sigma,j}$ this holds by construction). Since the prefix is topologically ordered, either $D_i\subseteq D_j$ or $D_i\parallel D_j$.

If $D_i\subseteq D_j$, then $\tau$ and $\sigma$ agree on $D_i$. Moreover, $\hat g(\sigma)$ and $\sigma$ agree on $D_j$, and hence on
$D_i$. Using the dependency condition and Proposition~\ref{prop:induced-semantic-symmetry},
\[
\tau_{g(s)}(y_i)
=\sigma_{g(s)}(y_i)
=\hat g(\sigma)_{g(s)}(y_i)
=[g(y_i)]_{\sigma_s}
=\sigma_s(y_i)
=\tau_s(y_i).
\]
If $D_i\parallel D_j$, lex-permissibility gives $g(y_i)=y_i$ and $g(x)=x$ for every $x\in D_i\setminus D_j$. As in the proof of Lemma~\ref{lem:lex_gluing}, this implies $\hat g(\tau)|_{D_i}=\tau|_{D_i}$, and therefore
\[
\tau_{g(s)}(y_i)
=\hat g(\tau)_{g(s)}(y_i)
=[g(y_i)]_{\tau_s}
=\tau_s(y_i).
\]
Thus $s'_i=s_i$ for every $i<j$.

For $y_j$ we have, since $\hat g(\sigma)$ and $\sigma$ agree on $D_j$,
\[
\sigma_{s'}(y_j)
=\sigma_{g(s)}(y_j)
=\hat g(\sigma)_{g(s)}(y_j)
=[g(y_j)]_{\sigma_s}
=0,
\]
whereas $\sigma_s(y_j) = 1$. Now let $\rho$ be the $<_{\mathrm{lex}}$-least assignment in $U_{\sigma,j}$. Since $s'$ and $s$ agree outside $U_{\sigma,j}$, they agree at every universal assignment preceding $\rho$. Moreover, at $\rho$ we have
\[
\rho_{s'}(y_i)=\rho_s(y_i)
\qquad \text{for every } i < j,
\]
whereas
\[
\rho_{s'}(y_j) = \sigma_{s'}(y_j) = 0 < 1 = \sigma_s(y_j) = \rho_s(y_j).
\]
Here the first and last equality follow from the dependency condition for $s'$ and $s$ respectively, combined with the fact that $\rho \in U_{\sigma,j}$. Thus $\rho$ is the first universal assignment at which the interpretations differ, and $y_j$ is the first existential variable at which they differ there. It follows that $s' < s$.
\end{proof}

\begin{cor} \label{cor: lex_decrease}
Under the assumptions of Lemma~\ref{lem:lex_decrease}, there exists $f\in G_{\sem}$ such that $f(s)<s$. 
\end{cor}

\begin{proof}
By Lemma~\ref{lem:lex_gluing} and Proposition~\ref{prop:local_sem_sym},
\[
f:=\bigl(s\;\;g(s)\sqcup_{U_{\sigma,j}}s\bigr)\in G_{\sem}.
\]
The claim then follows from Lemma~\ref{lem:lex_decrease}.
\end{proof}

\subsection{Local Lex-leader Symmetry Breaker} 
We have now developed all the ingredients to prove a local lex-leader symmetry breaking construction for DQBF. 

\begin{thm}[Local lex-leader symmetry breaker]
\label{thm:symm_breaker}
Let
\[
P= \forall X\,\exists y_1(D_1)\cdots\exists y_k(D_k)
\]
be a topologically ordered DQBF prefix, let $\Phi=P.\varphi$, and let $G_{\syn}$ be a syntactic symmetry group of $\Phi$ with associated semantic symmetry group $G_{\sem}$. Let $\Lambda \subseteq [k] \times G_{\syn}$ be a set of lex-permissible pairs. For each $(j,g) \in \Lambda$, let
\[
\psi_{j,g}
:=
\left(
\bigwedge_{x\in D_j}(x\leftrightarrow g(x))
\wedge
\bigwedge_{i<j}(y_i\leftrightarrow g(y_i))
\right)
\to
(y_j\to g(y_j)).
\]
Then
\[
\psi
:=
\bigwedge_{(j,g)\in\Lambda}\psi_{j,g}
\]
is a conjunctive symmetry breaker with respect to $G_{\sem}$, and admits a CNF encoding of size $O(r(|X|+|Y|) + S)$, where $r$ is the number of distinct selected symmetries and $S$ is the total number of gates in circuits representing their required images $g(v)$ (see Appendix~\ref{app:lex-cnf}).
\end{thm}
\begin{proof}
Let $s_0 \in \mathbb S(P)$. Since $\mathbb S(P)$ is finite, the $G_{\sem}$-orbit of $s_0$ contains a minimal interpretation $s$ with respect to the lexicographic order defined in Section \ref{sec:lex_order_int}. We claim that $[P.\psi]_s = 1$. Suppose otherwise. Then there exists $\sigma\in\mathbb A(X)$ such that $[\psi]_{\sigma_s}=0$, and hence there exists $(j,g) \in \Lambda$ such that $[\psi_{j,g}]_{\sigma_s} = 0$. By Corollary~\ref{cor: lex_decrease}, there exists $f \in G_{\sem}$ such that $f(s) < s$. But $f(s)$ belongs to the same $G_{\sem}$-orbit as $s$, contradicting the minimality of $s$. Hence $[P.\psi]_s=1$.

Since $s$ lies in the orbit of the arbitrary interpretation $s_0$, there exists $h \in G_{\sem}$ with $h(s_0)=s$. Therefore $[P.\psi]_{h(s_0)}=1$, which shows that $\psi$ is a conjunctive symmetry breaker with respect to $G_{\sem}$. For the CNF encoding and size bound we refer to Appendix~\ref{app:lex-cnf}. 
\end{proof}

\begin{rem}
For any subset $\Gamma \subseteq G_{\syn}$, we may take
\[
\Lambda
=
\bigl\{
(j,g)\in[k]\times\Gamma
\mid (j,g)\text{ is lex-permissible}
\bigr\}.
\]
The set $\Gamma$ need not be a subgroup, and its elements need not be lex-permissible as a whole: each symmetry contributes precisely the constraints at its permissible indices. If every $g\in\Gamma$ is lex-permissible, then $\Lambda=[k]\times\Gamma$. Taking $\Gamma = G_{\syn}$ includes all lex-permissible pairs and yields the full local construction. In practice, one may instead use a generating set or another suitable collection of symmetries.
\end{rem}

\begin{ex}[Continuation of Example~\ref{ex:mixing}]
Recall the DQBF $\Phi$ and symmetry $g$ from Example~\ref{ex:mixing}. The symmetry $g$ is lex-permissible. For $\Gamma=\{g\}$, the construction of Theorem~\ref{thm:symm_breaker} gives
\[
\psi_{1,g} = \neg y_1,
\qquad
\psi_{2,g} = \psi_{3,g} = \top.
\]
Hence $\psi = \neg y_1$. Thus the breaker selects interpretations in which the first Skolem function is constantly $0$. The local semantic symmetries of Proposition~\ref{prop:local_sem_sym} allow $g$ to be applied to a single $D_1$-dependency cluster rather than uniformly to the whole interpretation. For instance, applying $g$ only on the cluster $x_1=1$ transforms $s=(x_1,x_1,x_2)$ into $(0,0,x_2)$, which satisfies the breaker $\psi = \neg y_1$.
\end{ex}

The following example shows that condition (ii) of lex-permissibility cannot in general be omitted from the lex-leader construction. 

\begin{ex}[Continuation of Example~\ref{ex1}] \label{ex:nec_condition}
Consider the DQBF of Example \ref{ex1},
\[
\Phi = P.\varphi = \forall x_1,x_2,x_3 \exists y_1 (\{x_2\}), y_2(\{x_3\}).((x_1 \leftrightarrow y_1) \vee (x_2 \leftrightarrow y_2) \vee (x_1 \leftrightarrow x_3)), 
\]
with the model $s = (x_2,x_3)$, which is illustrated in Figure \ref{fig:model_exmp_mix1}. A syntactic symmetry $g$ of $\Phi$ is given by $g(x_2) = \neg x_2$ and $g(y_2) = \neg y_2$, while $g$ is the identity on other variables. The syntactic symmetry $g$ induces a semantic symmetry $g : \mathbb{S}(P) \to \mathbb{S}(P)$ which complements both Skolem functions in $s$, $g(s) = (\neg x_2, \neg x_3)$. One can verify that $s$ and $g(s)$ are the only models of $\Phi$, so that $\Mod(\Phi) = \{s, g(s)\}$. For this symmetry, $\psi_{1,g} = \top$, while
\[
\psi_{2,g}
=
\bigl(
(x_3\leftrightarrow g(x_3))
\land
(y_1\leftrightarrow g(y_1))
\bigr)
\to
(y_2\to g(y_2))
= \neg y_2.
\]
The pair $(2,g)$ satisfies condition~(i) of Definition \ref{def:lex_perm}, since $g(x_3)=x_3 \in \B(D_2)$. However, it violates condition~(ii): we have $D_1 \parallel D_2$ and $x_2 \in D_1 \setminus D_2$, but $g(x_2)=\neg x_2 \ne x_2$.

If condition~(ii) were omitted, the pair $(2,g)$ would be lex-permissible, and taking $\Lambda=\{(2,g)\}$ would yield the predicate
\[
\psi = \psi_{2,g} = \neg y_2.
\]
However, $P.(\varphi\land\neg y_2)$ is false, since neither of the two models of $\Phi$ satisfies $P.\neg y_2$. Hence $\psi$ is not a conjunctive symmetry breaker by Proposition~\ref{prop:breaker-truth}. This shows that condition~(ii) cannot in general be omitted.
\end{ex}

\begin{ex}[Symmetries of the succinct pigeonhole principle]
\label{ex:dphp_symmetries}
Recall $\DPHP_n$ from Section~\ref{sect:succinct}, and put $\mathcal A=\{0,1\}^n$. For $\pi,\rho\in\Sym(\mathcal A)$, consider
\[
  (X,Y,H,F,F')
  \mapsto
  \bigl(\pi(X),\pi(Y),\rho(H),\rho(F),\rho(F')\bigr),
\]
where tuple maps abbreviate their coordinate Boolean formulas. The induced maps on assignments are bijective, so by Corollary~\ref{cor:hat}(ii) the associated endomorphisms on $\B(V)$ are automorphisms. Moreover, the maps are prefix-admissible and preserve the matrix, whence by Remark \ref{rem:generators} they generate a syntactic symmetry group. Exchanging the copies $(X,F) \leftrightarrow (Y,F')$ is also a symmetry. Together, these operations generate a group
\[
  G_{\mathrm{gen}}
  \cong
  \Sym(\mathcal A)\times\Sym(\mathcal A)\times C_2
\]
of size $|G_{\mathrm{gen}}| = 2((2^n)!)^2$ whose factors permute pigeons, permute holes, and exchange copies. Its literal subgroup is
\[
  G_{\mathrm{lit}} \cong B_n \times B_n \times C_2,
  \qquad B_n:=C_2^n\rtimes S_n
\]
of size $|G_{\lit}| = 2(2^n n!)^2$, where $B_n$ consists of coordinate permutations and bit-flips. Generalised symmetries therefore capture arbitrary pigeon and hole permutations, whereas literal symmetries are restricted to these signed coordinate permutations. For example, when $n=2$, the transposition $(10\;11)$ is represented by $\pi(a_1,a_2)=(a_1,a_1\oplus a_2)$, which corresponds to a generalised syntactic symmetry 
\[
x_1\mapsto x_1,\qquad x_2\mapsto x_1\oplus x_2,
\qquad
y_1\mapsto y_1,\qquad y_2\mapsto y_1\oplus y_2. 
\]
We have $|G_{\mathrm{gen}}| = 1152$ and $|G_{\lit}| = 128$ for $n = 2$. This illustrates the role of generalised syntactic symmetries: literal transformations can only permute and complement its bits, whereas generalised transformations realise arbitrary permutations of its $2^n$ addresses.

Applying Theorem~\ref{thm:symm_breaker} to all lex-permissible pairs in either group gives exactly
\[
\psi_{\loc} := \bigwedge_{i=1}^n \neg h_i,
\]
normalising the missing hole to $H=0^n$, which is obtained from the literal syntactic symmetry group generated by the output bit-flips
\[
h_i \mapsto \neg h_i, \quad f_i \mapsto \neg f_i, \quad f'_i \mapsto \neg f'_i. 
\]
Thus our generic lex-leader construction leaves substantial symmetry unbroken. In Section \ref{sec:canonisers} we show that specialised problem-specific predicates can break additional symmetries in $\DPHP$. 
\end{ex}

\subsection{Tightness of the Lex-leader Construction}
\label{sec:tightness}

Lex-permissibility is restrictive: condition~(ii) of Definition~\ref{def:lex_perm} forces $g$ to act trivially on every Skolem function whose dependency set is incomparable with $D_j$, and for $\DPHP_n$ this leaves only the normalisation $H=0^n$ (Example~\ref{ex:dphp_symmetries}). Whether $\psi_{j,g}$ is a conjunctive symmetry breaker with respect to $G_{\sem}$ is a semantic question about the orbits of $G_{\sem}$ (Definition~\ref{def:breaker}). Theorem~\ref{thm:symm_breaker} gives a syntactic sufficient condition, and Example~\ref{ex:nec_condition} shows that condition~(ii) cannot simply be dropped. The following example shows that lex-permissibility is not necessary.

\begin{ex}[Lex-permissibility is not necessary]\label{ex:not-necessary}
Let 
\[
P=\forall x_1,x_2,x_3\;\exists y_1(\{x_1\})\;\exists y_2(\{x_1,x_2\})\;\exists y_3(\{x_1,x_3\}),
\]
and let $g$ act by 
$
y_2 \mapsto y_2\oplus y_1, y_3 \mapsto \neg y_3,
$
while fixing the other variables. Then $g$ is prefix-admissible and
\[
\psi_{3,g}=\bigl((y_1\leftrightarrow y_1)\wedge(y_2\leftrightarrow y_2\oplus y_1)\bigr)\to(y_3\to\neg y_3)=y_1\vee\neg y_3\neq\top .
\]
The pair $(3,g)$ violates condition~(ii) of Definition~\ref{def:lex_perm}, since $D_2\parallel D_3$ and $g(y_2)\neq y_2$. Nevertheless, $\psi_{3,g}$ is a conjunctive symmetry breaker with respect to $G_{\sem}$ for $\Gsyn=\langle g\rangle$. To see this, let $s=(s_1,s_2,s_3)\in\Sk(P)$ and put $t:=(s_1,\ s_2,\ s_3\wedge s_1)\in\Sk(P)$. Then $t$ satisfies $P.(y_1\vee\neg y_3)$, since $t_3=0$ wherever $s_1=0$. We claim that the transposition $f:=(s\;\;t)$ lies in $G_{\sem}$. As $g$ fixes $X$, we have $\hat g(\sigma)=\sigma$, and the condition of Definition~\ref{def:ass_sem_group} for $f$ at $s$ reads: for every $\sigma$, $\sigma_t=\sigma_s$ or $\sigma_t=\hat g(\sigma_s)$. If $s_1(\sigma)=1$ or $s_3(\sigma)=0$, then $\sigma_t=\sigma_s$. Otherwise $s_1(\sigma)=0$ and $s_3(\sigma)=1$, so
\[
\hat g(\sigma_s)=\bigl(\sigma,\ s_1,\ s_2\oplus s_1,\ \neg s_3\bigr)=\bigl(\sigma,\ s_1,\ s_2,\ 0\bigr)=\sigma_t .
\]
The condition for $f$ at $t$ follows by symmetry, since $g$ is an involution and so $\hat g(\sigma_t)=\sigma_s$ in the second case; at every other interpretation $f$ is the identity. Hence $f\in G_{\sem}$ and $f(s)=t$, as required. 
\end{ex}

\begin{rem}[Open question]\label{rem:open}
We do not know a syntactic characterisation of the pairs $(j,g)$ for which $\psi_{j,g}$ is a conjunctive symmetry breaker with respect to $G_{\sem}$. 
\end{rem}

\subsection{Recovering the QBF Construction}\label{sect:QBF_SB}
For DQBFs with a linearly ordered dependency structure there are no incomparable dependency sets, so condition~(ii) of Definition~\ref{def:lex_perm} is vacuous. In this case our construction recovers the QBF symmetry breaker of \cite[Theorem~22]{KS18}. We make this precise. Let
\[
\Phi = Q_1 v_1 \cdots Q_n v_n .\, \varphi
\]
be a QBF with universal variables $X$ and existential variables $Y$, and let
\[
\Phi^{\dep} = \forall X\, \exists y_1(D_1) \cdots \exists y_k(D_k) .\, \varphi
\]
be its associated DQBF (Section~\ref{prelim:DQBF}), where $y_1,\dots,y_k$ are enumerated in their original prefix order. The dependency sets are then linearly ordered by inclusion, and the prefix is topologically ordered.

\paragraph{Symmetries in the sense of~\cite{KS18}.}
In~\cite{KS18}, formulas are not identified up to equivalence. A syntactic symmetry of $\Phi$ is a bijection
$f\colon \BF(V) \to \BF(V)$, $V = X \cup Y$, such that
\begin{enumerate}[label=(\arabic*)]
\item $\eval{f(\varphi)}{\sigma} = \eval{\varphi}{\hat f(\sigma)}$ for every $\varphi \in \BF(V)$ and $\sigma \in \A(V)$; 
\item for every $v \in V$, the formula $f(v)$ contains only variables from the quantifier block of $v$;
\item $f(\varphi) \equiv \varphi$;
\end{enumerate}
see \cite[Definitions~3 and~7]{KS18}. It follows from Condition~(1) that $f$ descends to a map $\bar f \colon \B(V) \to \B(V)$, $\bar f(\psi) := f(\psi)$, which is an automorphism of $\B(V)$. As noted in~\cite{KS18}, an admissible function is ``essentially determined by its values for the variables''. Since $\overline{f \circ f'} = \bar f \circ \bar f'$, a syntactic symmetry group $G$ of $\Phi$ in the sense of~\cite{KS18} induces a group $\bar G = \{\bar f \mid f \in G\} \leq \Aut(\B(V))$ of automorphisms, each fixing $\varphi$ by Condition~(3).

\paragraph{Prefix-admissibility and lex-permissibility.}
We show that every $\bar f \in \bar G$ is prefix-admissible for
$\Phi^{\dep}$ and lex-permissible, so that Theorem~\ref{thm:symm_breaker}
applies with $\Lambda = [k] \times \bar G$.

Condition~(i) of Definition~\ref{def:prefix-admissible} follows from~(2): every variable occurring in $f(x)$ lies in the block of $x$, which is universal, so $\supp(\bar f(x)) \subseteq \Var(f(x)) \subseteq X$ and $\bar f(x) \in \B(X)$ by Lemma~\ref{lem:support}; likewise $\bar f(y) \in \B(Y)$ for $y \in Y$.

For condition~(ii), let $x_i \in \Dep(\bar f(y_j))$. Then $x_i \in D_\ell$ for some $y_\ell \in \supp(\bar f(y_j)) \subseteq \Var(f(y_j))$. By~(2), $y_\ell$ lies in the same existential block as $y_j$, so $D_\ell = D_j$ and $x_i \in D_j$. The inverse $f^{-1}$ is again a symmetry in the sense of~\cite{KS18}, so by~(2) every variable occurring in $f^{-1}(x_i)$ lies in the universal block of $x_i$. This block precedes $y_j$ and is hence contained in $D_j$. Therefore $\supp(\bar f^{-1}(x_i)) \subseteq \Var(f^{-1}(x_i)) \subseteq D_j$, and therefore $\bar f^{-1}(x_i) \in \B(D_j)$ by Lemma \ref{lem:support}, as required. Thus $\bar G$ is a syntactic symmetry group of $\Phi^{\dep}$ in the sense of Definition~\ref{def:syntactic}.

Finally, let $j \in [k]$ and $x \in D_j$. The universal block of $x$ precedes $y_j$ and is therefore contained in $D_j$, so $\supp(\bar f(x)) \subseteq D_j$, i.e.\ $\bar f(x) \in \B(D_j)$. This is condition~(i) of Definition~\ref{def:lex_perm}; condition~(ii) is vacuous since the dependency sets are linearly ordered. Hence every pair $(j, \bar f) \in [k] \times \bar G$ is lex-permissible.

\paragraph{The symmetry breaker.}
Let $v_i = y_j$ be existential. In $\Phi^{\dep}$, $D_j$ consists precisely of the universal variables preceding $v_i$, and $y_1,\dots,y_{j-1}$ are precisely the existential variables preceding $v_i$. Hence
\[
D_j \cup \{y_\ell \mid \ell < j\} = \{v_\ell \mid \ell < i\}.
\]
Theorem~\ref{thm:symm_breaker} with $\Lambda = [k] \times \bar G$ thus yields the conjunctive symmetry breaker
\[
  \bigwedge_{\substack{i=1\\ Q_i = \exists}}^{n}\;
  \bigwedge_{\bar f \in \bar G}
  \left(
    \Bigl(\bigwedge_{\ell < i} \bigl(v_\ell \leftrightarrow \bar f(v_\ell)\bigr)\Bigr)
    \to
    \bigl(v_i \to \bar f(v_i)\bigr)
  \right),
\]
which is exactly the symmetry breaker of \cite[Theorem~22]{KS18}.

Theorem~\ref{thm:symm_breaker} therefore strictly generalises the QBF construction. For symmetries in the sense of~\cite{KS18} it yields the same predicate. In the DQBF setting, however, Definition~\ref{def:lex_perm} does not require the individual quantifier blocks to be preserved: it only requires that the cumulative dependency sets $D_j$ be respected. This admits symmetries that are not admissible in the sense of~\cite{KS18}, as the following example shows.

\begin{ex}\label{ex:strict_gen}
Consider the canonical DQBF representation
\[
\forall x_1,x_2\,
\exists y_1(\{x_1\})\,
\exists y_2(\{x_1,x_2\}).\,
\bigl((x_1\leftrightarrow y_1)\wedge
      (x_2\leftrightarrow y_2)\bigr)
\]
of the QBF $\forall x_1\exists y_1\forall x_2\exists y_2.\,
\bigl((x_1\leftrightarrow y_1)\wedge
      (x_2\leftrightarrow y_2)\bigr)
$. 
The map
\[
g(x_1)=x_1,\qquad
g(x_2)=x_1\oplus x_2,\qquad
g(y_1)=y_1,\qquad
g(y_2)=y_1\oplus y_2
\]
is a syntactic symmetry: it is prefix-admissible and preserves the matrix $\varphi$, since
\[
\begin{aligned}
g(\varphi)
&=
(x_1\leftrightarrow y_1)
\wedge
\bigl((x_1\oplus x_2)\leftrightarrow
      (y_1\oplus y_2)\bigr)\\
&=
(x_1\leftrightarrow y_1)
\wedge
(x_2\leftrightarrow y_2)
=
\varphi.
\end{aligned}
\]
For each $j\in\{1,2\}$, the pair $(j,g)$ satisfies condition~(i) of Definition~\ref{def:lex_perm}, while condition~(ii) is vacuous because $D_1\subseteq D_2$. Hence $g$ is lex-permissible. However, it is not admissible in the QBF sense of~\cite{KS18}, since it mixes variables belonging to distinct quantifier blocks.
\end{ex}

\section{Canonisation Functions} 
\label{sec:canonisers}

The general lex-leader construction of Section \ref{sect:lex_leader} applies a symmetry within a single dependency cluster of a target variable $y_j$. It leaves the interpretation unchanged outside that cluster and fixes every Skolem function whose dependency set is incomparable with $D_j$. For QBFs, the dependency sets form a chain, making the latter restriction vacuous and recovering the symmetry breaker of~\cite{KS18}.
Section \ref{sec:experiments} demonstrates that this local lex-leader symmetry breaking construction can also be beneficial for genuine DQBF families that are not induced by QBFs and whose truth-value depends on its incomparable dependency sets (e.g. the XOR-game). 

However, succinct combinatorial encodings often contain synchronised tuples of Skolem functions with two disjoint chains of dependency sets. Symmetries acting jointly on these tuples may yield only trivial local predicates or fail the lex-permissibility conditions at relevant targets. For example, the local lex-leader construction yields no nontrivial constraints for succinct graph colouring, while for $\DPHP$ it yields only the normalisation $H=0^n$, leaving substantial symmetry unbroken.

A broader perspective is provided by viewing symmetry breaking as the selection of representatives from symmetry orbits. As emphasised by \citet{ABR24complexity}, complete symmetry breaking may select any canonical representative of each orbit; lexicographic minimality is only one possible choice. This freedom allows the choice of representatives to reflect the structure of the problem. We therefore complement the local construction with a framework for canonisation functions that enables one to exploit the structure of the encoding to break coordinated symmetries across incomparable dependency sets, which we demonstrate on the succinct pigeonhole principle. This also follows the structure-based approach to SAT symmetry breaking exemplified by \textsc{Satsuma}~\cite{ABR26,ACH26}, whose experiments demonstrate the benefits of exploiting specific symmetry structures over generic lex-leader symmetry breaking. In particular, we do not need lex-permissibility conditions for the symmetries considered in this section. 

\subsection{Function Encodings}
\label{sec:function_encodings}

Many succinct combinatorial encodings represent a function through two synchronised tuples of Skolem functions. In this section we consider DQBFs $\Phi=P.\varphi$ with
\[
P=
\forall X\,\forall Y\,
\exists F(X)\,\exists F'(Y)\,
\exists z_1(D_1)\cdots\exists z_\ell(D_\ell),
\]
and
\[
\varphi=
\bigl((X=Y)\to(F=F')\bigr)\land\theta,
\]
where $X, Y$ are disjoint tuples of $n$ universal variables, $F,F'$ are tuples of $m$ existential variables, and $D_i \subseteq X \cup Y$. Tuple quantification abbreviates quantification of its components, each with the indicated dependency set. The variables $z_i$ accommodate auxiliary information required by the encoding. Write $\mathcal A=\{0,1\}^n$ and $\mathcal B=\{0,1\}^m$. For an interpretation $s \in \mathbb{S}(P)$, write $F_s(a)$ for the tuple of values of the Skolem functions interpreting $F$ at the input assignment $X=a$. An interpretation $s \in\mathbb S(P)$ is \emph{synchronised} if $F_s(a)=F'_s(a)$ for every $a\in\mathcal A$. Let
\[
  \mathcal K_{\sync} :=
  \{s \in \mathbb S(P) \mid
    F_s(a) = F'_s(a) \text{ for every }a\in\mathcal A\}.
\]
Every $s\in\mathcal K_{\mathrm{sync}}$ represents a function
\[
  c_s : \mathcal A \to \mathcal B, \qquad c_s(a)=F_s(a)=F'_s(a).
\]
The synchronisation constraint $(X = Y) \to (F = F')$ ensures that $\Mod(\Phi)\subseteq\mathcal K_{\sync}$. We call
\[
  \mathcal S_\Phi :=\{c_s\mid s \in \Mod(\Phi)\}
\]
the \emph{represented solution set}. Different interpretations may represent the same function, since their auxiliary Skolem functions may differ. 

\paragraph{Symmetries of represented solutions.}
For permutations $\pi\in\operatorname{Sym}(\mathcal A)$ and $\rho\in\operatorname{Sym}(\mathcal B)$, define
\[
(\pi,\rho)\cdot c
:=
\rho\circ c\circ\pi^{-1}.
\]
Thus $\pi$ permutes the inputs and $\rho$ permutes the output values. We call $(\pi,\rho)$ a \emph{solution symmetry} of $\Phi$ if
\[
\{(\pi,\rho)\cdot c\mid c\in\mathcal S_\Phi\}
=
\mathcal S_\Phi.
\]
For example, input permutations may exchange interchangeable positions in a combinatorial witness, while output permutations may rename colours (see Example \ref{ex:colouring3}). A solution symmetry acts on represented functions and need not arise from a syntactic symmetry of the encoding. 

In the next section, we use solution symmetries to choose canonical representatives of the represented functions, using the notion of canonisers. 

\subsection{Specialised Canonical Symmetry Breakers}
\label{sec:canonical_breakers}
Fix a syntactic symmetry group $G_{\syn}$ of $\Phi$, with associated semantic symmetry group $G_{\sem}$, and a $G_{\sem}$-invariant set $\mathcal{K}$ with $\Mod(\Phi)\subseteq\mathcal K \subseteq\mathcal K_{\mathrm{sync}}$. Under certain conditions, one can combine symmetry breakers. 

\begin{lemma}
\label{lem:combining_breakers}
Let $\psi_1$ be a $\mathcal K$-conjunctive symmetry breaker w.r.t. $G_{\sem}$, and put
\[
\mathcal K_1
:= \{s\in\mathcal K\mid[P.\psi_1]_s=1\}.
\]
Suppose that a subgroup $H\leq G_{\sem}$ preserves $\mathcal K_1$ and that $\psi_2$ satisfies
\[
\forall t\in\mathcal K_1\;\exists h\in H:
\qquad [P.\psi_2]_{h(t)}=1.
\]
Then $\psi_1\land\psi_2$ is a $\mathcal K$-conjunctive symmetry breaker w.r.t. $G_{\sem}$.
\end{lemma}
\begin{proof}
Let $s\in\mathcal K$. Since $\psi_1$ is a $\mathcal K$-conjunctive symmetry breaker, choose $f \in G_{\sem}$ with $t := f(s) \in \mathcal K_1$. Choose $h\in H$ with $[P.\psi_2]_{h(t)}=1$. Since $H$ preserves $\mathcal K_1$, we also have $[P.\psi_1]_{h(t)}=1$. Thus by Lemma \ref{lem:eval_conjunction} it follows that $[P.(\psi_1\land\psi_2)]_{(h\circ f)(s)}=1$, and since $h\circ f\in G_{\sem}$ the conclusion follows.
\end{proof}

\paragraph{From semantic orbits to represented functions.}
Let $\mathcal C=\mathcal B^{\mathcal A}$ and define the following map, projecting an interpretation $s \in \mathcal{K}$ on its represented function
\[
  \eta:\mathcal K\to\mathcal C,
  \qquad \eta(s)=c_s.
\]
For $G\leq\Sym(\mathcal A)\times\Sym(\mathcal B)$,
consider the lifting condition
\begin{equation}
\label{eq:solution_semantic_lift}
  \forall s\in\mathcal K\;\forall g\in G\;
  \exists f\in G_{\sem}:
  \qquad \eta(f(s))=g\cdot\eta(s).
\end{equation}
It requires each transformed function to be represented within the semantic orbit of its original interpretation. Invariance of $\mathcal K$ ensures that $\eta(f(s))$ is defined. Lemma \ref{lem:lifting_solution_symmetries} shows that if $G$ satisfies the lifting condition, it consists of solution symmetries of $\Phi$. 

\begin{lemma}
\label{lem:lifting_solution_symmetries}
Let $\Phi$, $G_{\sem}$, and $\mathcal K$ be as above. If $G\leq\Sym(\mathcal A)\times\Sym(\mathcal B)$ satisfies \eqref{eq:solution_semantic_lift}, then $G$ consists of solution symmetries of $\Phi$.
\end{lemma}
\begin{proof}
Let $g\in G$ and $c\in\mathcal S_\Phi$. Choose a model $s \in \Mod(\Phi) \subseteq \mathcal K$ representing $c$. By~\eqref{eq:solution_semantic_lift}, there exists $f \in G_{\sem}$ such that $c_{f(s)} = g \cdot c_s = g \cdot c$.  Since semantic symmetries preserve models, $f(s) \in \Mod(\Phi)$, and hence $g\cdot c\in\mathcal S_\Phi$. Thus $g \cdot \mathcal S_\Phi \subseteq \mathcal S_\Phi$. Applying the same argument to $g^{-1}\in G$ gives the reverse inclusion, so $g \cdot \mathcal S_\Phi = \mathcal S_\Phi$. Therefore every $g \in G$ is a solution symmetry.
\end{proof}

A sufficient condition is that every $g\in G$ admit a syntactic lift $\tilde g\in G_{\syn}$ satisfying
\begin{equation}
\label{eq:syntactic_solution_lift}
  c_{\tilde g(s)}=g\cdot c_s
  \qquad\text{for every }s\in\mathcal K.
\end{equation}
The induced semantic symmetry then witnesses \eqref{eq:solution_semantic_lift}. Lemma \ref{lem:generator_lifts} shows that it suffices to provide such lifts for a generating set of $G$. 

\begin{lemma}
\label{lem:generator_lifts}
Let $G \leq \Sym(\mathcal A) \times \Sym(\mathcal B)$, and let $A \subseteq G$ be a generating set. If every $a \in A$ admits a syntactic lift $\tilde{a} \in G_{\syn}$, then every $g \in G$ admits a syntactic lift $\tilde{g} \in G_{\syn}$. 
\end{lemma}

\begin{proof}
First observe that if $\tilde a$ is a lift of $a \in A$, then $\tilde a^{-1}$ is a lift of $a^{-1}$. Indeed, for $s \in \mathcal{K}$ put $t = \tilde a^{-1}(s)$. By $G_{\sem}$-invariance we have $t \in \mathcal{K}$, and $c_s = c_{\tilde a(t)} = a \cdot c_t$.  Applying $a^{-1}$ yields $a^{-1} \cdot c_s = c_t = c_{\tilde a^{-1}(s)}$. 

Thus, if $\tilde g$ lifts $g$ and $\tilde h$ lifts $h$, then $\tilde h\circ\tilde g$ lifts $g\circ h$, since
\[
\begin{aligned}
c_{(\tilde h\circ\tilde g)(s)} = c_{\tilde g(\tilde h(s))} = g\cdot c_{\tilde h(s)} = g\cdot(h\cdot c_s) = (g\circ h)\cdot c_s.
\end{aligned}
\]
Together with the identity lift and the preceding observation about inverses, this shows that the elements admitting lifts form a subgroup of $G$ containing $A$, hence equal to $G$.
\end{proof}

%One can verify that the lifting condition \eqref{eq:solution_semantic_lift} implies that $G \leq \Sym(\mathcal A) \times \Sym(\mathcal B)$ consists of solution symmetries. Its requirement is stronger: each transformed function must be represented within the semantic orbit of the original interpretation, and this must hold for every $s \in \mathcal K$, not only for models.

Next we introduce the concept of a canonisation function, which generalises the lex-leader procedure of picking a canonical representative from each orbit. Such functions provide a standard means of normalisation up to symmetry in graph isomorphism~\cite{SW19} and permutation-group algorithms~\cite{JJPW19}, and underlie complete lex-leader symmetry breaking~\cite{symmetry-SAT}.

\begin{defi}[Canoniser]
\label{def:canoniser}
A \emph{canonisation function} or \emph{canoniser} for the action of $G$ on $\mathcal C$ is a map $\can_G: \mathcal C \to \mathcal C$ such that, for every $c \in \mathcal C$ and $g\in G$,
\[
  \can_G(c) \in \Orb_G(c)
  \qquad \text{and} \qquad
  \can_G(g \cdot c) = \can_G(c).
\]
We may drop the subscript $G$ from the notation when $G$ is fixed.
We call $c$ \emph{canonical} if $\can_G(c) = c$.
For a function encoding $\Phi=P.\varphi$, we say that a Boolean formula $\psi$ over the variables of $P$ \emph{realises} a canoniser $\can$ on $\mathcal K$ if, for every $s \in \mathcal K$,
\[
  [P.\psi]_s=1
  \quad\Longleftrightarrow\quad
  \can(c_s)=c_s.
\]
\end{defi}
Thus a canoniser selects one representative from each $G$-orbit in $\mathcal C$. In particular, $\can^2=\can$, and its fixed points are precisely the selected representatives. For example, any total order $\preceq$ on $\mathcal C$ defines a canoniser by
\[
\can(c):=\min_{\preceq}\Orb_G(c).
\]
For the lexicographic order, the fixed points are precisely the lex-leaders in $\mathcal C$.

The next proposition establishes the main connection between canonisers and their realisers: if $\psi$ realises the canoniser $\can$ on $\mathcal{K}$, it is a $\mathcal{K}$-conjunctive symmetry breaker. 

\begin{prop}[Canonical symmetry breaking]
\label{prop:canonical_breaking}
Let $\Phi=P.\varphi$, $G_{\syn}$, $G_{\sem}$, and $\mathcal K$ be as above. Let $G \leq \Sym(\mathcal A) \times \Sym(\mathcal B)$ satisfy \eqref{eq:solution_semantic_lift}, and let $\can$ be a canoniser for $G$. If $\psi_{\can}$ realises $\can$ on $\mathcal K$, then  $\psi_{\can}$ is a $\mathcal K$-conjunctive symmetry breaker w.r.t. $G_{\sem}$. Moreover, $P.(\varphi\land\psi_{\can})$ represents precisely the canonical represented solutions $\can(\mathcal S_\Phi)$.
\end{prop}
\begin{proof} 
Let $s \in \mathcal K$. Since $\can(c_s) \in \Orb_G(c_s)$, the lifting condition \eqref{eq:solution_semantic_lift} yields $f \in G_{\sem}$ such that $c_{f(s)}=\can(c_s)$. By invariance, $f(s) \in \mathcal K$. Idempotence of $\can$ gives
\[
\can(c_{f(s)})=\can(\can(c_s))=\can(c_s)=c_{f(s)}.
\]
Since $\psi_{\can}$ realises $\can$ on $\mathcal K$, we obtain $[P.\psi_{\can}]_{f(s)}=1$. Thus $\psi_{\can}$ is a $\mathcal K$-conjunctive symmetry breaker. 
If $s$ is a model of $\Phi$, then so is $f(s)$. Consequently, $\can(\mathcal S_\Phi)\subseteq\mathcal S_\Phi$. By realisation, the functions represented by $P.(\varphi\land\psi_{\can})$ are exactly
\[
\{c\in\mathcal S_\Phi\mid\can(c)=c\} = \can(\mathcal S_\Phi),
\]
where the equality follows from the inclusion and idempotence.
\end{proof}

Proposition \ref{prop:canonical_breaking} provides the main machinery for constructing conjunctive symmetry breakers based on canonisers for function encodings. The main challenge that remains is to identify problem-specific canonisers that admit compact realisers. We do so for two specialised canonisation principles: sorting and rooting. 

\paragraph{Sorting interchangeable inputs.}
Suppose that all input permutations are solution symmetries, and take $G = \Sym(\mathcal A) \times \{\mathrm{id}_{\mathcal B}\}$. Order $\mathcal A$ and $\mathcal B$ lexicographically, and write $\mathcal A = \{a_0 < \cdots < a_{N-1}\}$. Sorting the sequence $c(a_0),\ldots,c(a_{N-1})$ into nondecreasing order defines a canoniser $\can_{\sort}$: two functions belong to the same orbit precisely when their multisets of values agree. Its fixed points are characterised by
\[
\begin{aligned}
  \can_{\sort}(c)=c
  &\quad\Longleftrightarrow\quad
  c(a_0)\leq_{\lex}\cdots\leq_{\lex}c(a_{N-1})\\
  &\quad\Longleftrightarrow\quad
  \forall a,b\in\mathcal A,\;
  a<_{\lex}b\to c(a)\leq_{\lex}c(b).
\end{aligned}
\]
The second equivalence follows from transitivity of $\leq_{\lex}$. The predicate
\begin{equation}
\label{eq:canonical_sorting}
  \psi_{\sort}:=(X<_{\lex}Y)\to(F\leq_{\lex}F')
\end{equation}
realises $\can_{\sort}$ on $\mathcal K$. Under~\eqref{eq:solution_semantic_lift}, it is therefore a $\mathcal K$-conjunctive symmetry breaker by Proposition \ref{prop:canonical_breaking}. For the lifting condition \eqref{eq:solution_semantic_lift}, it suffices by Lemma \ref{lem:generator_lifts} that each adjacent transposition $\tau_i = (a_i\ a_{i+1})$ admits a syntactic lift $\tilde{\tau_i} \in G_{\syn}$, since these transpositions generate $\Sym(\mathcal A)$. Note that no lex-permissibility assumption is required, the lifted symmetries may act simultaneously on incomparable dependency sets. 

\paragraph{Sorting within interchangeable classes.}
More generally, let $E$ be an equivalence relation on $\mathcal A$. Suppose that every input permutation mapping each $E$-class to itself is a solution symmetry, and let $G$ consist of these permutations with the identity output action. Sorting within each class defines a canoniser, with realiser
\[
  \psi_E
  :=
  \bigl(E(X,Y) \land X <_{\mathrm{lex}} Y \bigr)
  \to (F \leq_{\mathrm{lex}} F'), 
\]
where the lifting condition \eqref{eq:solution_semantic_lift} applies to permutations $\pi \in \Sym(\mathcal A)$ preserving each $E$-class. The sorting predicate has a circuit of size $O(n+m)$, block sorting has size $O(|E|+n+m)$, where $|E|$ denotes the size of a supplied circuit for $E$.

\paragraph{Rooting cyclic solutions.}
Write $\mathcal A = \{a_0 < \cdots < a_{N-1}\}$, let $r(a_i)=a_{(i+1)\bmod N}$, and take
$G = \langle r \rangle \times \{\mathrm{id}_{\mathcal B}\}$. The canoniser $\can_{\mathrm{root}}$ selects the lexicographically least cyclic rotation of each function table. If every $c \in \eta(\mathcal K)$ takes the value $0^m$ exactly once, this is the unique rotation starting with $0^m$. Hence
\[
\psi_{\mathrm{root}}:=(X=0^n)\to(F=0^m)
\]
realises $\can_{\mathrm{root}}$ on $\mathcal K$. Under~\eqref{eq:solution_semantic_lift}, it is therefore a $\mathcal K$-conjunctive symmetry breaker by Proposition \ref{prop:canonical_breaking}. A syntactic lift of $r$ suffices for the lifting condition by Lemma \ref{lem:generator_lifts}. The rooting predicate has a circuit of size $O(n+m)$.

An instantiation of Lemma \ref{lem:combining_breakers} shows that canonisers can be combined with the local symmetry breaker of Theorem \ref{thm:symm_breaker} under an additional assumption, as stated in Corollary \ref{cor:combining_breakers}. 
 
\begin{cor}
\label{cor:combining_breakers}
Let $\psi_{\loc}$ be the local breaker of Theorem~\ref{thm:symm_breaker}, and let $\psi_{\can}$ realise a canoniser $\can$ for $G$ on $\mathcal K$. If every $g\in G$ admits a syntactic lift satisfying
\eqref{eq:syntactic_solution_lift} and preserving
\[
  \mathcal K_{\loc}
  :=\{s\in\mathcal K\mid[P.\psi_{\loc}]_s=1\},
\]
then $\psi_{\loc}\land\psi_{\can}$ is a $\mathcal K$-conjunctive symmetry breaker w.r.t. $G_{\sem}$.
\end{cor}
\begin{proof}
Let $H\leq G_{\sem}$ be generated by the semantic symmetries induced by the chosen lifts. By assumption, $H$ preserves $\mathcal K_{\loc}$. The orbit argument in the proof of Proposition~\ref{prop:canonical_breaking}, applied within $\mathcal K_{\loc}$ using these lifts, shows that every $H$-orbit in $\mathcal K_{\loc}$ contains an interpretation satisfying $\psi_{\can}$. Lemma~\ref{lem:combining_breakers} now gives the conclusion.
\end{proof}

The general method is not limited to the three examples presented here: any rule selecting one representative per orbit is a canoniser, and Table \ref{tab:compact_canonisers} lists natural canonisers for several succinct combinatorial problems. Whether these admit compact realisers is a separate question. For sorting and rooting the answer is positive, as shown above. For value precedence, which refers to first occurrences of values over an exponentially large domain, it is not clear that a predicate over two synchronised copies suffices, and we leave the construction of realisers for the remaining examples to future work.

\begin{table}[h]
\caption{Canonisers for selected symmetry actions in succinct combinatorial encodings.}
\label{tab:compact_canonisers}
\centering
\small
\begin{tabularx}{\linewidth}{@{}
  >{\raggedright\arraybackslash}X
  >{\raggedright\arraybackslash}X
  >{\raggedright\arraybackslash}X
@{}}
\toprule
\textbf{Symmetry action}
& \textbf{Corresponding canoniser}
& \textbf{Succinct combinatorial examples} \\
\midrule

Permutations of interchangeable inputs
& Sorting
& Pigeonhole principle, set packing, clique, independent set \\
\addlinespace

Permutations of interchangeable values
& Value precedence
& Graph colourability, bin packing, Ramsey edge-colouring \\
\addlinespace

Input permutations within interchangeable classes
& Block sorting
& Colouring of complete multipartite graphs, subset-sum \\
\addlinespace

Cyclic shifts
& Minimal rotation
& Directed Hamiltonian cycle \\
\addlinespace

Reversal
& Canonical orientation
& Undirected Hamiltonian path \\

\bottomrule
\end{tabularx}

\end{table}

\subsection{Succinct Pigeonhole Principle}
\label{sec:canonical_examples}

We illustrate the above construction for the succinct pigeonhole principle of Section \ref{sect:succinct}. Let $\mathcal A=\{0,1\}^n$, where $n\geq 1$, and recall the encoding of $\DPHP_n$ from Section \ref{sect:succinct}. In Example \ref{ex:dphp_symmetries} we have seen that Theorem \ref{thm:symm_breaker} produces the breaker $\psi_{\loc} := (H = 0^n)$, normalising the missing hole, and that every permutation $\pi \in \Sym(\mathcal A)$ admits a generalised syntactic lift for $\DPHP_n$. These therefore provide the full input permutation action used by the sorting canoniser. 

For the combined construction, let $G_{\syn}$ be generated by the output bit-flips used for local normalisation and the input permutations used for sorting, and set $\mathcal K=\mathcal K_{\sync}$. Every generator preserves the synchronisation predicate $(X=Y) \to (F=F')$. The associated semantic group $G_{\sem}$ therefore preserves $\mathcal K_{\mathrm{sync}}$. The input lifts satisfy~\eqref{eq:syntactic_solution_lift} on all synchronised interpretations. Consequently,
\[
  \psi_{\sort}
  :=(X <_{\lex} Y) \to (F\leq_{\lex}F')
\]
is a $\mathcal K_{\mathrm{sync}}$-conjunctive symmetry breaker. Since these lifts fix $H$, they preserve the local normalisation $H = 0^n$. By Corollary \ref{cor:combining_breakers} we therefore obtain the combined breaker
\[
\psi_{\DPHP}
:=\psi_{\loc}\land\psi_{\sort}
= (H=0^n)\land
\bigl((X<_{\lex}Y)\to(F\leq_{\lex}F')\bigr).
\]
Explicitly, a synchronised interpretation represents a pair $(h,c)$ with $h\in\mathcal A$ and
$c : \mathcal A \to \mathcal A$. Simultaneous output bit-flips followed by an input permutation
normalise this pair to
\[
(h,c) \mapsto \bigl(0^n,\operatorname{sort}(c\oplus h)\bigr),
\]
where $(c\oplus h)(a):=c(a)\oplus h$. The resulting canonisation is illustrated in Figure \ref{fig:dphp-canonisation} for $n = 2$. An explicit CNF encoding of $\psi_{\DPHP}$ of size $\Theta(n)$ is provided in Appendix \ref{app:dphp-cnf:sorting}. 

\begin{figure}[h]
\centering
\resizebox{\linewidth}{!}{%
\begin{tikzpicture}[
    vertex/.style={
        circle, draw, minimum size=1.5mm, inner sep=0pt
    },
    ordinary/.style={vertex, fill=blue!70},
    collision/.style={vertex, fill=red!70},
    excluded/.style={vertex, draw=gray, fill=gray!20},
    mapping/.style={
        ->, thick, >={Latex[length=1.6mm,width=1.1mm]}
    },
    font=\small
]

% --- input vertices ---
\foreach \panel/\offset in {A/0,B/4.8,C/9.6}{
    \foreach \i/\lab in {0/00,1/01,2/10,3/11}{
        \node[vertex, fill=blue!70]
            (\panel-p\i) at (\offset,{1.35-0.9*\i}) {};
        \node[left=3pt] at (\panel-p\i.west) {$\lab$};
    }
}

% --- first picture ---
\node[vertex, fill=red!70, label=right:{$00$}]
    (A-h0) at (2,1.35) {};
\node[vertex, fill=blue!70, label=right:{$01$}]
    (A-h1) at (2,0.45) {};
\node[vertex, fill=blue!70, label=right:{$10$}]
    (A-h2) at (2,-0.45) {};
\node[vertex, draw=gray, fill=gray!20,
      label={[text=gray]right:{$11=H$}}]
    (A-h3) at (2,-1.35) {};

% --- second picture: hole normalised ---
\node[vertex, draw=gray, fill=gray!20,
      label={[text=gray]right:{$00=H$}}]
    (B-h0) at (6.8,1.35) {};
\node[vertex, fill=blue!70, label=right:{$01$}]
    (B-h1) at (6.8,0.45) {};
\node[vertex, fill=blue!70, label=right:{$10$}]
    (B-h2) at (6.8,-0.45) {};
\node[vertex, fill=red!70, label=right:{$11$}]
    (B-h3) at (6.8,-1.35) {};

% --- third picture: sorted ---
\node[vertex, draw=gray, fill=gray!20,
      label={[text=gray]right:{$00=H$}}]
    (C-h0) at (11.6,1.35) {};
\node[vertex, fill=blue!70, label=right:{$01$}]
    (C-h1) at (11.6,0.45) {};
\node[vertex, fill=blue!70, label=right:{$10$}]
    (C-h2) at (11.6,-0.45) {};
\node[vertex, fill=red!70, label=right:{$11$}]
    (C-h3) at (11.6,-1.35) {};

% --- cross out the excluded holes ---
\foreach \hole in {A-h3,B-h0,C-h0}{
    \draw[gray, thick]
        (\hole.north west) -- (\hole.south east)
        (\hole.south west) -- (\hole.north east);
}
% --- arrows, drawn behind the vertices ---
\begin{scope}[on background layer]

    % Initial candidate: c = (00,01,10,00), H = 11.
    \draw[mapping, red!70]  (A-p0.east) -- (A-h0.west);
    \draw[mapping, blue!70] (A-p1.east) -- (A-h2.west);
    \draw[mapping, blue!70] (A-p2.east) -- (A-h1.west);
    \draw[mapping, red!70]  (A-p3.east) -- (A-h0.south west);

    % Output XOR 11: c = (11,10,01,11), H = 00.
    \draw[mapping, red!70] (B-p0.east) -- (B-h3.north west);
    \draw[mapping, blue!70] (B-p1.east) -- (B-h1.west);
    \draw[mapping, blue!70] (B-p2.east) -- (B-h2.west);
    \draw[mapping, red!70]  (B-p3.east) -- (B-h3.west);

    % Input sorting: c = (01,10,11,11), H = 00.
    \draw[mapping, blue!70] (C-p0.east) -- (C-h1.west);
    \draw[mapping, blue!70] (C-p1.east) -- (C-h2.west);
    \draw[mapping, red!70]  (C-p2.east) -- (C-h3.north west);
    \draw[mapping, red!70]  (C-p3.east) -- (C-h3.west);

\end{scope}

% --- stages and transitions ---
\node at (1.2,2.0)  {Initial candidate};
\node at (6.0,2.0)  {Hole normalised};
\node at (10.8,2.0) {Sorted};

\draw[mapping] (3.2,0) -- (4.0,0);
\draw[mapping] (8.0,0) -- (8.8,0);

\end{tikzpicture}%
}
\caption{Two-stage normalisation of a candidate for $\mathrm{DPHP}_2$.
Output bit-flips set $H=00$, then an input permutation sorts the
function values. The function on the right is a canonical orbit-representative.}
\label{fig:dphp-canonisation}
\Description{Three mapping diagrams show a pigeonhole candidate, normalisation of the excluded output to 00 by output bit flips, and sorting by an input permutation. The sorted output sequence is 01, 10, 11, 11.}
\end{figure}

\section{Experiments} \label{sec:experiments}
The main contributions of this paper are theoretical. The experiments below provide a proof of concept for the local lex-leader construction and the specialised canonisation predicates. We derive the required symmetry-breaking predicates manually and generate their encodings for selected benchmark families, leaving a fully automated tool and a comprehensive experimental evaluation to future work.

%For literal symmetries, an automated preprocessing procedure can follow the QBF approach~\cite{KS18}. Encode the DQBF $\Phi$ as a coloured graph $\mathcal G_\Phi$ representing its clause structure, literal complementation, quantifier types, and dependencies, whose automorphisms induce exactly the literal syntactic symmetries preserving its CNF structure. Highly optimised graph-automorphism tools such as \textsc{Saucy}~\cite{saucy} can then compute generators of this symmetry group. For each generator $g$, test the pairs $(j,g)$ for lex-permissibility and construct the corresponding predicates using Theorem~\ref{thm:symm_breaker}. Their conjunction can then be encoded in CNF and appended to the input formula.

\subsection{Experimental Setup}

We evaluate our symmetry-breaking predicates with the four DQBF solvers that performed best in the DQBF tracks of the QBF Galleries 2023~\cite{HPS26} and 2026~\cite{PPS26}. These solvers follow different solving approaches:
\begin{itemize}
\item \textsc{Pedant}~\cite{RS22,Rei26} combines propositional definition extraction with counterexample-guided inductive synthesis (CEGIS) to compute Skolem functions.
\item \textsc{P3}~\cite{Rei26} is a simplified version of \textsc{Pedant} with optimised internal data-structure for representing Skolem functions.
\item \textsc{HQS}~\cite{hqs,HQSweb} is based on quantifier elimination: it eliminates dependencies to transform the DQBF into an equisatisfiable QBF, which it then solves.
\item \textsc{DQBDD}~\cite{dqbdd,Sic22} combines quantifier localisation and quantifier elimination with a representation of the formula by binary decision diagrams (BDDs).
\end{itemize}

\textsc{HQS} and \textsc{DQBDD} were disqualified from the QBF Gallery 2026 because they returned incorrect results on some instances~\cite{PPS26}. We nevertheless include them here because their solving approaches differ from the CEGIS-based approach of \textsc{Pedant} and \textsc{P3}. This allows us to observe how the effect of symmetry breaking depends on the solving technique. In our experiments, every solver run that terminated returned the correct result.

All experiments were run on a compute cluster of dual-socket AMD EPYC 7313 @ \SI{3.7}{\giga\hertz} machines running Ubuntu 24.04.5.
Each solver run was given 16~cores and \SI{16}{\giga\byte} of memory, and was limited to \SI{1}{\hour} of wall-clock time. 
We report wall-clock running times in seconds, where $\timeout$ indicates that the time limit was exceeded; the memory limit was not reached in any of the runs reported below. 

\subsection{Benchmark Families}
We present four experiments that demonstrate the effectiveness of our framework on different DQBF formula families and illustrate complementary aspects thereof. First, we assess local lex-leader symmetry breaking on three variants of a \emph{XOR-game family}: a false family $\mathrm{XG}^{-}_n$ whose falsity relies on incomparable dependencies; a true family $\mathrm{XG}^{+}_n$; and a false family $\mathrm{DXG}^{-}_n$ in which the hidden-vector variables have non-empty dependency sets. The third variant additionally illustrates local normalisation of entire Skolem functions, extending the normalisation of a constant hidden vector in the first two variants. Second, we assess the combination of local lex-leader symmetry breaking with a specialised canonisation predicate for the succinct pigeonhole principle. Table~\ref{tab:experimental_design} summarises the experimental benchmark set-up.

\begin{table}[h]
\caption{Overview of the experiments.}
\label{tab:experimental_design}
\centering
\small
\begin{tabularx}{\linewidth}{@{}
  >{\raggedright\arraybackslash}X
  l
  >{\raggedright\arraybackslash}X
  >{\raggedright\arraybackslash}X
@{}}
\toprule
\textbf{Setting}
& \textbf{Family}
& \textbf{Configurations}
& \textbf{Purpose} \\
\midrule

Genuine DQBFs
&
& No SB, local SB
& Assess local lex-leader symmetry breaking on: \\
\addlinespace

& $\XG^-_n$
&
& \textbullet\ a false family; \\
\addlinespace

& $\XG^+_n$
&
& \textbullet\ a true family; \\
\addlinespace

& $\DXG^-_n$
&
& \textbullet\ a false family with non-empty dependency sets
  for the hidden-vector variables. \\

\midrule

Succinct combinatorial encodings
& $\DPHP_n$
& No SB, local SB, local + specialised SB
& Assess the additional benefit of specialised
  canonisation predicates. \\

\bottomrule
\end{tabularx}
\end{table}

All benchmark instances are given in the DQDIMACS format. They consist of CNF encodings of the formula families to which the CNF encodings of the respective symmetry-breaking predicates are appended.
We describe the CNF encodings of our symmetry-breaking predicates in Appendix~\ref{app:cnf}. 
All used instances, together with the scripts generating them and the logs of all solver runs, are available at \url{https://gitlab.sai.jku.at/dqbf/symmetry-breaking.git}.

\paragraph{XOR-game.}
We evaluate the local lex-leader construction on the synthetic benchmark \emph{hidden-vector XOR game}, a DQBF family whose falsity relies on incomparable dependency sets. 
For $n\geq1$, let $U,V,X$ be $n$-bit tuples and define
\[
\XG^{-}_n :=
\forall U\,\forall V\,
\exists X(\varnothing)\,
\exists t(U)\,\exists s(V).
\;
\bigl((t\oplus s)\leftrightarrow
      ((U=X)\land(V=X))\bigr),
\]
where tuple equality abbreviates componentwise equality. An interpretation chooses a constant hidden vector $x\in\{0,1\}^n$ and two response functions $t(U)$ and $s(V)$. The responses must differ precisely when both inputs equal $x$, and agree otherwise. Each response can access only its own input, as specified by the incomparable dependency sets $U$ and $V$. Every formula $\XG^-_n$ is false, and the dependency restrictions are essential: if $t$ could depend on both inputs, the choices $t(U,V)=[U=x]\land[V=x]$ and $s(V)=0$ would satisfy the formula.
For $n = 1$ and $X = 1$, complementing~$s$ turns $\XG^-_1$ into a well-known false DQBF from~\cite{Rab17}. 
The family $\XG^-_n$ generalises this example to $n$-bit inputs and a hidden vector. For each $i$, simultaneously complementing $g_i \colon x_i \mapsto \neg x_i, u_i \mapsto \neg u_i, v_i \mapsto \neg v_i$, while fixing all other variables, preserves the equality tests and hence the matrix. These transformations also respect the dependency sets and are therefore syntactic symmetries. Moreover, every pair $(i, g_i)$ is lex-permissible. Applying Theorem~\ref{thm:symm_breaker} therefore yields 
\begin{equation} \label{XOR_game_SB}
\psi_{\loc} := \bigwedge_{i=1}^{n} \neg x_i.
\end{equation}
Thus $\psi_{\loc}$ reduces the $2^n$ possible hidden vectors to a single representative $x = 0^n$. 

We also introduce a true variant of the XOR game, which is given by 
\[
\XG^{+}_n :=
\forall U\,\forall V\,
\exists X(\varnothing)\,
\exists t(U)\,\exists s(V).
\;
\left( (U = X) \vee (V = X) \right) \to 
\bigl((t\oplus s)\leftrightarrow
      ((U=X)\land(V=X))\bigr).
\]
Unlike $\XG^-_n$, this variant promises that at least one input equals the hidden vector. The parity condition is required only for input pairs satisfying this promise. The family is true for every $n \geq 1$: for any fixed $x \in \{0,1\}^n$, a model is given by $X=x, t(U) = [U=x], s(V) = \neg [V=x]$. The simultaneous bit-flips $g_i$ still preserve the formula, so the local construction again yields
$\psi_{\loc}:=(X=0^n)$. For each hidden vector, there are exactly two response pairs, related by simultaneous complementation of $t$ and $s$. The breaker therefore reduces the $2^{n+1}$ models to two. 

For $n = 3,\ldots,16$, we compare each instance of both $\XG^+_n$ and $\XG^-_n$ with and without the symmetry breaker $\psi_{\loc}$.
The results are presented in Table \ref{tab:xor-game-results}.

To additionally evaluate local symmetry breaking on existential variables with nonempty dependency sets, we consider a third alternative of the XOR-game, given by
\[
\mathrm{DXG}^{-}_{n}
:=
\forall U\,\forall V\,\forall Z\;
\exists X(U)\,\exists s(V)\,\exists t(U\cup Z).
\bigl((t\oplus s)\leftrightarrow((Z=X)\land(U=V))\bigr),
\]
where \(U,V,Z,X\) are \(n\)-bit tuples and \(n \geq 1\). Here the hidden vector is chosen by a Skolem function \(X(U)\). Every member of this family is false: fixing \(u\) and any \(v\neq u\) forces \(t(u,z)=s(v)\) for every \(z\), whereas setting \(v=u\) requires \(t(u,z)\oplus s(u)=[z=X(u)]\), which is nonconstant in \(z\). For each \(i\), simultaneously complementing \(x_i\) and \(z_i\), while fixing all other variables, preserves the matrix and respects the dependency sets. Moreover, each pair \((i,g_i)\) is lex-permissible: \(g_i\) fixes \(U\), as well as \(s\) and its incomparable dependency set \(V\). Theorem~\ref{thm:symm_breaker} therefore again yields the local breaker  $\psi_{\mathrm{loc}} := \bigwedge_{i=1}^{n}\neg x_i$. Note that in this family the breaker enforces \(X(u)=0^n\) for every \(u\), fixing all \(n2^n\) truth-table entries of the hidden-vector function using only \(n\) unit clauses. This benchmark thus complements \(\mathrm{XG}^{-}_{n}\) and \(\mathrm{XG}^{+}_{n}\) by exercising local normalisation of an entire Skolem function. For $n = 3,\ldots,16$, we compare each instance of $\DXG^-_n$ with and without the symmetry breaker $\psi_{\loc}$. The results are presented in Table \ref{tab:dxor-game-results}.

\begin{table}[t]
\caption{Wall-clock runtimes in seconds for the false family $\XG^-_n$ (left) and the true family $\XG^+_n$ (right) without and with local lex-leader symmetry breaking. 
$\timeout$ denotes a $3600$\,s timeout.}
\label{tab:xor-game-results}
\centering
\small
\setlength{\tabcolsep}{0pt}
% Centred header cell for a column followed by a gap of width #1.
\newcommand*{\colhd}[2]{\multicolumn{1}{c@{\hspace{#1}}}{#2}}
\resizebox{\linewidth}{!}{%
\begin{tabular}{@{}r@{\hspace{7pt}}%
  r@{\hspace{3pt}}r@{\hspace{9pt}}r@{\hspace{3pt}}r@{\hspace{9pt}}%
  r@{\hspace{3pt}}r@{\hspace{9pt}}r@{\hspace{3pt}}r@{\hspace{18pt}}%
  r@{\hspace{3pt}}r@{\hspace{9pt}}r@{\hspace{3pt}}r@{\hspace{9pt}}%
  r@{\hspace{3pt}}r@{\hspace{9pt}}r@{\hspace{3pt}}r@{}}
\toprule
& \multicolumn{8}{c@{\hspace{18pt}}}{$\XG^-_n$ (false)}
& \multicolumn{8}{c@{}}{$\XG^+_n$ (true)} \\
\cmidrule(r{18pt}){2-9}\cmidrule{10-17}
& \multicolumn{2}{c@{\hspace{9pt}}}{\textsc{Pedant}} & \multicolumn{2}{c@{\hspace{9pt}}}{\textsc{P3}}
& \multicolumn{2}{c@{\hspace{9pt}}}{\textsc{HQS}} & \multicolumn{2}{c@{\hspace{18pt}}}{\textsc{DQBDD}}
& \multicolumn{2}{c@{\hspace{9pt}}}{\textsc{Pedant}} & \multicolumn{2}{c@{\hspace{9pt}}}{\textsc{P3}}
& \multicolumn{2}{c@{\hspace{9pt}}}{\textsc{HQS}} & \multicolumn{2}{c@{}}{\textsc{DQBDD}} \\
\cmidrule(r{9pt}){2-3}\cmidrule(r{9pt}){4-5}\cmidrule(r{9pt}){6-7}\cmidrule(r{18pt}){8-9}
\cmidrule(r{9pt}){10-11}\cmidrule(r{9pt}){12-13}\cmidrule(r{9pt}){14-15}\cmidrule{16-17}
& \multicolumn{2}{c@{\hspace{9pt}}}{Breaker} & \multicolumn{2}{c@{\hspace{9pt}}}{Breaker}
& \multicolumn{2}{c@{\hspace{9pt}}}{Breaker} & \multicolumn{2}{c@{\hspace{18pt}}}{Breaker}
& \multicolumn{2}{c@{\hspace{9pt}}}{Breaker} & \multicolumn{2}{c@{\hspace{9pt}}}{Breaker}
& \multicolumn{2}{c@{\hspace{9pt}}}{Breaker} & \multicolumn{2}{c@{}}{Breaker} \\
\colhd{7pt}{$n$}
& \colhd{3pt}{none} & \colhd{9pt}{local} & \colhd{3pt}{none} & \colhd{9pt}{local}
& \colhd{3pt}{none} & \colhd{9pt}{local} & \colhd{3pt}{none} & \colhd{18pt}{local}
& \colhd{3pt}{none} & \colhd{9pt}{local} & \colhd{3pt}{none} & \colhd{9pt}{local}
& \colhd{3pt}{none} & \colhd{9pt}{local} & \colhd{3pt}{none} & \colhd{0pt}{local} \\
\midrule
 3 & 0.01 & 0.01 & 0.02 & 0.02 & 0.02 & 0.01 & 0.01 & 0.01
   & 0.01 & 0.01 & 0.02 & 0.02 & 0.01 & 0.01 & 0.01 & 0.01 \\
 4 & 0.03 & 0.01 & 0.03 & 0.02 & 0.06 & 0.01 & 0.01 & 0.01
   & 0.02 & 0.01 & 0.02 & 0.02 & 0.04 & 0.01 & 0.01 & 0.01 \\
 5 & 0.08 & 0.01 & 0.05 & 0.02 & 0.07 & 0.01 & 0.01 & 0.01
   & 0.17 & 0.01 & 0.02 & 0.02 & 0.10 & 0.01 & 0.02 & 0.01 \\
 6 & 0.22 & 0.01 & 0.11 & 0.02 & 0.16 & 0.01 & 0.07 & 0.01
   & 0.62 & 0.01 & 0.02 & 0.02 & 1.32 & 0.01 & 0.08 & 0.01 \\
 7 & 0.65 & 0.01 & 0.27 & 0.03 & 0.27 & 0.01 & 0.34 & 0.01
   & 7.15 & 0.01 & 0.03 & 0.02 & 5.23 & 0.01 & 0.41 & 0.01 \\
 8 & 2.46 & 0.01 & 1.01 & 0.02 & 0.86 & 0.01 & 4.51 & 0.01
   & 7.45 & 0.01 & 0.04 & 0.04 & 46.08 & 0.01 & 4.77 & 0.01 \\
 9 & 8.93 & 0.01 & 4.51 & 0.07 & 2.52 & 0.01 & 35.27 & 0.25
   & 18.13 & 0.01 & 0.08 & 0.18 & \timeout & 0.01 & 10.76 & 0.03 \\
10 & 30.99 & 0.01 & 20.30 & 0.07 & 16.17 & 0.01 & 481.02 & 1.10
   & 260.48 & 0.01 & \timeout & 0.52 & \timeout & 0.01 & 662.51 & 1.11 \\
11 & 135.16 & 0.01 & 88.84 & 0.07 & 62.11 & 0.01 & 43.89 & 5.01
   & 59.24 & 0.01 & \timeout & 0.52 & \timeout & 0.01 & 294.15 & 5.12 \\
12 & 477.92 & 0.01 & 613.23 & 0.07 & 229.94 & 0.01 & 1153.01 & 19.56
   & 35.65 & 0.01 & \timeout & 0.53 & \timeout & 0.01 & 768.83 & 20.72 \\
13 & 2040.70 & 0.01 & 1261.45 & 0.08 & 1088.03 & 0.01 & 384.46 & 62.97
   & \timeout & 0.01 & \timeout & 0.54 & \timeout & 0.01 & 1755.31 & 65.53 \\
14 & \timeout & 0.01 & \timeout & 0.02 & \timeout & 0.01 & \timeout & 274.30
   & 2228.68 & 0.01 & \timeout & 0.18 & \timeout & 0.01 & \timeout & 289.32 \\
15 & \timeout & 0.01 & \timeout & 0.08 & \timeout & 0.01 & \timeout & 1443.18
   & 3080.04 & 0.01 & \timeout & 0.55 & \timeout & 0.01 & \timeout & 1370.54 \\
16 & \timeout & 0.01 & \timeout & 0.08 & \timeout & 0.01 & \timeout & \timeout
   & \timeout & 0.01 & \timeout & 0.56 & \timeout & 0.01 & \timeout & \timeout \\
\bottomrule
\end{tabular}%
}
\end{table}

\begin{table}[t]
\caption{Wall-clock runtimes in seconds for the false family $\DXG^-_n$ without and with local lex-leader symmetry breaking.
$\timeout$ denotes a $3600$\,s timeout.}
\label{tab:dxor-game-results}
\centering
\small
\begin{tabular}{@{}r rr rr rr rr@{}}
\toprule
& \multicolumn{8}{c}{$\DXG^-_n$ (false)} \\
\cmidrule(l){2-9}
& \multicolumn{2}{c}{\textsc{Pedant}}
& \multicolumn{2}{c}{\textsc{P3}}
& \multicolumn{2}{c}{\textsc{HQS}}
& \multicolumn{2}{c@{}}{\textsc{DQBDD}} \\
\cmidrule(lr){2-3}\cmidrule(lr){4-5}\cmidrule(lr){6-7}\cmidrule(l){8-9}
& \multicolumn{2}{c}{Breaker}
& \multicolumn{2}{c}{Breaker}
& \multicolumn{2}{c}{Breaker}
& \multicolumn{2}{c@{}}{Breaker} \\
\multicolumn{1}{@{}c}{$n$}
& \multicolumn{1}{c}{none} & \multicolumn{1}{c}{local}
& \multicolumn{1}{c}{none} & \multicolumn{1}{c}{local}
& \multicolumn{1}{c}{none} & \multicolumn{1}{c}{local}
& \multicolumn{1}{c}{none} & \multicolumn{1}{c@{}}{local} \\
\midrule
 3 & 0.19 & 0.01 & 0.29 & 0.02 & 0.14 & 0.01 & 0.13 & 0.02 \\
 4 & 0.22 & 0.01 & 0.30 & 0.05 & 0.25 & 0.01 & 0.14 & 0.02 \\
 5 & 0.97 & 0.02 & 0.41 & 0.03 & 0.30 & 0.06 & 0.14 & 0.01 \\
 6 & 0.81 & 0.02 & 1.45 & 0.04 & 0.37 & 0.08 & 0.17 & 0.03 \\
 7 & 2.91 & 0.02 & 0.52 & 0.05 & 0.49 & 0.72 & 0.27 & 0.16 \\
 8 & 6.41 & 0.01 & 4.09 & 0.07 & 1.60 & 4.11 & 0.62 & 0.60 \\
 9 & 36.40 & 0.02 & 12.74 & 0.10 & 7.45 & 24.12 & 8.40 & 4.42 \\
10 & 104.03 & 0.01 & 139.76 & 0.08 & 18.10 & 122.17 & 24.49 & 16.34 \\
11 & 1608.32 & 0.02 & 158.07 & 0.08 & 133.23 & 655.08 & 65.91 & 61.61 \\
12 & \timeout & 0.02 & 745.62 & 0.09 & 248.98 & \timeout & 142.16 & 107.12 \\
13 & \timeout & 0.02 & 1689.80 & 0.09 & 848.18 & \timeout & 337.21 & 235.43 \\
14 & \timeout & 0.02 & \timeout & 0.09 & \timeout & \timeout & 757.70 & 724.71 \\
15 & \timeout & 0.02 & \timeout & 0.09 & \timeout & \timeout & 2818.30 & 1071.92 \\
16 & \timeout & 0.02 & \timeout & 0.11 & \timeout & \timeout & \timeout & \timeout \\
\bottomrule
\end{tabular}
\end{table}

\paragraph{Succinct pigeonhole principle.} For the succinct pigeonhole principle, as introduced in Section \ref{sect:succinct}, we have the local symmetry breaking predicate $\psi_{\loc} := (H = 0^n)$ (see Example \ref{ex:dphp_symmetries}) and the specialised symmetry breaking predicate $\psi_{\sort} := (X <_{\lex} Y) \to (F(X) \leq_{\lex} F'(Y))$ of Section \ref{sec:canonical_breakers}. These can be conjoined, leading to the predicate $\psi_{\DPHP} := \psi_{\loc} \wedge \psi_{\sort}$. Table \ref{tab:dphp-runtimes} shows the impact of local symmetry breaking and local plus specialised symmetry breaking for $\DPHP_n$, for $n = 2, \ldots, 10$.

\subsection{Results}

The impact of symmetry breaking varies substantially between solvers and benchmark families. 
On $\XG^-_n$ and $\XG^+_n$, local lex-leader symmetry breaking is highly effective for all four solvers (Table~\ref{tab:xor-game-results}).
Without symmetry breaking, the runtimes of all solvers increase rapidly with $n$, although not always monotonically, and no solver solves $\XG^-_n$ for $n \geq 14$ or $\XG^+_{16}$ within the time limit. 
With the breaker $X = 0^n$, \textsc{Pedant}, \textsc{P3}, and \textsc{HQS} solve every instance of both families in less than one second. 
\textsc{DQBDD} also benefits substantially and solves the instances of both families for $n = 14, 15$ only with the breaker. Its runtime, however, still grows rapidly with $n$, and it times out on both families for $n = 16$.

Without the breaker, the runtimes of some solvers grow very non-monotonically with~$n$, although the instances grow uniformly ($5n+4$ variables and $10n+8$ clauses).
\textsc{DQBDD} needs $481.02$\,s for $\XG^-_{10}$ but only $43.89$\,s for $\XG^-_{11}$ and behaves similarly on $\XG^+_n$.
\textsc{Pedant} times out on $\XG^+_{13}$ but solves $\XG^+_{14}$ and $\XG^+_{15}$. 
For \textsc{Pedant}, the runtime largely follows the number of conflicts in its counterexample-guided synthesis, which varies irregularly with~$n$ on $\XG^+_n$.
We do not have an explanation for this behaviour.

On $\DXG^-_n$, the effect of the breaker depends strongly on the solver (Table~\ref{tab:dxor-game-results}).
Without symmetry breaking, \textsc{Pedant} times out for $n \geq 12$, and \textsc{P3} and \textsc{HQS} time out for $n \geq 14$.
With the breaker, which fixes the hidden-vector function to $X(U) = 0^n$, \textsc{Pedant} and \textsc{P3} solve every instance in at most $0.11$\,s.
\textsc{DQBDD} benefits only slightly: it solves the instances up to $n = 15$ in both configurations, and on the largest solved instance, \(\mathrm{DXG}^{-}_{15}\), the breaker reduces its runtime by a factor of about $2.6$.
In contrast, the breaker slows down \textsc{HQS} for $n \geq 7$. On $\DXG^-_{11}$, its runtime increases from $133.23$\,s to $655.08$\,s, and with the breaker it no longer solves $\DXG^-_{12}$ and $\DXG^-_{13}$ within the time limit.
In both configurations, preprocessing takes negligible time and \textsc{HQS} expands the same $n$ universal variables to transform the formula into a QBF. We do not have an explanation for this slowdown.
This is the only family in our experiments on which symmetry breaking substantially increases the runtime of a solver, which again shows that its practical benefit depends on the solving technique.

For $\DPHP_n$, local symmetry breaking alone has only a moderate effect (Table~\ref{tab:dphp-runtimes}). It reduces the runtimes of \textsc{Pedant} and \textsc{P3} on $\DPHP_4$ by factors of about $3$ and $4.5$, respectively, but without the sorting predicate, no solver solves any instance with $n \geq 5$. Adding the sorting predicate changes this picture considerably. On $\DPHP_4$, it reduces the runtime of \textsc{Pedant} from $88.10$\,s with local symmetry breaking alone to $0.27$\,s, and both \textsc{Pedant} and \textsc{P3} solve all instances up to $n = 9$. \textsc{HQS} and \textsc{DQBDD} benefit to a lesser extent and solve the instances up to $n = 5$ and $n = 4$, respectively, compared to $n = 3$ without the sorting predicate. These differences suggest that the practical benefits depend on the interaction between the symmetry-breaking predicate, the encoding, and the solver's search procedure. Isolating these effects requires a more comprehensive experimental study. Moreover, the effectiveness of the sorting predicate illustrates the practical relevance of problem-specific canonicality predicates and generalised symmetries. %where generalised symmetries realise the arbitrary pigeon permutations used to justify this normalisation.

\begin{table}[t]
\caption{Wall-clock runtimes for $\DPHP_n$ without symmetry breaking, with local lex-leader symmetry breaking, and with local and specialised symmetry breaking. $\timeout$ denotes a $3600$\,s timeout.}
\label{tab:dphp-runtimes}
\centering
\small
\setlength{\tabcolsep}{3.5pt}
\resizebox{\linewidth}{!}{%
\begin{tabular}{@{}r
  rrr @{\hspace{3.0em}}
  rrr @{\hspace{3.0em}}
  rrr @{\hspace{3.0em}}
  rrr@{}}
%\begin{tabular}{@{}r rrr rrr rrr rrr@{}}
\toprule
& \multicolumn{12}{c}{$\DPHP_n$ - Runtime (s)} \\
\cmidrule(l){2-13}
& \multicolumn{3}{c}{\textsc{Pedant}}
& \multicolumn{3}{c}{\textsc{P3}}
& \multicolumn{3}{c}{\textsc{HQS}}
& \multicolumn{3}{c}{\textsc{DQBDD}} \\
\cmidrule(lr){2-4}\cmidrule(lr){5-7}\cmidrule(lr){8-10}\cmidrule(l){11-13}
& \multicolumn{3}{c}{Breaker}
& \multicolumn{3}{c}{Breaker}
& \multicolumn{3}{c}{Breaker}
& \multicolumn{3}{c}{Breaker} \\
\multicolumn{1}{@{}c}{$n$}
& \multicolumn{1}{c}{none} & \multicolumn{1}{c}{local} & \multicolumn{1}{c}{\shortstack{local+\\spec.}}
& \multicolumn{1}{c}{none} & \multicolumn{1}{c}{local} & \multicolumn{1}{c}{\shortstack{local+\\spec.}}
& \multicolumn{1}{c}{none} & \multicolumn{1}{c}{local} & \multicolumn{1}{c}{\shortstack{local+\\spec.}}
& \multicolumn{1}{c}{none} & \multicolumn{1}{c}{local} & \multicolumn{1}{c@{}}{\shortstack{local+\\spec.}} \\
\midrule
 2 & 0.02 & 0.01 & 0.01 & 0.02 & 0.06 & 0.07 & 0.01 & 0.01 & 0.01 & 0.01 & 0.01 & 0.02 \\
 3 & 0.18 & 0.11 & 0.05 & 0.15 & 0.11 & 0.08 & 0.01 & 0.01 & 0.01 & 0.02 & 0.03 & 0.03 \\
 4 & 266.57 & 88.10 & 0.27 & 452.42 & 100.65 & 0.19 & \timeout & \timeout & 3.47 & \timeout & \timeout & 367.27 \\
 5 & \timeout & \timeout & 1.31 & \timeout & \timeout & 0.61 & \timeout & \timeout & 25.69 & \timeout & \timeout & \timeout \\
 6 & \timeout & \timeout & 7.30 & \timeout & \timeout & 3.99 & \timeout & \timeout & \timeout & \timeout & \timeout & \timeout \\
 7 & \timeout & \timeout & 47.27 & \timeout & \timeout & 23.11 & \timeout & \timeout & \timeout & \timeout & \timeout & \timeout \\
 8 & \timeout & \timeout & 346.84 & \timeout & \timeout & 174.09 & \timeout & \timeout & \timeout & \timeout & \timeout & \timeout \\
 9 & \timeout & \timeout & 3093.56 & \timeout & \timeout & 1345.86 & \timeout & \timeout & \timeout & \timeout & \timeout & \timeout \\
10 & \timeout & \timeout & \timeout & \timeout & \timeout & \timeout & \timeout & \timeout & \timeout & \timeout & \timeout & \timeout \\
\bottomrule
\end{tabular}%
}
\end{table}

\section{Conclusion}
We developed a framework for syntactic and semantic symmetries of DQBF that accounts for explicit quantifier dependencies, laying theoretical foundations for further research on DQBF symmetries. Within this framework, we established a local lex-leader construction that strictly generalises its QBF counterpart. For succinct function encodings, we complemented the local construction with a framework for canonisation functions that enables one to exploit the structure of the encoding to break additional symmetries left unbroken by the local lex-leader construction. We then used the succinct pigeonhole principle to illustrate how these approaches combine to eliminate additional symmetry.

Our proof-of-concept experiments show that
(i) the local lex-leader construction is highly effective on the hidden-vector XOR game,
a genuine DQBF family whose truth value relies on incomparable
dependency sets, including a variant in which it normalises an entire Skolem function; and
(ii) for the succinct pigeonhole principle, specialised canonicality predicates yield substantial additional speedups beyond local lex-leader symmetry breaking alone. However, as was already observed in \cite{ASM06, KS18}, the exact effect of symmetry breaking depends strongly on the solver mechanism. On $\DXG^-_n$, for instance, the same breaker that lets \textsc{Pedant} and \textsc{P3} solve every instance almost instantly slows down \textsc{HQS} considerably.

Future work includes automating local lex-leader symmetry breaking (Theorem~\ref{thm:symm_breaker}) and evaluating it on a broader range of application benchmarks, as well as deriving compact canonicality predicates for additional succinct encodings (Table~\ref{tab:compact_canonisers}) and automating their construction into a general preprocessing procedure. Another open question is to characterise syntactically the pairs $(j,g)$ for which $\psi_{j,g}$ is a conjunctive symmetry breaker (Remark~\ref{rem:open}).

\begin{acks}
This research was funded by the Austrian Science Fund (FWF) \href{https://doi.org/10.55776/COE12}{10.55776/COE12} and the State of Upper Austria (LIT AI Lab).
\end{acks}

\bibliographystyle{ACM-Reference-Format}
\bibliography{lib}

\clearpage

\appendix

\section{CNF encodings}
\label{app:cnf}

For CNF encodings, we use the full definitional Tseitin transformation, which introduces fresh existential variables defined by gate equivalences and equipped with suitable dependency sets. This preserves the soundness of symmetry breaking established for the original formula. Indeed, the matrix and the symmetry-breaking predicate can be encoded separately using disjoint auxiliary variables, and the resulting encodings can be conjoined, as we explain below.

\paragraph{Symmetry breaking and CNF conversion.}
A symmetry-breaking predicate can be justified for the original DQBF and encoded separately from its matrix. Let $\Phi=P.\varphi$ and let $\psi$ be a sound symmetry-breaking predicate for $\Phi$. For  $\gamma \in \{\varphi,\psi\}$, write
\[
E_\gamma := \Delta_\gamma \land C_\gamma
\]
for a full definitional CNF encoding, where $\Delta_\gamma$ defines the auxiliary gates and $C_\gamma$ requires the output to be true. Use disjoint sets of fresh auxiliary variables, with dependency declarations ensuring that every original interpretation $s$ extends uniquely to an interpretation $s^+$ satisfying $\Delta_\varphi \land \Delta_\psi$. For the extended prefix $P^+$, we then have
\[
[P^+.(E_\varphi\land E_\psi)]_{s^+}
=
[P.(\varphi\land\psi)]_s.
\]
Conversely, every model of the combined encoding restricts to a model of $P.(\varphi\land\psi)$. Consequently,
\[
  P.\varphi\text{ is true}
  \quad\Longleftrightarrow\quad
  P.(\varphi\land\psi)\text{ is true}
  \quad\Longleftrightarrow\quad
  P^+.(E_\varphi\land E_\psi)\text{ is true}.
\]
Thus the matrix and the breaker may be encoded independently and their encodings conjoined. Correctness follows from the original symmetry-breaking result and the extension and restriction properties of the encoding, without requiring a structural symmetry lift to the CNF.

\subsection{Local Lex-leader Predicates}
\label{app:lex-cnf}

Let $P=\forall X\,\exists y_1(D_1)\cdots\exists y_k(D_k)$ be a topologically ordered prefix, and let $\Lambda$ be a set of lex-permissible pairs for a syntactic symmetry group of $P.\varphi$. Put $\Gamma=\{g\mid(j,g)\in\Lambda\}$ and $J_g=\{j\mid(j,g)\in\Lambda\}$. We encode
\begin{equation}
 \psi_\Lambda=
 \bigwedge_{g\in\Gamma}\ \bigwedge_{j\in J_g}
 \left[
  \left(\bigwedge_{x\in D_j}(x\leftrightarrow g(x))
    \land\bigwedge_{i<j}(y_i\leftrightarrow g(y_i))\right)
  \to(y_j\to g(y_j))
 \right].
 \label{app:lex-cnf:predicate}
\end{equation}
For $\Lambda=\varnothing$, the encoding is $\top$ without auxiliaries. Otherwise, we use full definitions, so every original interpretation uniquely determines the auxiliary functions.

\paragraph{Sharing antecedent prefixes.}
For each fixed symmetry $g$, we rewrite the antecedents of the selected symmetry-breaking predicates $\psi_{j,g}$ as prefixes of a common sequence of comparisons. Their common initial conjunctions can then be encoded once and reused.

Set $E_j:=\bigcup_{i\le j}D_i$. List each universal variable when it first occurs in a dependency set, immediately before the corresponding existential variable:
\[
 D_1,\ y_1,\ D_2\setminus E_1,\ y_2,\ \ldots,\
 D_k\setminus E_{k-1},\ y_k.
\]
Choose an arbitrary fixed order within each universal block and denote the resulting sequence by $v_1,\ldots,v_m$. Then $m=|E_k|+k$, and $y_j$ occurs at position $d_j=|E_j|+j$. The variables preceding $y_j$ are precisely those in $E_j$ together with $y_1,\ldots,y_{j-1}$. Universals outside $E_k$ are omitted.

For each selected lex-permissible pair $(j,g)$, we may replace $D_j$ by $E_j$ in the antecedent of \eqref{app:lex-cnf:predicate} without changing the predicate. This adds the comparisons
\[
\bigwedge_{x\in E_j\setminus D_j}(x\leftrightarrow g(x)).
\]
If $\psi_{j,g}$ is a tautology, the implication remains a tautology after adding these requirements to its antecedent. Otherwise, every added comparison is itself a tautology. Indeed, any $x\in E_j\setminus D_j$ belongs to some $D_i$ with $i<j$. Since $x\notin D_j$, we have $D_i\not\subseteq D_j$. Topological ordering also excludes $D_j\subsetneq D_i$, whence $D_i\parallel D_j$, and lex-permissibility gives $g(x)=x$.

Consequently, $\psi_{j,g}$ is equivalent to
\begin{equation}
 \left(\bigwedge_{r<d_j}(v_r\leftrightarrow g(v_r))\right)
 \to(y_j\to g(y_j)).
 \label{app:lex-cnf:padded}
\end{equation}
For a fixed $g$, these antecedents differ only in the length of the prefix, enabling the shared encoding below.

\paragraph{Representing the images.}
For each required image $g(v)$, let $\lambda_g(v)$ be a literal representing it. If $g(v)$ is already a literal, set $\lambda_g(v)=g(v)$ without adding variables or clauses. Otherwise, we construct a full definitional (Tseitin) CNF encoding of a circuit computing $g(v)$. Specifically, choose an explicit circuit over the original variables, using binary conjunctions and negated wires, and introduce full Tseitin definitions for its gates. Subcircuits may be shared within each $g$. Write $\Delta_g$ for these defining clauses; a constant output is represented by a fresh variable and a unit clause.

Dependency sets are propagated through the circuits. For original variables and literals, put
\[
\delta(x)=\{x\}\quad(x\in X),\qquad
\delta(y_i)=D_i,\qquad \delta(\neg a)=\delta(a).
\]
A gate $w \leftrightarrow (a\land b)$ receives dependency set $\delta(w) = \delta(a) \cup \delta(b)$ and the three clauses
\[
 (\neg w\lor a)\land(\neg w\lor b)
 \land(w\lor\neg a\lor\neg b).
\]
Negation is represented by signed wires and introduces no auxiliary variables or clauses. In particular, $\delta(\neg a)=\delta(a)$. Constant gates have empty dependencies. The lex-leader construction below is the same for literal and generalised symmetries.

\paragraph{Prefix definitions and clauses.}
For each $g$, introduce $z_{g,0},\ldots,z_{g,m-1}$ with definitions
\begin{equation}
 z_{g,0}\leftrightarrow\top,\qquad
 z_{g,r}\leftrightarrow
 \bigl(z_{g,r-1}\land(v_r\leftrightarrow\lambda_g(v_r))\bigr)
 \quad(1\le r<m).
 \label{app:lex-cnf:definitions}
\end{equation}
For literals $p,a,b$, the following five-clause template encodes $z \leftrightarrow (p \land (a \leftrightarrow b))$ directly:
\begin{equation}
\begin{aligned}
\Step(z;p,a,b)={}&
 (\neg z\lor p)
 \land(\neg z\lor\neg a\lor b)
 \land(\neg z\lor a\lor\neg b)\\
 &{}\land(z\lor\neg p\lor a\lor b)
 \land(z\lor\neg p\lor\neg a\lor\neg b).
\end{aligned}
\label{app:lex-cnf:gate-clauses}
\end{equation}
Negation of a literal denotes its complement. The complete encoding is $C_\Lambda = \bigwedge_{g\in\Gamma}C_g$, with
\begin{equation}
\begin{aligned}
 C_g = \Delta_g\land z_{g,0}
 \land\bigwedge_{r=1}^{m-1}
   \operatorname{Step}(z_{g,r};z_{g,r-1},v_r,\lambda_g(v_r))
  \land\bigwedge_{j\in J_g}
   (\neg z_{g,d_j-1}\lor\neg y_j\lor\lambda_g(y_j)).
\end{aligned}
\label{app:lex-cnf:encoding}
\end{equation}
Auxiliaries for distinct $g$ are disjoint. The $\Step$ clauses define $z_{g,r}$ to indicate whether every variable among $v_1,\ldots,v_r$ has the same truth value as its image under $g$. Such comparisons may still occur in the antecedent of a later selected predicate. In contrast, the final conjunction ranges only over $j \in J_g$, so a lex-leader clause targeting $y_j$ is added only when $(j,g)$ is selected.

The dependency sets of the agreement variables are obtained by collecting the dependencies of their inputs:
\begin{equation}
 \delta(z_{g,0})=\varnothing,\qquad
 \delta(z_{g,r})=\delta(z_{g,r-1})\cup\delta(v_r)
                 \cup\delta(\lambda_g(v_r)).
 \label{app:lex-cnf:dependencies}
\end{equation}

Let $P^+$ extend $P$ by quantifying every auxiliary variable $a$ existentially with dependency set $\delta(a)$. Original dependencies are unchanged. These dependency sets are sufficient, though not necessarily minimal.

\begin{prop}[Correctness of the CNF encoding]
\label{app:lex-cnf:correctness}
Every interpretation $s\in\mathbb S(P)$ has a unique extension
$s^+\in\mathbb S(P^+)$ satisfying all auxiliary definitions, and
\begin{equation}
 [P^+.(\varphi\land C_\Lambda)]_{s^+}
 =[P.(\varphi\land\psi_\Lambda)]_s.
 \label{app:lex-cnf:interpretation-equivalence}
\end{equation}
Every model of the encoded formula restricts to a model of $P.(\varphi\land\psi_\Lambda)$.
\end{prop}
\begin{proof}
Fix $s\in\mathbb S(P)$. Evaluate the image circuits and then recursively determine the variables $z_{g,r}$ from their $\Step$ definitions, starting with $z_{g,0} = 1$. The declared dependencies permit these computations, and the full defining clauses force their values uniquely. Thus $s$ has a unique extension $s^+$ satisfying all auxiliary definitions.

For every universal assignment $\sigma$, the image definitions ensure that $\lambda_g(v)$ evaluates to $[g(v)]_{\sigma_s}$, and the $\Step$ definitions give
\[
 [z_{g,r}]_{\sigma_{s^+}}
 =
 \bigwedge_{\ell=1}^{r}
 [v_\ell\leftrightarrow g(v_\ell)]_{\sigma_s},
\]
with the empty conjunction true. Each final clause therefore evaluates to the padded predicate~\eqref{app:lex-cnf:padded}, which is equivalent to $\psi_{j,g}$. Since all auxiliary definitions hold and the original variables retain their values, conjoining these clauses and taking the conjunction over all universal assignments proves~\eqref{app:lex-cnf:interpretation-equivalence}.

Finally, any model $t$ of the encoded formula restricts to an interpretation $s\in\mathbb S(P)$. As $t$ satisfies all auxiliary definitions, uniqueness gives $t=s^+$. The established equality then implies that $s$ satisfies $P.(\varphi\land\psi_\Lambda)$.
\end{proof}

Together with the local symmetry-breaking theorem, this gives $P.\varphi$ true if and only if $P^+.(\varphi\land C_\Lambda)$ is true. If $\varphi$ is not already in CNF, its matrix must also undergo CNF conversion.

\paragraph{Size.}
Put $r:=|\Gamma|$. For each $g\in\Gamma$, the shared antecedent encoding and final lex-leader clauses introduce $m$ auxiliary variables and $5(m-1)+1+|J_g|$ clauses before simplification: five per $\Step$
definition, one unit clause for $z_{g,0}$, and one final clause per selected pair. If the image circuits contain $S$ gates in total, including constants, they introduce a further $S$ variables and at most $3S$ clauses. For literal symmetries, no image circuits are needed. Since every clause of $C_\Lambda$ has width at most four, its size, counting literal occurrences, is $O(S+rm+|\Lambda|)$. Moreover,
\[
  m=\left|\bigcup_iD_i\right|+|Y|
  \le |X|+|Y|,
  \qquad
  |\Lambda|=\sum_{g\in\Gamma}|J_g|\le r|Y|.
\]
Consequently, the symmetry-breaking matrix $C_\Lambda$ has size
\[
  O\bigl(r(|X|+|Y|)+S\bigr),
\]
as claimed in Theorem~\ref{thm:symm_breaker}. This bound is relative to the supplied circuit representations, no compact circuit representation of an arbitrary generalised symmetry is assumed.

\paragraph{Specialisation to the experiments.}
The local predicates used in our experiments simplify before CNF conversion. 
For the hidden-vector XOR families, it is $X=0^n$, giving $n$ unit clauses. For $\mathrm{DPHP}_n$, the local predicate is $H=0^n$, again giving $n$ unit clauses. These simplified local predicates require no auxiliary variables. The base encoding of $\mathrm{DPHP}_n$ and its additional sorting predicate are described below.

\subsection{The Succinct Pigeonhole Principle}
\label{app:dphp-cnf}
Let $n \geq 1$, let $X = (x_1,\ldots,x_n)$ and $Y = (y_1,\ldots,y_n)$ be universal tuples, and let $H = (h_1,\ldots,h_n)$, $F = (f_1,\ldots,f_n)$, and $F'=(f'_1,\ldots,f'_n)$ be existential tuples. Write
\begin{align}
  P_n&=\forall X\,\forall Y\,
       \exists H(\varnothing)\,\exists F(X)\,\exists F'(Y),\\
  \varphi_n&=((X=Y)\leftrightarrow(F=F'))
       \land(F\ne H)\land(F'\ne H).
  \label{app:dphp-cnf:formula}
\end{align}
Thus $\DPHP_n = P_n.\varphi_n$. All dependencies of a tuple are understood componentwise.

\paragraph{Clause templates.} For literals $z,a,b$, define
\begin{equation}
\begin{aligned}
  \Eq(z;a,b)={}&
  (\neg z\lor\neg a\lor b)\land(\neg z\lor a\lor\neg b)\\
  &{}\land(z\lor a\lor b)\land(z\lor\neg a\lor\neg b).
\end{aligned}
\label{app:dphp-cnf:eq-template}
\end{equation}
This encodes $z \leftrightarrow (a \leftrightarrow b)$; in particular, $\Eq(\neg z;a,b)$ encodes $z \leftrightarrow (a \oplus b)$. Negation of a literal denotes its complement. For a tuple $A = (a_1,\ldots,a_r)$, put
\begin{equation}
  \OAnd(z;A)=
  \bigwedge_{i=1}^{r}(\neg z\lor a_i)
  \land(z\lor\neg a_1\lor\cdots\lor\neg a_r).
  \label{app:dphp-cnf:and-template}
\end{equation}
This encodes $z \leftrightarrow \bigwedge_i a_i$ using $r+1$ clauses.

\paragraph{Base encoding.}
Introduce auxiliary variables $d_i,d'_i,a_i,b_i$ for $1\le i\le n$, and two further variables $A,B$. Their intended meanings are
\begin{align*}
  d_i&\leftrightarrow(f_i\oplus h_i),&
  d'_i&\leftrightarrow(f'_i\oplus h_i),\\
  a_i&\leftrightarrow(x_i\leftrightarrow y_i),&
  b_i&\leftrightarrow(f_i\leftrightarrow f'_i),\\
  A&\leftrightarrow\bigwedge_{i=1}^n a_i,&
  B&\leftrightarrow\bigwedge_{i=1}^n b_i.
\end{align*}
Writing $\mathbf a = (a_1,\ldots,a_n)$ and $\mathbf b = (b_1,\ldots,b_n)$, define
\begin{equation}
\begin{aligned}
 C_n^{\mathrm{base}}={}&
 \bigwedge_{i=1}^{n}\bigl(
   \Eq(\neg d_i;f_i,h_i)
   \land\Eq(\neg d'_i;f'_i,h_i)\bigr)\\
 &{}\land\bigwedge_{i=1}^{n}\bigl(
   \Eq(a_i;x_i,y_i)
   \land\Eq(b_i;f_i,f'_i)\bigr)\\
 &{}\land(d_1\lor\cdots\lor d_n)
   \land(d'_1\lor\cdots\lor d'_n)\\
 &{}\land\OAnd(A;\mathbf a)
   \land\OAnd(B;\mathbf b)
   \land(\neg A\lor B)\land(A\lor\neg B).
\end{aligned}
\label{app:dphp-cnf:base}
\end{equation}
The two long disjunctions express avoidance of $H$, while the final two clauses express equivalence of input equality and output equality. Extend $P_n$ with the following auxiliary dependencies:
\[
\begin{array}{c|c}
 \text{Auxiliary variable}&\text{Dependency set}\\ \hline
 d_i&X\\
 d'_i&Y\\
 a_i&\{x_i,y_i\}\\
 b_i&X\cup Y\\
 A,\ B&X\cup Y
\end{array}
\]
We use $X$ and $Y$ also for the sets of variables in the corresponding tuples. In particular, an output equality bit $b_i$ may depend on all of $X\cup Y$, since $f_i$ and $f'_i$ have those respective input sets.

\paragraph{Correctness and the local predicate.}
Every interpretation of $P_n$ uniquely determines the auxiliary functions through the displayed definitions, with the declared dependencies. Under this extension, \eqref{app:dphp-cnf:base} has the value of $\varphi_n$ for every universal assignment. Conversely, the defining clauses force these functions, establishing the extension and restriction property.

The local predicate $\psi_{\mathrm{loc}}=(H=0^n)$ is added as $n$ negative unit clauses:
\begin{equation}
  C_n^{\mathrm{loc}}=
  C_n^{\mathrm{base}}\land\bigwedge_{i=1}^{n}\neg h_i.
  \label{app:dphp-cnf:local}
\end{equation}
It introduces no variables and changes no dependency sets.

\paragraph{Output bit-flip symmetries of the base encoding.}
For each coordinate $i$, simultaneously complementing $h_i,f_i,f'_i$ and fixing every other variable, including all auxiliaries, is a structural literal symmetry of $C_n^{\mathrm{base}}$ with its dependency prefix. Each affected XOR or equality gate has both inputs complemented, which permutes its defining clauses. Also every remaining clause and dependency declaration is preserved. Thus the output bit-flips used to derive $H=0^n$ also admit literal lifts to this particular PCNF encoding.

\subsection{The Sorting Predicate for the Pigeonhole Principle}
\label{app:dphp-cnf:sorting}

We encode
\begin{equation}
  \psi_{\mathrm{sort}}=
  (X<_{\mathrm{lex}}Y)\to(F\le_{\mathrm{lex}}F'),
  \label{app:dphp-cnf:sorting-predicate}
\end{equation}
The construction reuses the equality bits $\mathbf a$ and $\mathbf b$ from the base encoding. Introduce $L_1,\ldots,L_n$ and $R_1,\ldots,R_n$ with meanings
\[
  L_i\leftrightarrow(X_{i:n}<_{\mathrm{lex}}Y_{i:n}),\qquad
  R_i\leftrightarrow(F'_{i:n}<_{\mathrm{lex}}F_{i:n}),
\]
where $U_{i:n}=(u_i,\ldots,u_n)$ denotes a suffix. The reversed output comparison allows \eqref{app:dphp-cnf:sorting-predicate} to be expressed by the single clause $\neg L_1\lor\neg R_1$.

\paragraph{A recursive comparison circuit.}
For two $n$-tuples $U,V$ and a tuple $E=(e_1,\ldots,e_n)$ of already defined equality bits $e_i\leftrightarrow(u_i\leftrightarrow v_i)$, define comparison bits $T=(t_1,\ldots,t_n)$ recursively by
\begin{align}
 t_n&\leftrightarrow(\neg u_n\land v_n),\nonumber\\
 t_i&\leftrightarrow
   \bigl((\neg u_i\land v_i)\lor(e_i\land t_{i+1})\bigr)
   \qquad(1\le i<n).
 \label{app:dphp-cnf:less-recursion}
\end{align}
Let $\mathcal L(U,V;E,T)$ be the conjunction of the three final-position clauses
\begin{equation}
  (\neg t_n\lor\neg u_n)\land(\neg t_n\lor v_n)
  \land(t_n\lor u_n\lor\neg v_n)
  \label{app:dphp-cnf:less-base}
\end{equation}
and, for each $1\le i<n$, the four clauses
\begin{equation}
\begin{aligned}
 &(\neg e_i\lor\neg t_i\lor t_{i+1})
   \land(\neg e_i\lor t_i\lor\neg t_{i+1})\\
 &{}\land(u_i\lor\neg v_i\lor t_i)
   \land(\neg u_i\lor v_i\lor\neg t_i).
\end{aligned}
\label{app:dphp-cnf:less-step}
\end{equation}
Using the equality definition of $e_i$, these clauses force $t_i=t_{i+1}$ when $u_i=v_i$, $t_i=1$ when $u_i=0$ and $v_i=1$, and $t_i=0$ when $u_i=1$ and $v_i=0$. Backward induction gives $t_i = [U_{i:n} <_{\mathrm{lex}} V_{i:n}]$. Thus, assuming $e_i\leftrightarrow(u_i\leftrightarrow v_i)$ for every
$i\in\{1,\ldots,n\}$, we have
\[
\mathcal L(U,V;E,T)
\quad\Longleftrightarrow\quad
\bigwedge_{i=1}^{n}
\left(
  t_i\leftrightarrow
  \bigl((u_i,\ldots,u_n)<_{\lex}(v_i,\ldots,v_n)\bigr)
\right).
\]

\paragraph{Combined encoding and dependencies.}
Put $L=(L_1,\ldots,L_n)$ and $R=(R_1,\ldots,R_n)$, and define
\begin{equation}
\begin{aligned}
 C_n^{\mathrm{sort}}={}&
   \mathcal L(X,Y;\mathbf a,L)
   \land\mathcal L(F',F;\mathbf b,R)
   \land(\neg L_1\lor\neg R_1),\\
 C_n^{\mathrm{loc+sort}}={}&
   C_n^{\mathrm{loc}}\land C_n^{\mathrm{sort}}.
\end{aligned}
\label{app:dphp-cnf:combined}
\end{equation}
The bits $\mathbf b$ can be reused for the reversed comparison because $f_i\leftrightarrow f'_i$ and $f'_i\leftrightarrow f_i$ are equivalent. The new dependencies are
\begin{equation}
  \delta(L_i)=\{x_j,y_j\mid i\le j\le n\},\qquad
  \delta(R_i)=X\cup Y.
  \label{app:dphp-cnf:sorting-dependencies}
\end{equation}
Although $R_i$ compares only an output suffix, its coordinates may depend on every input variable, so it receives the full set $X\cup Y$. All original dependencies and base clauses are retained.

Every interpretation of $P_n$ extends uniquely to satisfy the auxiliary definitions of the base encoding and the sorting comparisons. The local breaker introduces no additional variables. The final clause is then equivalent to \eqref{app:dphp-cnf:sorting-predicate} for every pair of inputs. This proves the extension and restriction property for $P_n.(\varphi_n\land\psi_{\mathrm{loc}}\land\psi_{\mathrm{sort}})$.

\paragraph{Exact sizes.}
The base encoding adds $4n+2$ auxiliary variables to the $5n$ original variables. Each comparison circuit adds $n$ variables and $4n-1$ clauses. Including the final binary clause, sorting adds $2n$ variables and $8n-1$ clauses. Table~\ref{app:dphp-cnf:size-table} gives the counts before simplification or solver preprocessing.

\begin{table}[!htbp]
\caption{Sizes of the $\DPHP$ encodings before simplification. All configurations have $2n$ universal variables, the remaining
variables are existential.}
\label{app:dphp-cnf:size-table}
\centering
\begin{tabular}{lrrr}
\hline
Configuration & Total variables & Auxiliary variables & Clauses\\
\hline
Base & $9n+2$ & $4n+2$ & $18n+6$\\
Local & $9n+2$ & $4n+2$ & $19n+6$\\
Local and sorting & $11n+2$ & $6n+2$ & $27n+5$\\
\hline
\end{tabular}

\end{table}

The matrix has $\Theta(n)$ literal occurrences in every configuration. Explicit dependency lists have $\Theta(n^2)$ entries, since the original output bits alone contribute $2n^2$ dependency entries. The local and sorting clauses are appended to the same base encoding.

\end{document}